\pdfoutput=1

\documentclass[11pt]{article}
\usepackage[letterpaper,margin=1in]{geometry} 
\usepackage{iftex}

\ifPDFTeX
\usepackage[utf8]{inputenc}
\usepackage[T1]{fontenc}
\fi

\usepackage{amsmath}
\usepackage{amssymb}
\usepackage{amsthm}
\usepackage{mathtools}
\usepackage{thmtools}
\usepackage{isomath}
\usepackage{algorithm}
\usepackage{thm-restate}
\usepackage{enumitem}
\usepackage{subcaption}

\newcounter{desccount}

\newcommand{\descref}[1]{\hyperref[#1]{#1}}

\usepackage{lineno}
\usepackage[disable]{todonotes}
\usepackage{changes}
\usepackage{comment}
\usepackage[most]{tcolorbox}

\ifPDFTeX
\usepackage{libertine}
\usepackage[libertine]{newtxmath} 
\usepackage[scaled=0.96]{zi4} 
\else
\usepackage{fontspec}
\usepackage{unicode-math}
\fi

\usepackage{microtype}
\usepackage[noend]{algpseudocode}

\usepackage{xcolor}

\usepackage[
    style=alphabetic,
    maxalphanames=4,
    minalphanames=4,
    maxcitenames=4,
    minnames=1,
    maxbibnames=20,
    backref=true,
    doi=true,
    url=true,
    backend=biber
]{biblatex}

\ifpdf
\usepackage[ocgcolorlinks]{hyperref} 
\else
\usepackage[hidelinks]{hyperref}
\fi

\usepackage[capitalize]{cleveref}
\usepackage{xparse}
\newcommand{\linecref}[1]{Line~\ref{#1}}
\newcommand{\linescref}[2]{Lines~\ref{#1} and~\ref{#2}}
\newcommand{\linesscref}[3]{Lines~\ref{#1},~\ref{#2}, and~\ref{#3}}

\usepackage{xspace}

\allowdisplaybreaks[1]

\makeatletter
\g@addto@macro\bfseries{\boldmath}
\makeatother

\makeatletter
\g@addto@macro\mdseries{\unboldmath}
\g@addto@macro\normalfont{\unboldmath}
\g@addto@macro\rmfamily{\unboldmath}
\g@addto@macro\upshape{\unboldmath}
\makeatother

\makeatletter
\g@addto@macro\bfseries{\boldmath}
\def\thmhead@plain#1#2#3{%
  \thmname{#1}\thmnumber{\@ifnotempty{#1}{ }\@upn{#2}}%
  \thmnote{ {\the\thm@notefont\unboldmath(#3)}}}
\let\thmhead\thmhead@plain
\makeatother

\DeclareCiteCommand{\citem}
    {}
    {\mkbibbrackets{\bibhyperref{\usebibmacro{postnote}}}}
    {\multicitedelim}
    {}

\renewcommand*{\multicitedelim}{\addcomma\space}

\AtEveryCitekey{\clearfield{note}}

\makeatletter
\AtBeginDocument{%
    \newlength{\temp@x}%
    \newlength{\temp@y}%
    \newlength{\temp@w}%
    \newlength{\temp@h}%
    \def\my@coords#1#2#3#4{%
      \setlength{\temp@x}{#1}%
      \setlength{\temp@y}{#2}%
      \setlength{\temp@w}{#3}%
      \setlength{\temp@h}{#4}%
      \adjustlengths{}%
      \my@pdfliteral{\strip@pt\temp@x\space\strip@pt\temp@y\space\strip@pt\temp@w\space\strip@pt\temp@h\space re}}%
    \ifpdf
      \typeout{In PDF mode}%
      \def\my@pdfliteral#1{\pdfliteral page{#1}}
      \def\adjustlengths{}%
    \fi
    \ifxetex
      \def\my@pdfliteral #1{}
      \def\adjustlengths{\setlength{\temp@h}{-\temp@h}\addtolength{\temp@y}{1in}\addtolength{\temp@x}{-1in}}%
    \fi%
    \def\Hy@colorlink#1{%
      \begingroup
        \ifHy@ocgcolorlinks
          \def\Hy@ocgcolor{#1}%
          \my@pdfliteral{q}%
          \my@pdfliteral{7 Tr}
        \else
          \HyColor@UseColor#1%
        \fi
    }%
    \def\Hy@endcolorlink{%
      \ifHy@ocgcolorlinks%
        \my@pdfliteral{/OC/OCPrint BDC}%
        \my@coords{0pt}{0pt}{\pdfpagewidth}{\pdfpageheight}%
        \my@pdfliteral{F}
        \my@pdfliteral{EMC/OC/OCView BDC}%
        \begingroup%
          \expandafter\HyColor@UseColor\Hy@ocgcolor%
          \my@coords{0pt}{0pt}{\pdfpagewidth}{\pdfpageheight}%
          \my@pdfliteral{F}
        \endgroup%
        \my@pdfliteral{EMC}%
        \my@pdfliteral{0 Tr}
        \my@pdfliteral{Q}%
      \fi
      \endgroup
    }%
}
\makeatother

\newcommand{\antonis}[1]{\todo[linecolor=orange!50!black,backgroundcolor=orange!25,bordercolor=orange!50!black]{\scriptsize \textbf{AS:} #1}}

\colorlet{DarkRed}{red!50!black}
\colorlet{DarkGreen}{green!50!black}
\colorlet{DarkBlue}{blue!50!black}

\hypersetup{
    hypertexnames=false,
	linkcolor = DarkRed,
	citecolor = DarkGreen,
	urlcolor = DarkBlue,
	bookmarksnumbered = true,
	linktocpage = true
}

\declaretheorem[numberwithin=section]{theorem}
\declaretheorem[numberlike=theorem]{lemma}

\declaretheorem[numberlike=theorem]{corollary}
\declaretheorem[numberlike=theorem]{definition}
\declaretheorem[numberlike=theorem]{claim}
\declaretheorem[numberlike=theorem]{observation}

\newcommand{\dist}{\operatorname{dist}}
\newcommand{\dd}{\operatorname{d}}
\newcommand{\cost}{\operatorname{cost}}
\newcommand{\ball}{\operatorname{Ball}}
\newcommand{\Retain}{\operatorname{Retain}}
\newcommand{\wt}{\operatorname{wt}}
\newcommand{\old}{\text{(old)}}
\newcommand{\OPT}{\text{OPT}}
\newcommand{\myhs}{\hspace{0.05em}}

\DeclarePairedDelimiter\ceil{\lceil}{\rceil}

\DeclareMathOperator*{\argmin}{argmin}

\DeclareMathOperator*{\Prob}{\ensuremath{\textnormal{Pr}}}
\renewcommand{\Pr}{\Prob}

\newcommand{\prob}[1]{\Pr\paren{#1}}
\newcommand{\card}[1]{|#1|}
\newcommand{\paren}[1]{\left ( #1 \right )}

\newcommand{\set}[1]{\ensuremath{\left\{ #1 \right\}}}

\title{Incremental \((k, z)\)-Clustering on Graphs}
\author{Emilio Cruciani\thanks{European University of Rome, Italy. This research was partially conducted when the author was a postdoc at the University of Salzburg.}
  \and 
  Sebastian Forster\thanks{Department of Computer Science, University of Salzburg, Salzburg, Austria. This work is supported by the Austrian Science Fund (FWF): P 32863-N. This project has received funding from the European Research Council (ERC) under the European Union's Horizon 2020 research and innovation programme (grant agreement No 947702).}
  \and
  Antonis Skarlatos\thanks{Department of Computer Science, University of Warwick, Coventry, England. This research was partially conducted when the author was a PhD student at the University of Salzburg. Antonis Skarlatos is funded by the European Union (ERC grant, DYNALP, 101170133). Views and opinions expressed are however those of the author(s) only and do not necessarily reflect those of the European Union or the European Research Council Executive Agency. Neither the European Union nor the granting authority can be held responsible for them.}}
\date{}

\begin{document}
    \maketitle
    \begin{abstract}
        Given a weighted undirected graph, a number of clusters $k$, and an exponent $z$, the goal in the $(k, z)$-clustering problem on graphs is to select $k$ vertices as centers that minimize the sum of the distances raised to the power $z$ of each vertex to its closest center. This problem includes the well-known $k$-median ($z = 1)$ and $k$-means ($z = 2$) clustering problems. In the dynamic setting, the graph is subject to adversarial edge updates, and the goal is to maintain explicitly an exact $(k, z)$-clustering solution in the induced shortest-path metric.

Prior works by Bhattacharya, Costa, Garg, Lattanzi, and Parotsidis [FOCS 2024]
and by Bhattacharya, Costa, and Farokhnejad [STOC 2025] consider the dynamic $(k, z)$-clustering problem for point sets in metric spaces. These algorithms support adversarial point insertions and deletions under a model with access to pairwise distances. This model differs significantly from the dynamic graph setting, where no oracle access is given to pairwise distances and a single edge update can affect many distances---making these approaches inefficient when applied to graphs. 
While efficient dynamic $k$-center approximation algorithms on graphs exist [Cruciani, Forster, Goranci, Nazari, and Skarlatos, SODA 2024], to the best of our knowledge, no prior work provides similar results for the dynamic $(k,z)$-clustering problem.

As the main result of this paper, we develop a randomized incremental $(k, z)$-clustering algorithm that maintains with high probability a constant-factor approximation in a graph undergoing edge insertions with a total update time of $\tilde O(k \myhs m^{1+o(1)}+ k^{1+\frac{1}{\lambda}} \myhs m)$, where $\lambda \geq 1$ is an arbitrary fixed constant. 
Our incremental algorithm also achieves an amortized update time of $\tilde O(k \myhs n^{o(1)}+ k^{1+\frac{1}{\lambda}})$ and consists of two stages. In the first stage, we maintain a constant-factor bicriteria approximate solution of size~$\tilde{O}(k)$ with a total update time of $m^{1+o(1)}$ (independent of the parameter $k$) over all adversarial edge insertions. This first stage is an intricate adaptation of the bicriteria approximation algorithm by Mettu and Plaxton [Machine Learning 2004] to incremental graphs. One of our key technical results is that the radii in their algorithm can be assumed to be non-decreasing while the approximation ratio remains constant---a property that may be of independent interest.

In the second stage, we maintain a constant-factor approximate $(k,z)$-clustering solution on a dynamic weighted instance induced by the bicriteria approximate solution. 
For this subproblem, we employ a dynamic spanner algorithm together with a static $(k,z)$-clustering algorithm.
    \end{abstract}

    \thispagestyle{empty}
    \newpage
    \tableofcontents
    \thispagestyle{empty}
    \newpage
    
    \setcounter{page}{1}
    \section{Introduction}
Clustering is a fundamental and well-studied problem in computer science with numerous applications across various domains~\cite{JainMF99,hansen1997cluster,ShiM00,jain2010data,coates2012learning}.
Conceptually, clustering involves partitioning  data points into groups such that points within the same group are more similar to each other than to those in different groups, according to some measure of similarity. Clustering algorithms optimize a given objective function, and for $k$-clustering problems
the goal is to output a set of $k$ representative points that minimize a $k$-clustering objective.

While the graph clustering toolkit for network analysis tasks (e.g., community detection) encompasses a wide range of clustering objectives suited to different scientific communities~\cite{Newman04,schaeffer2007graph,fortunato2010community,Aggarwal2010}, researchers have in particular performed experimental studies with $k$-clustering objectives on graphs~\cite{SchenkerLBK04,RattiganMJ07}.
It has been observed in these experiments that (a) high computational cost and (b) the sensitivity of distances to single edge updates pose a challenge when employing $k$-clustering in graphs~\cite{RattiganMJ07} (see also~\cite[Section~2.3]{Aggarwal2010}), which motivates further algorithm design in this direction.

Among the most widely studied $k$-clustering objectives are
\emph{$k$-median} and \emph{$k$-means}. The $k$-median objective
is the sum of distances from each point to its closest center, 
while the $k$-means objective is the sum of squared distances from each point to its closest center.
A generalization of both is the $(k, z)$-clustering objective, which
is the sum of distances raised to the power $z$ from each point to its closest center. 
Hence given a metric space $(\mathcal{X}, \dist)$, a set of points $P \subseteq \mathcal{X}$, and a positive integer $k \leq |P|$, the goal in the $(k, z)$-clustering problem is to output a subset $S \subseteq P$ of at most $k$ points, referred to as \emph{centers}, such that the value $\sum_{p \in P} \dist(p, S)^z$ is minimized. Observe that the $k$-median clustering problem corresponds to $z=1$, and the $k$-means clustering problem corresponds to $z=2$.

The $k$-median clustering problem is known to be NP-hard in arbitrary metric spaces \cite{kariv1979algorithmic}.
The first approximation algorithms leveraged tree embeddings and polynomial-time algorithms for the tree setting~\cite{tamir1996pn2}, achieving an approximation ratio of $O(\log n \log\log n)$ \cite{bartal1996probabilistic,bartal1998approximating}; later refined to~$O(\log k \log\log k)$ \cite{charikar1998rounding}. The first constant-factor approximation was established in \cite{charikar1999constant} using LP rounding and was later improved to $3 + \epsilon$ through local search \cite{arya2004locals, gupta2008simpler, cohen2022improved}. Subsequently, approximation ratios below $3$ were achieved~\cite{LiS13, byrka2015kmedian,cohen2022improved},
with the best-known approximation ratio currently standing at~$2+\epsilon$~\cite{cohen2025kmedian}. Meanwhile, the best known hardness result rules out an approximation factor less than~$ 1+2/e \approx 1.74 $~\cite{jain2002facility}.
Many of the aforementioned algorithmic approaches carry over to $k$-means and $(k, z)$-clustering with worse approximation factors.
For the $k$-means clustering problem, the best known approximation ratio is~$5+\epsilon$~\cite{byrka2026kclusteringiterativerandomizedrounding} and the best known hardness result rules out an approximation factor less than $ 1+8/e \approx 3.94 $~\cite{jain2002facility,ahmadian2020better}.

\paragraph{Efficient $k$-clustering algorithms.}
Fast algorithms have been proposed for the $k$-median clustering problem in arbitrary metric spaces with oracle access to the distances between points. Jain and Vazirani~\cite{JainV01} developed a deterministic $6$-approximation algorithm that runs in $\tilde{O}(n^2)$ time.
Subsequently, Mettu and Plaxton presented another approach for obtaining a constant-factor approximation in~$\tilde{O}(n^2)$ time~\cite{MettuP03online}.
Later, they presented a randomized $O(1)$-approximation algorithm with a more efficient running time of $\tilde{O}(nk)$~\cite{mettu2004optimal}. As noted by Huang and Vishnoi~\cite{HuangV20}, the algorithm in~\cite{mettu2004optimal} can be easily generalized to the $(k, z)$-clustering problem.

\paragraph{Graph setting.}
In the \emph{graph setting}, the $(k, z)$-clustering problem is defined as before, where the metric space is induced by the shortest-path distances in a weighted undirected graph with $n$ vertices, $m$ edges, and maximum edge weight $W$. Here, no oracle access is available, and distance computations must be performed explicitly by the algorithm if needed. As a consequence, although algorithms designed for point sets can be applied in the graph setting, this approach leads to inefficient algorithms. Thorup~\cite{thorup2005quick} addressed this issue for the $k$-median clustering problem by developing a randomized $(12+o(1))$-approximation algorithm that runs in $\tilde{O}(m)$ time.\footnote{The notation $\tilde{O}(\cdot)$ hides polylogarithmic factors in $nW$.}
As noted in~\cite{latour2025fastersimplergreedyalgorithm}, it seems plausible that the approach of Thorup~\cite{thorup2005quick} generalizes to $ (k, z) $-clustering using recent insights on the primal-dual method used as a subroutine.
Very recently, alternative  near-linear and almost-linear time constant-factor approximation algorithms based on a greedy approach~\cite{latour2025fastersimplergreedyalgorithm} and local search~\cite{JiangJLL25} have been developed.
It is worth noting that the fastest constant-factor approximation algorithms for Euclidean $ (k, z) $-clustering can be obtained by running the state-of-the-art graph clustering algorithms on a sparse metric spanner~\cite{Har-PeledIS13}.
The graph setting has also received notable attention for the $k$-center clustering problem~\cite{thorup2005quick,EppsteinHS20,AbboudCLM23,biabani2024k,jin2025beyond}.

\paragraph{Dynamic $k$-clustering algorithms.}
In recent years, the rapid growth and continuous change of data have led to strong interest in developing dynamic $k$-clustering algorithms~\cite{BateniEFHJMW23, Bhattacharya24clustering, TourHS24, bhattacharya2024fully, cruciani2024dynamic, Cohen-AddadHPSS19, FichtenbergerLN21, ForsterSkarlatos}. The related work most similar to ours is by Cruciani, Forster, Goranci, Nazari, and Skarlatos~\cite{cruciani2024dynamic}, who study the $k$-center clustering problem on dynamic graphs under adversarial edge updates; they achieve for example an amortized update time of $ k \myhs n^{o(1)} $ in the incremental setting. Since the shortest-path metric of a graph induces a metric space, it may seem natural to apply algorithms designed for dynamic point sets in metric spaces~\cite{BateniEFHJMW23, Bhattacharya24clustering, bhattacharya2024fully}. However, as noted in~\cite{cruciani2024dynamic}, this approach leads to inefficient algorithms as discussed in the next paragraph.

\paragraph{Comparison of dynamic graphs and dynamic point sets.}
Let us briefly highlight the key differences and challenges between the model with dynamic graphs, which we also consider in this paper, and the model with dynamic point sets in metric spaces~\cite{BateniEFHJMW23,Bhattacharya24clustering, bhattacharya2024fully}. In the model with dynamic point sets, the adversarial updates are point insertions or point deletions. On the other hand, in our model with dynamic graphs, the adversarial updates are edge insertions or edge deletions. This distinction introduces two main challenges, as explained in~\cite{cruciani2024dynamic}: (i) there is no oracle access to all-pairs shortest-path distances, and (ii) a single edge update can affect multiple distances simultaneously, whereas a point update does not affect distances between previously inserted points. As a result, directly applying black-box algorithms designed for dynamic point sets to dynamic graphs is very inefficient.

\bigskip
To the best of our knowledge, no algorithms have been developed so far for the $k$-median, $k$-means, or their generalization $(k, z)$-clustering on graphs undergoing edge updates.
For this reason, a natural question arises about the efficiency of dynamic algorithms for the $(k, z)$-clustering problem on dynamic graphs:
\begin{tcolorbox}[colback=gray!5, colframe=blue!50!black, coltitle=black, sharp corners=south, boxrule=0.7pt]
\begin{center}
    Are there efficient constant-factor approximation algorithms for the $(k, z)$-clustering problem on graphs undergoing edge updates?
\end{center}
\end{tcolorbox}

\subsection{Our Contributions}
In this paper, we answer this question in the affirmative for the incremental setting, in which the weighted undirected graph undergoes adversarial edge insertions. We believe that the incremental setting is particularly interesting and well-motivated from a practical point of view, as there exist inherently incremental graphs in practice.
For example, real-world graphs such as co-authorship networks are incremental, since the fact that two scientists have co-authored a research paper (almost) never changes over time.

Throughout the paper, we assume an oblivious adversary for our randomized results. 
An oblivious adversary fixes the entire sequence of updates
before the algorithm begins. Namely, the adversary cannot adapt the edge updates based on the (randomized) outputs of the
algorithm during the execution. 
Our central result is a highly efficient incremental algorithm---independent of the parameter $k$---that maintains a constant-factor \textit{bicriteria approximate solution} of size $\tilde{O}(k)$ (see \Cref{def:bicriteria}), which is optimal up to polylogarithmic factors for the $(k,z)$-clustering problem. 

\antonis{Reminder: This algorithm generalizes to $k$-center.}

\begin{restatable}{theorem}{incralg}\label{thm:incr_alg}
    There is a randomized incremental algorithm for the $(k, z)$-clustering problem that, given a weighted undirected graph $G = (V, E, w)$ with maximum edge weight $W$
    subject to edge insertions, an integer~$k \geq 1$, and constants $z \geq 1, \epsilon \in (0, 1)$, maintains with high probability:
    \begin{itemize}
        \item a $(O(1), O(\log^3 n \, \log_{1+\epsilon}nW))$-bicriteria approximate solution and
        
        \item a $(O(1), O(\log^3 n \, \log_{1+\epsilon}nW))$-bicriteria approximate assignment,
    \end{itemize}
    and with high probability~has an amortized update time of $\tilde O (n^{o(1)})$. 
\end{restatable}
For clarity of presentation, in~\cref{thm:incr_alg} we hide the dependence on $z$ in the approximation ratio---which is $O((1+\epsilon)^{2z})$---and the dependence on $\epsilon$ and $ W $ in the update time. By combining~\cref{thm:incr_alg} with our incremental reduction algorithm (in~\cref{thm:bicr_to_kzclustering}), we obtain our main result.

\begin{restatable}{theorem}{mainthmkzclust}\label{thm:mainthm_kzcluster}
  There is a randomized incremental algorithm for the $(k,z)$-clustering problem that, given a weighted undirected graph $G=(V,E,w)$ subject to edge insertions, an integer $k \geq 1$, and constants~$z \geq 1, \lambda \geq 1$, maintains with high probability an $O(1)$-approximate $(k,z)$-clustering solution with a total update time of~$\tilde O(k \myhs m^{1+o(1)}+ k^{1+\frac{1}{\lambda}} \myhs m)$ and an amortized update time of $\tilde O(k \myhs n^{o(1)}+ k^{1+\frac{1}{\lambda}})$.
\end{restatable}

For clarity of presentation, in~\cref{thm:mainthm_kzcluster} (our incremental algorithm) we hide the dependence on~$W$ inside the notation~$\tilde{O}(\cdot)$, which we use to suppress polylogarithmic factors in~$nW$.
The hidden dependence on~$\lambda$ in the $O(1)$-approximation ratio is $O(\lambda^{6z})$.

\paragraph{Comparison between $(k, z)$-clustering and $k$-center.}
Let us make a few remarks about our results in comparison with the results in~\cite{cruciani2024dynamic}, which studies the $k$-center objective. In the work by Cruciani, Forster, Goranci, Nazari, and Skarlatos~\cite{cruciani2024dynamic}, the most challenging setting was the incremental one; the decremental setting naturally follows from the connection between the maximal independent set problem and the $k$-center clustering problem. Moreover, the fully dynamic setting follows from a reduction to  fully dynamic approximate SSSP that ``simulates'' a run of a static algorithm after each adversarial edge update.
The latter approach would also work for fully dynamic $ (k, z) $-clustering, but because of the larger constants in the approximation ratio would not be competitive with the naive approach of recomputing from scratch on a dynamic spanner after each adversarial edge update.

In the incremental setting, the approach of~\cite{cruciani2024dynamic} is also to first maintain a bicriteria approximate solution and then construct a $k$-center solution on top of it. Our update time guarantee in~\cref{thm:incr_alg} matches theirs (for bicriteria approximation), and the total update times match the static runtime up to subpolynomial factors.
Their second step, which converts the bicriteria approximate solution into a $k$-center solution, exploits a fast dynamic maximal independent set algorithm~\cite{ChechikZ19,BehnezhadDHSS19} as a black box. In contrast, for the $(k, z)$-clustering problem we need to take additional steps using a dynamic spanner algorithm.

\subsection{Structure of the Paper}
We provide a technical overview in~\cref{sec:techinical_overview}, followed by basic preliminaries in~\cref{sec:preliminaries}. In~\cref{sec:static_bicr_alg}, we present a static version of the bicriteria approximation which helps introduce the ideas behind the incremental version in Section~\ref{sec:incr_bicr}. In particular, it shows that the order of the radii can be made non-decreasing without affecting the approximation ratio by more than a constant factor. In~\cref{sec:bicr_to_kmed_incr}, we combine our incremental bicriteria approximation algorithm with our incremental reduction algorithm to address the $(k, z)$-clustering problem on incremental graphs. 

    \section{Technical Overview}\label{sec:techinical_overview}
In the following, we give a brief overview of our main technical contributions for the incremental setting, where the input graph $G = (V, E, w)$ undergoes edge insertions sent by the adversary. The goal is to develop an efficient incremental constant-factor approximation algorithm for the $(k, z)$-clustering problem on graphs. Our incremental $(k, z)$-clustering algorithm consists of two subroutines.
The first subroutine described in~\cref{sec:incr_bicr}, is an incremental $(O(1), O(\log^3 n \, \log_{1+\epsilon}nW))$-bicriteria approximation algorithm for the $(k, z)$-clustering problem on graphs. 
The second subroutine described
in~\cref{sec:bicr_to_kmed_incr}, efficiently converts the previous bicriteria approximate solution into a constant-factor approximate $(k, z)$-clustering solution after every adversarial edge insertion.

\subsection{Basis of Our Incremental Bicriteria Approximation Algorithm}
The basis of our incremental $(O(1), O(\log^3 n \log_{1+\epsilon} nW))$-bicriteria approximation algorithm for the $(k, z)$-clustering problem is the static $(O(1), O(\log^2 n))$-bicriteria approximation algorithm by Mettu and Plaxton~\cite{mettu2004optimal}, henceforth referred to as the \emph{MP-bi algorithm}. A high-level overview of the static MP-bi algorithm is presented in the next paragraph.

\paragraph{MP-bi algorithm.}
Let $U_0 = V$ be the whole vertex set.
The MP-bi algorithm performs at most $t \coloneqq O(\log \frac{n}{k})$ iterations, where at the $i$-th iteration:
\begin{enumerate}
    \item A candidate set $S_i$ is constructed by sampling a small subset of $U_i$, where $|S_i| = \tilde{O}(k)$.
    
    \item The smallest radius $\nu_i$ is computed such that the ball $B_i$ of radius $\nu_i$ around $S_i$ contains a constant fraction $\beta$ of vertices from $U_i$ (i.e., $|B_i \cap U_i| \geq \beta \myhs |U_i|$).
    
    \item The ball $B_i$ is removed from $U_i$, and the MP-bi algorithm recurses on the remaining vertices $U_{i+1} \coloneq U_i \setminus B_i$.
\end{enumerate}

\vspace{0.5em}
In the incremental setting, the adversary inserts new edges into the graph; our aim is to design an incremental version of the MP-bi algorithm under adversarial edge insertions. We refer to the $i$-th iteration of the MP-bi algorithm as the $i$-th level. Thus, our goal is to incrementally maintain the data structures of the MP-bi algorithm at each level $i$. 

\subsubsection{Challenges in the Incremental Setting}
The central step of the MP-bi algorithm is its second step, where the smallest radius $\nu_i$ is computed. Under edge insertions the distances are non-increasing, and so for a fixed radius $\nu_i$, the ball $B_i$ can only ``receive'' new vertices. In this case, whenever a ball $B_i$ contains more than a $\beta$ fraction of vertices from $U_i$, the radius~$\nu_i$ should be decreased accordingly. The core issue however, is that the subsequent set $U_{i+1}$ of remaining vertices can change in a \emph{fully dynamic} manner, as follows:
\begin{itemize}
    \item Vertices and edges can be removed from the set $U_{i+1}$ by entering the ball at a previous level. This occurs when a newly inserted adversarial edge decreases the distance between the candidate set $S_j$ and a vertex in $U_{i+1}$ to at most $\nu_j$, for some previous level $j < i + 1$.

    \item Vertices and edges can be added to the set $U_{i+1}$ by leaving the ball at a previous level. This occurs when  the radius $\nu_j$ is decreased accordingly because many vertices enter the ball $B_j$ due to a newly inserted adversarial edge, for some previous level $j < i+1$.
\end{itemize}

\noindent
As a consequence, the current candidate set $S_{i+1}$ which was sampled from the old set $U_{i+1}^\old$, may no longer be a good representative in the new set $U_{i+1}$. For this reason, a new candidate set $S_{i+1}$ should be sampled from the modified set $U_{i+1}$, and the new radius $\nu_{i+1}$ should be computed. In other words, the first two steps of the MP-bi algorithm should be repeated at level $i+1$. 
The second step of the MP-bi algorithm in the incremental setting can be implemented straightforwardly by maintaining the ball $B_j$ and the radius $\nu_j$ at each level $j$, as follows:
\begin{itemize}
    \item either by using a fully dynamic SSSP (single-source shortest paths) algorithm with a super-source attached to the candidate set $S_j$,
    \item or by restarting an incremental SSSP algorithm with a super-source attached to the candidate set $S_j$, whenever the set $U_j$ is modified.
\end{itemize}

The first option of employing a fully dynamic SSSP algorithm has a prohibitive update time. The second option of employing an incremental SSSP algorithm depends on the total number of times any set $U_j$ is modified, which is naively proportional to the number of distinct sequences of the radii $\nu_j$. 
Therefore, we can infer that the total update time of the second option depends on the number of distinct sequences of the radii $\nu_j$. 

\paragraph{Number of distinct sequences of the radii.}
Even when the radii are powers of $(1+\epsilon)$ and the number of levels $t$ is $O(\log n)$, the number of distinct sequences of the radii $\nu_j$ is naively at least: 
\[
    \Omega\left( \min\left( \left(\log_{1+\epsilon} nW\right)^{\log n},\, m \right) \right),
\] 
where $W$ is the maximum weight of an edge, and thus $nW$ is the maximum possible distance between two vertices in the graph.
The first term is because the new radius $\nu_j$ may differ arbitrarily from the old radius~$\nu^\old_j$ at level $j$, after repeating the second step of the MP-bi algorithm. The second term is because the sequence of the radii could change after each adversarial edge insertion.

For the second straightforward option, an incremental SSSP algorithm is restarted at most $t = O(\log n)$ times (for the $t$ sets $U_j$) per distinct sequence of the radii. As a result, the total update time required for the second straightforward option is at least:
\[\Omega\left( \min\left( \left(\log_{1+\epsilon} nW\right)^{\log n},\, m \right) \right) \cdot \log n \cdot \text{ the running time of an SSSP algorithm.}\]
In other words, the total update time required for the second straightforward option is at least $\Omega(m \cdot (\log_{1+\epsilon} nW)^{\log n})$, which is excessively high.

\subsection{Incremental Bicriteria Approximation Algorithm}
In principle our approach follows the second option, which is to restart an incremental SSSP algorithm as many times as the number of distinct sequences of the radii $\nu_i$. However, to overcome the issue of exponentially many distinct sequences of $\nu_i$ while ensuring constant-factor approximation, we introduce a refined and elegant solution for managing the radii $\nu_i$. Our solution is a variant of the MP-bi algorithm and relies on two properties that are enforced on the radii. 

The first property which is referred to as the \emph{non-increasing property}, is that the value of each radius is non-increasing over time (i.e., the value of a specific $\nu_i$ either decreases or remains the same). By the non-increasing property, along with the facts that each radius is a power of $(1+\epsilon)$ and that $t = O(\log n)$, we argue that there are at most $O\left(\log_{1+\epsilon} nW \cdot \log n\right)$ distinct sequences of radii. In turn, the update time becomes:
\[
    O\left(\log_{1+\epsilon} nW \cdot \log n\right) \cdot \log n  \cdot \text{ the update time of an incremental SSSP algorithm.}
\]

Nonetheless, since the non-increasing property is enforced, there is no guarantee on the approximation ratio. In particular, simultaneously enforcing the non-increasing property and maintaining $t = O(\log n)$ levels may lead to decisions that do not guarantee a constant-factor approximation. To that end, in order to maintain a constant-factor approximation, we establish the following crucial structural property of the MP-bi algorithm:
\begin{tcolorbox}[colback=gray!5, colframe=blue!50!black, coltitle=black, sharp corners=south, boxrule=0.7pt]
\begin{center}
   The order of the values of the radii $\nu_i$ can be enforced to be non-decreasing (i.e., $\nu_0 \leq \nu_1 \leq \cdots \leq \nu_t$) without affecting the approximation ratio by more than a constant factor.
    \end{center}
\end{tcolorbox}
\noindent We refer to this second property as the \emph{monotonicity property} of the radii.
Using the monotonicity property, along with the fact that the distances are non-increasing in the incremental setting (under adversarial edge insertions), we show that the approximation ratio is constant. As a warm-up, we introduce the monotonicity property also in the static setting. Our static MP-bi variant enforces the monotonicity property from level $0$ up to level $t$, as follows:
\begin{enumerate}
    \item If $\nu_i \geq \nu_{i-1}$, then our MP-bi variant retains the radius $\nu_i$.

    \item If $\nu_i < \nu_{i-1}$, then our MP-bi variant replaces the value of the radius $\nu_i$ with the (updated) value of the radius $\nu_{i-1}$ from the previous level.
\end{enumerate}
Notably, we prove that this modification preserves the constant-factor approximation of the MP-bi algorithm. 
Therefore, the core idea of our approach is to enforce the non-increasing property to achieve efficient update time, while relying on the monotonicity property to guarantee a constant-factor approximation.

\subsubsection{Combining Non-Increasing and Monotonicity Properties in the Incremental Setting} 
In the incremental setting, we demonstrate how both the non-increasing and monotonicity properties can be enforced simultaneously, and how the monotonicity property can be leveraged to argue about the approximation ratio.
Let $\nu_i^\old$ be the old radius at level $i$, and $\nu_i$ be the new radius at level $i$ maintained by the incremental algorithm after an adversarial edge insertion.
  
Consider the set $U_i$ at level $i$, and note that after an adversarial edge insertion some vertices may move closer to the corresponding candidate set $S_i$. In other words, the ball $B_i$ may receive new vertices, causing~$|B_i \cap U_i|$ to exceed $\beta |U_i|$ and the radius $\nu_i$ to be decreased accordingly (i.e., $\nu_i < \nu_i^\old$). 
Consider a vertex~$v \in U_i$ that is within a distance of $\nu^\old_i$ from the candidate set $S_i$. Since under edge insertions the distances are non-increasing, the distance between $v$ and $S_i$ remains at most $\nu_i^\old$. However, it is possible that the distance between $v$ and $S_i$ is greater than the new radius $\nu_i$, which implies that the vertex $v$ does not belong to the ball $B_i$ anymore. In this case, we add such a vertex $v$ in a set $Z$, which we call the \emph{leaking set}. Specifically, the leaking set $Z$ contains all vertices that have exited a ball from a previous level and have \emph{leaked} to a subsequent level. 

Since the set $U_{i+1}$ may be modified, we repeat the first two steps of the MP-bi algorithm at level $i + 1$: we resample a fresh candidate set $\tilde{S}_{i+1}$, add it to the old candidate set~$S_{i+1}$, and compute a \emph{suggested} radius~$\tilde{\nu}_{i+1}$ based on the second step of the MP-bi algorithm. The difference between our incremental algorithm and the straightforward option lies in the fact that the new radius~$\nu_{i+1}$ is updated according to our MP-bi variant, while enforcing both the non-increasing and monotonicity properties:
\begin{enumerate}
    \item If $\tilde{\nu}_{i+1} \geq \nu_{i+1}^\old$, then our incremental algorithm ignores the suggested radius $\tilde{\nu}_{i+1}$ and
    retains the old radius $\nu_{i+1}^\old$ (i.e., $\nu_{i+1} = \nu_{i+1}^\old)$.

    \item Else if $\tilde{\nu}_{i+1} < \nu_i$, then our incremental algorithm sets the new radius $\nu_{i+1}$ to the current radius $\nu_i$ of the previous level.

    \item Otherwise, our incremental algorithm sets the new radius $\nu_{i+1}$ to the suggested radius $\tilde{\nu}_{i+1}$.
\end{enumerate}

\noindent
The first condition enforces the non-increasing property, while the first and second conditions enforce the monotonicity property. As a consequence of the non-increasing property, the number of distinct sequences of the radii $\nu_j$ is reduced to $O(\log n \cdot \log_{1+\epsilon} nW)$. Although we prove that the monotonicity property preserves the constant-factor approximation of the MP-bi algorithm in both the static and incremental settings, in the incremental setting we must also handle the leaking set, which poses an additional challenge. Recall that the leaking set contains vertices that have exited a ball at a previous level and have leaked into a subsequent level. 

\paragraph{Role of the leaking set.}
Since the non-increasing property is enforced for efficiency reasons, in order to identify a $\beta$ fraction of vertices from $U_{i+1}$ (as dictated by the second step of the MP-bi algorithm), we may need to use some vertices from the leaking set $Z$. 
In the analysis, we prove that a suitable subset of the
leaking set $Z$ exists.
Finding a ball $B_{i+1}$ containing a constant fraction $\beta$ of vertices from $U_{i+1}$ (second step of the MP-bi algorithm), and removing $B_{i+1}$ from $U_{i+1}$ (third step of the MP-bi algorithm) is important for maintaining at most $t = O(\log n)$ levels.

\vspace{1em}
To appropriately provide an upper bound on the assigned cost of leaked vertices, we leverage the monotonicity property, as described in the next paragraph.
The monotonicity property acts as a glue, ensuring both the non-increasing property (i.e., efficiency) and the constant-factor approximation are maintained.

\paragraph{Assigned cost of the leaking set.}
Roughly speaking, subsets of the leaking set $Z$ are distributed among the sets~$U_i$ for which the corresponding suggested radius $\tilde{\nu}_i$ is greater than the current radius $\nu_i^\old = \nu_i$.
For a fixed vertex $v \in U_i$ that belongs to the leaking set $Z$, there is a crucial structural observation to be made: The vertex $v$ must have entered the set $U_i$ from a previous level $j \leq i$. Thus, the vertex $v$ is within a distance of at most $\nu_j^\old$ from the candidate set~$S_j$. Observe also that the vertex $v$ may be far from the candidate set~$S_i$. Since under edge insertions the distance between $v$ and $S_j$ remains at most $\nu_j^\old$ and~$\nu_j^\old \leq \nu^\old_i = \nu_i$ by the monotonicity property, we can charge for the vertex $v$ a cost of $\nu_i$. 

Therefore using the monotonicity property, we can argue that each relevant vertex of any set $U_i$ can be charged a cost of $\nu_i$. This in turn helps us to upper bound the approximation ratio of our incremental algorithm, while ensuring that the value of the last level $t$ remains $O(\log n)$ throughout the algorithm.

\subsection{Maintaining $(k, z)$-Clustering on Bicriteria Approximate Solution}
So far, we have described how our incremental algorithm maintains a constant-factor bicriteria approximate solution $S \subseteq V$ of size~$\tilde{O}(k)$. The bicriteria approximate solution $S$ induces a dynamic $(1 + \epsilon)$-metric space~$(P, \dd)$\footnote{Recall that in a $\rho$-metric space, the triangle inequality may be violated by a factor of~$\rho$.} by maintaining incremental $(1+\epsilon)$-approximate~SSSP algorithms for each vertex in $S$ (i.e., for each node in $P = S$).
We ``dynamize'' a known argument to show that the optimal $(k, z)$-clustering solution on the $(1 + \epsilon)$-metric space $(P, \dd)$, with appropriate weights on the nodes in $P$, is an $ O(1) $-approximation to the optimal $(k, z)$-clustering solution on the graph $G$ (see~\cref{lem:convert_bicr_kzclustering}).

Over the course of the adversarial edge insertions to the incremental graph~$G$, both the set $P$ and
the weights of nodes in $P$ change at most $\tilde{O}(1)$ times. The set $P$ can change because the bicriteria approximate solution~$S$ may receive at most $\tilde{O}(k)$ new vertices due to an adversarial edge insertion into $G$.
We emphasize that, for the efficiency of all the subroutines of our incremental $(k, z)$-clustering algorithm, it is important that the total number of vertices that ever belong to $S$ (and in turn in $P$) is at most $\tilde{O}(k)$.
Additionally, the distance between any pair of nodes in $P$ may decrease at most $\tilde{O}(1)$ times, resulting in a total of~$\tilde{O}(k^2)$ distance decreases.

The straightforward approach for maintaining an approximate $(k, z)$-clustering on $(P, \dd)$ would be to recompute it from scratch in time $ \tilde O(|P| \myhs k) = \tilde O(k^2) $ after every change to the set $P$ or to the distance function $\dd$.
This straightforward approach would result in a total update time of $ \tilde O (k \myhs m^{1+o(1)} + k^4) $ for our overall incremental $(k, z)$-clustering algorithm, which implies an amortized update time of $k n^{o(1)} $ for any $ k \leq m^{1/3 + o(1)}$. To obtain an update time close to $ k $ for the full parameter range, we instead employ an approach involving a dynamic spanner algorithm. At a high level, since we achieve vertex sparsification via the bicriteria approximate solution $S$, we exploit it further by performing edge sparsification on top of the vertex-sparsified graph $H = (P, P \times P)$ whose edges are weighted by approximate distances in $G$. 

\paragraph{Our approach.}
In more detail, we view $(P, \dd)$ as a complete graph $ H $ with $\tilde{O}(k^2)$ edges, where the weight of an edge~$(x, y)$ in $H$ corresponds to the approximate distance between $x$ and $y$ in $G$.
We further bin the edges of $H$ into weight classes and maintain a fully dynamic $O(\lambda)$-spanner of size $ \tilde O (k^{1+1/\lambda}) $ on the edges of each weight class~\cite{BaswanaKS12}.
The union $ \tilde{H} $ of these spanners is a $ O (\lambda) $-spanner for $ H $ of size $ \tilde O (k^{1+1/\lambda})$; that is, $\tilde{H}$ is a subgraph of $H$ (in terms of edges) in which the distance between every pair of vertices in $H$ is preserved up to a multiplicative factor of $O(\lambda)$.
This dynamic spanner algorithm~\cite{BaswanaKS12} can deal with adversarial edge updates in polylogarithmic time per update.

Since the changes to the vertex set of $H$ (and of $\tilde{H}$) occur only in $ \tilde O(1) $ many batches, we can restrict the dynamic spanner algorithm to adversarial edge insertions by simply restarting it after every batch of vertex insertions.
Since distances in $ G $ are non-increasing, each edge appears at most once in each weight class in between restarts and thus the total update time for maintaining the spanner $\tilde{H}$ is $ \tilde O (k^2) $. As the final step of handling each of the $ m $ adversarial edge insertions to $ G $, we compute a $(k, z)$-clustering solution on $ \tilde{H} $ in time~$\tilde{O} (k^{1+1/\lambda}) $ from scratch using one of the state-of-the-art static  $(k, z)$-clustering algorithms~\cite{thorup2005quick,latour2025fastersimplergreedyalgorithm,JiangJLL25}.
Overall, all these steps result in an incremental constant-factor approximation $(k, z)$-clustering algorithm with a total update time of $\tilde O(k \myhs m^{1+o(1)}+ k^{1+\frac{1}{\lambda}} \myhs m)$ and an amortized update time of $\tilde O(k \myhs n^{o(1)}+ k^{1+\frac{1}{\lambda}})$, where $\lambda \geq 1$ is an arbitrary fixed constant.


    \section{Preliminaries}\label{sec:preliminaries}

Consider a weighted undirected graph $G = (V, E, w)$ with nonnegative weights.
We denote by $n \coloneqq |V|$ the number of vertices, by $m \coloneqq |E|$ the number of edges, and by $W$ the maximum weight of an edge.%
\footnote{We assume that the weights are at least $1$ and are upper bounded by a polynomial in $n$. This can be achieved by rescaling all the weights.\label{ftnote:W_poly_n}} For a vertex $v \in V$ and a subset of vertices $S \subseteq V$, let $\dist(v, S) \coloneqq \min_{s \in S} \dist(v, s)$ be the \emph{distance} between the vertex $v$ and the set $S$. We write $\dist_H(\cdot, \cdot)$ to explicitly denote the corresponding distance in a graph $H$; when the subscript is omitted, it refers to $\dist_G(\cdot, \cdot)$. 

\begin{definition}[$(k,z)$-clustering problem on graphs]\label{def:k_median}
    Given a weighted undirected graph $G = (V, E, w)$, an integer $k \geq 1$, and a constant $z \ge 1$, the goal in the $(k,z)$-clustering problem
    is to output a subset of vertices $S \subseteq V$ with size at most $k$ (i.e., $|S| \leq k$) such that the value of $\cost^z(S) \coloneqq \sum_{v \in V} \dist(v, S)^z$ is minimized.
\end{definition}

The $(k,z)$-clustering problem includes the problems of $k$-median ($z=1$) and $k$-means ($z=2$). A \emph{$(k, z)$-clustering instance} is the triple consisting of the given input object, the integer parameter $k$, and the constant $z$.
Any subset of vertices $S \subseteq V$ with size at most $k$ is referred to as a \emph{(feasible) solution}, and its cost is referred to as its \emph{$(k,z)$-clustering cost}. Hence, the goal of the $(k,z)$-clustering problem can be rephrased as finding a solution $S$ with minimum $(k,z)$-clustering cost $\OPT \coloneqq \min_{S \subseteq V: |S| \leq k} \text{cost}^z(S)$. The value $\OPT$ is called \emph{optimal $(k,z)$-clustering cost}, and a solution that minimizes the cost is called an \emph{optimal $(k,z)$-clustering solution}.

\begin{definition}[$\rho$-approximate $(k,z)$-clustering problem]\label{def:approxkmedian}
    Given a $(k,z)$-clustering instance, the goal of the $\rho$-approximate $(k,z)$-clustering problem is to output a subset of vertices $S
    \subseteq V$ such that $|S|\leq k$ and $\cost^z(S)\leq \rho \cdot \OPT$. 
\end{definition}

\begin{definition}[$(\alpha, \beta)$-bicriteria approximation] \label{def:bicriteria}
    Given a $(k,z)$-clustering instance, an $(\alpha, \beta)$-bicriteria approximate solution is a subset of vertices $S \subseteq V$ such that:
    \[
        \cost^z(S) \coloneqq \sum_{v \in V} \dist(v, S)^z \leq \alpha \cdot \OPT \;\text{ and } \; |S| \leq \beta \cdot k.
    \]
\end{definition}

\begin{definition}[$(\alpha, \beta)$-bicriteria approximate assignment] \label{def:assignment}
    Given a weighted undirected graph $G = (V, E, w)$ and an integer $k \geq 1$,
    an $(\alpha, \beta)$-bicriteria approximate assignment is a function $\sigma: V \to V$
    such that:
    \[
        \cost^z(\sigma) \coloneqq \sum_{v \in V} \dist(v, \sigma(v))^z \leq \alpha \cdot \OPT \;\text{ and } \; |\sigma(V)| \leq \beta \cdot k.
    \]
Conversely, the preimage of every vertex $u \in V$ is $\sigma^{-1}(u) \coloneqq \set{v \in V \mid \sigma(v) = u}$. 
\end{definition}

Our incremental reduction algorithm in~\cref{sec:bicr_to_kmed_incr} deals with \emph{weighted} $(k, z)$-clustering instances on a subgraph $H$, where $H$ contains both edge and vertex weights.

\begin{restatable}[weighted $(k,z)$-clustering problem on graphs]{definition}{weightedkzclustering} \label{def:wgt_k_median}
    Given a weighted undirected graph $\hat{G} = (\hat{V}, \hat{E}, \hat{w})$ with vertex weights $\wt(\cdot)$, the corresponding weighted $(k,z)$-clustering instance is denoted by $\mathcal{K} \coloneqq (\hat{G}, k, z, \wt)$. The cost of a subset of vertices $C \subseteq \hat{V}$ is defined as $\cost^z_{\wt}(C)\coloneqq \sum_{v \in \hat{V}} \wt(v) \cdot \dist_{\hat{G}}(v, C)^z$. The optimal $(k,z)$-clustering cost of $\mathcal{K}$ is denoted by $\OPT_\mathcal{K}$.
\end{restatable}

Our incremental algorithms utilize the following incremental $(1+\epsilon)$-approximate single-source shortest paths (SSSP) algorithm. We note that similar results have been known implicitly for undirected graphs; these were randomized and assumed an oblivious adversary~\cite{HKN2014hopsets,AnconaHRWW19,LackiN22}.

\begin{lemma}[incremental $(1+\epsilon)$-approximate SSSP \cite{liu2025incremental}] \label{lem:incr_appr_sssp}
    Let $ m \geq 1 $ and $ \epsilon \geq \exp (- (\log m )^{0.99}) $.
    Let $G = (V, E, w)$ be an incremental graph undergoing $m$ edge insertions with integral edge weights in the range $ [1, W] $ with a source $ s \in V $. There is a deterministic data structure that explicitly maintains distance estimates $ \delta_s: V \rightarrow \mathbb{R}_{\geq 0} $ such that:
    \begin{equation*}
        \dist(s,v) \,\le\, \delta_s(v) \,\le\, (1+\epsilon) \dist(s,v),
    \end{equation*}
    and the ability to report the corresponding approximate shortest paths $ \pi_{s, v} $ from $ s $ to $ v $ in time $ O (|\pi_{s, v}|) $.
    The total update time of the data structure is $ m^{1+o(1)} \log W $. Moreover, the algorithm detects and reports the distance estimate changes explicitly. 
\end{lemma}

For a fixed a subset of vertices $S \subseteq V$ and a positive real number $r$, the set $\ball(S, r) \coloneqq \{v \in V \mid \dist(v, S) < r\}$ denotes the \emph{open ball} of radius $r$ around $S$, and the set $\ball[S, r] \coloneqq \{v \in V \mid \dist(v, S) \leq r\}$ denotes the \emph{closed ball} of radius $r$ around $S$. Throughout our algorithms, we use incremental \( (1+\epsilon) \)-approximate SSSP algorithms \( \mathcal{A} \) (\cref{lem:incr_appr_sssp}) with a super-source being a fixed set \( S \), providing distance estimates \( \delta_S(\cdot)\).  
Hence, we extend the definition of balls as follows:
\begin{itemize}
    \item \(\ball(S, r, \delta_S) \coloneqq \{v \in V \mid \delta_S(v) < r\}\) to denote the \emph{open ball} of radius \(r\) around
    \(S\) using the distance estimates \(\delta_S(\cdot)\).  
    \item \(\ball[S, r, \delta_S] \coloneqq \{v \in V \mid \delta_S(v) \leq r\} \) to denote the \emph{closed ball} of radius \(r\) around \(S\) using the distance estimates \(\delta_S(\cdot)\).  
\end{itemize}

A dynamic algorithm has \emph{amortized update time} $u(n, m)$ if its total time spent for processing any sequence of $\ell$ adversarial updates is bounded by $\ell \cdot u(n, m)$.
Throughout this paper, whenever we write that an event holds \textit{with high probability}, we mean that it holds with probability at least $1 - O(1/n^c)$ for some positive constant $c$. We also make use of the following concentration bounds in our proofs.

\begin{lemma}[\cite{dubhashi2009concentration}]\label{lem:chernoff}
Let $X_1,X_2,\cdots,X_i$ be $i$ independent random variables with $0\leq X_j \leq 1$, and let $X=\sum_{j=1}^i X_j$. Then for any $\delta \geq 0$ it holds that:
\[
    \prob{X \geq (1+\delta) \cdot \mathbb{E}[{X}]} \,\leq\, \exp\left(-\frac{\delta^2}{2 + \delta} \mathbb{E}[{X}]\right).
\]
Furthermore, for any $\delta \in [0, 1]$ it holds that:
\[
    \prob{X \leq (1-\delta) \cdot \mathbb{E}[{X}]} \,\leq\, \exp\left(-\frac{\delta^2}{2} \mathbb{E}[{X}]\right).
\]
\end{lemma}
    \section{Static Bicriteria Approximation $(k,z)$-Clustering Algorithm} \label{sec:static_bicr_alg}
We denote by MP-bi the $(O(1), O(\log^2 n))$-bicriteria approximation algorithm of Mettu and Plaxton~\cite{mettu2004optimal}. 
In this section, we present a variant of the MP-bi algorithm that runs in time $O(m \log^2 n)$; the pseudocode of our variant is provided in Algorithm~\ref{alg:static}.

\paragraph{State of our MP-bi variant.}
Our variant utilizes the same data structures as the MP-bi algorithm. Specifically, Algorithm~\ref{alg:static} consists of multiple \emph{levels} starting from $0$. The last level of the algorithm is denoted by $t$, and by construction the value of $t$ is upper bounded by $O(\log n)$. At each level $i \in [0, t]$, the algorithm constructs an \emph{execution set} $U_i$, a \emph{candidate set} $S_i$, defines a \emph{radius} $\nu_i$, and constructs a \emph{ball} $B_i$. The bicriteria approximate solution $S$ is then the union of all candidate sets, that is, $S \coloneqq \bigcup_{i=0}^t S_i$. The algorithm also constructs an \emph{assignment} $\sigma: V \to S$, that maps each vertex to a \emph{candidate center} in~$S$ at a determined distance. The parameters used in the algorithm satisfy the following conditions: $\alpha \ge 1$ is a sufficiently large constant, and $\beta, \epsilon \in (0,1)$ are small constants.

\paragraph{High-level overview of our MP-bi variant.}
Roughly speaking, our variant of the MP-bi algorithm proceeds in iterations. In the $i$-th iteration, it samples a set of candidate centers $S_i$ and computes the smallest radius $\tilde{\nu}_i$, which is a power of $(1+\epsilon)$, such that the ball of radius $\tilde{\nu}_i$ around $S_i$ covers a constant fraction $\beta$ of the execution set $U_i$.
The actual radius is then set to $\nu_i \coloneqq \max(\tilde{\nu}_i, \nu_{i-1})$ in order to enforce the monotonicity property on the radii. This is a key property mainly needed for the incremental setting, which we discuss in detail in~\cref{sec:incr_bicr} (see~\ref{property_rad_3} in~\cref{sec:incr_bicr}).
The $i$-th ball $B_i$ is the subset of the ball of radius $\nu_i$ around~$S_i$ that is contained within $U_i$ (see~\linecref{line:construct_Bi} in~\cref{alg:static}).
The vertices in $B_i$ are assigned to candidate centers in $S_i$ via the assignment $\sigma$. The $(i+1)$-th execution set $U_{i+1}$ 
is then defined as~$U_i$ without $B_i$, and the next iteration proceeds.

The intuition is that the algorithm hits large clusters with high probability by sampling a subset of size $|S_i| = \Theta(k \log n)$ per level. Subsequent iterations are then responsible for hitting the smaller clusters. Since the size of $U_i$ decreases by at least a constant factor $(1-\beta)$, there are at most $O(\log n)$ iterations (levels). Thus with high probability, the size of the bicriteria approximate solution $S = \bigcup_{i=0}^t S_i$ is upper bounded by~$\tilde{O}(k)$.

\begin{algorithm}[H]\footnotesize
\algnewcommand{\LineComment}[1]{\State \(\triangleright\) #1}
\caption{Variant of MP-bi on Graphs}
\label{alg:static}
\begin{algorithmic}[1]

\Require{Graph $G$ and positive integer $k$ (problem input); 
positive constants $\alpha,\beta$ (MP-bi parameters);
positive constant $\epsilon$ (MP-bi variant parameter)}
\Ensure{Bicriteria approximate solution $S$ and bicriteria approximate assignment $\sigma$}
\vspace{0.5em}

\State $\nu_{-1} \gets 0$
\State $i \gets 0$
\State $U_0 \gets V$
\While{$|U_i| > \alpha k \log n$}
    \State Construct $S_i$ by sampling each $v\in U_i$ independently with probability $\min\Big(\frac{\alpha \myhs k \log n}{|U_i|}, 1\Big)$ \label{line:sample_Si} 
    \State $\tilde{\nu}_i \gets \min_{j\in \mathbb{Z}^{\geq 0}} \{(1+\epsilon)^j\mid \lvert\ball[S_i, (1+\epsilon)^j]\cap U_i\rvert \geq \beta \myhs |U_i|\}$ \label{line:dijkstra}
    \State $\nu_i \gets \max(\tilde{\nu}_i, \nu_{i-1})$ \label{line:max_nu}
    \State $B_i \gets \ball[S_i, \nu_i] \cap U_i$ \label{line:construct_Bi}
    \State For every $v \in B_i: \sigma(v) \gets \argmin_{u \in S_i} \dist(v, u)$
    \State $U_{i+1} \gets U_i \setminus B_i$ \label{line:new_exec_set}
    \State $i \gets i + 1$
\EndWhile

\State $t \gets i$
\State $S_t \gets U_t$, $B_t \gets U_t, \nu_t \gets \nu_{t-1}, \forall v \in U_t: \sigma(v) \gets v$ \label{line:last_level}
\State $S \gets \bigcup_{i=0}^{t} S_i$ \label{line:final_S}

\end{algorithmic}
\end{algorithm}

\paragraph{Implementation of~\linecref{line:dijkstra}.}
The radius $\tilde{\nu}_i$ is computed using a single-source shortest paths algorithm by connecting each vertex $v \in S_i$ to a super-source via zero-weight edges. Hence, the step in~\linecref{line:dijkstra} of~\cref{alg:static} takes $O(m \log n)$ time. Observe that the exploration of $\ball[S_i, \nu_i]$ is not restricted to $U_i$, because it may traverse edges outside $U_i$ to reach some vertices within $U_i$.

\begin{observation}\label{obs:order radii nondecreasing}
    The order of the values of the radii $\nu_i$ is non-decreasing, namely it holds that: 
    \[\nu_0 \;\leq\; \nu_1 \;\leq\; \cdots \;\leq \nu_t.\]
\end{observation}
\begin{proof}
    The $i$-th radius $\nu_i$ is assigned the value of $\max(\tilde{\nu}_i, \nu_{i-1})$ in~\linecref{line:max_nu}, which implies that
    $\nu_i \geq \nu_{i-1}$. Therefore, the claim follows by an induction argument on the radii.
\end{proof}

\paragraph{Differences between \cref{alg:static} and MP-bi algorithm.}
The main difference between our MP-bi variant and the standard MP-bi algorithm lies in Line~\ref{line:max_nu}, where Algorithm~\ref{alg:static} enforces a non-decreasing order of the radii. We leverage this monotonicity property on the radii later in the dynamic setting.
We also note that Algorithm~\ref{alg:static} has been adjusted in order to work with graphs and discrete distances of powers of $(1+\epsilon)$,
while the original MP-bi algorithm works with point sets and non-discrete distances. 
The discrete distances on the radii that are powers of $(1+\epsilon)$ are needed to achieve efficiency in the dynamic setting.
Since our static~\cref{alg:static} differs from the original MP-bi algorithm, we present its analysis in~\cref{sec:analysis_static}.

\subsection{Analysis of the MP-bi Variant} \label{sec:analysis_static}
Mettu and Plaxton prove in~\cite[Section 3]{mettu2004optimal}
that the MP-bi algorithm produces with high probability a constant-factor approximate solution of size at most $O(k \log^2 n)$. Our aim is to show that our MP-bi variant provides the same guarantees, as demonstrated in Theorem~\ref{thm:variant_approx}. 

\begin{restatable}{theorem}{variantapprox}\label{thm:variant_approx}
    The set $S$ output by \cref{alg:static} is a  $(O(1), O(\log^2 n))$-bicriteria approximate solution with high probability. 
\end{restatable}

Our goal in this section is to prove \cref{thm:variant_approx}, and its analysis is divided into four subsections, as outlined below.
First, in~\cref{sec:size_of_sol_variant} we provide an upper bound on $|S|$. Next, in~\cref{sec:upper_bound_variant} we provide an upper bound on $\cost^z(S)$. Then, in~\cref{sec:lower_bound_variant} we provide a lower bound on the optimal cost in terms of $\cost^z(S)$. Finally, in~\cref{sec:proof_variant} we conclude the proof of~\cref{thm:variant_approx}.
We remark that the essential ideas in the analysis are similar to those in~\cite[Section 3]{mettu2004optimal}.

\subsubsection{Size of the Bicriteria Approximate Solution} \label{sec:size_of_sol_variant}
\begin{lemma} \label{lem:exec_decr_beta}
    For a fixed level $i: 0 \leq i < t$, it holds that $|U_{i+1}| \leq (1-\beta) \myhs |U_i|$.
\end{lemma}
\begin{proof}
    Based on~\linecref{line:dijkstra} of~\cref{alg:static}, the value of $\tilde{\nu}_i$ is the smallest power of $(1+\epsilon)$ such that the corresponding closed ball
    contains at least $\beta |U_i|$ vertices from the $i$-th execution set $U_i$. By~\linecref{line:max_nu} we have $\nu_i\geq \tilde{\nu}_i$, and thus $|B_i|\geq \beta |U_i|$. By the definition of $U_{i+1}$ in Line~\ref{line:new_exec_set}, it follows that $|U_{i+1}| = |U_i|-|B_i|\leq (1-\beta)|U_i|$, as needed.
\end{proof}

\begin{corollary}\label{cor:value_t}
    The value of the last level $t$ is at most $O(\log n)$.
\end{corollary}
\begin{proof}
    Based on~\cref{lem:exec_decr_beta}, we can infer that $|U_{t-1}| \leq (1-\beta)^{t-1} |U_0|$.
    Recall that $|U_0| = n$, $|U_{t-1}| \geq 1$, and $0 < (1-\beta) < 1$. Hence, it follows that $t\leq \frac{\log n}{\log (\frac{1}{1-\beta})} + 1$.
    Therefore as $\beta$ is a fixed constant, we conclude that the value of $t$ is at most $O(\frac{\log n}{\beta}) = O(\log n)$.\footnote{Note that we can assume $n > 1$, as otherwise the $(k,z)$-clustering instance can be trivially solved.}
\end{proof}

\begin{lemma} \label{lem:size_of_S}
    With high probability, it holds that $|S| = O(k\log^2 n)$. 
\end{lemma}
\begin{proof}
    The bicriteria approximate solution $S$ is equal to $\bigcup_{i=0}^t S_i$ (see~\linecref{line:final_S} in~\cref{alg:static}). Observe that $|S_t| = O(k\log n)$. Next, we prove that with high probability it holds that $|S_i|=O(k\log n)$ for all levels $i \in [0, t-1]$. Combining this  with~\Cref{cor:value_t} implies that $|S| = O(k \log^2 n)$. 
    
    For a fixed level $i \in [0, t-1]$, the $i$-th candidate set $S_i$ is constructed by sampling each vertex of the $i$-th execution set $U_i$ independently with probability at most $\frac{\alpha \myhs k\log n}{|U_i|}$ (see~\linecref{line:sample_Si} in~\cref{alg:static}). This means that $|S_i|$ is a sum of independent $0\text{-}1$ indicator random variables $\set{X_{v}}_{v\in U_i}$. Thus, we have $\mathbb{E}[|S_i|]=\mathbb{E}[\sum_{v\in U_i} X_v] = \alpha \myhs k\log n$. The conditions of \Cref{lem:chernoff} are satisfied, and by setting $\delta = 10$ we get that:
    \begin{align*}
        \Pr\left(|S_i|\geq 11 \alpha \myhs k\log n\right) \;\leq\; \exp
        \left(-8\alpha \myhs k\log n\right) \;\leq\; \frac{1}{n^{8\alpha\myhs k}} \;\leq\; \frac{1}{n^{8\alpha}},
    \end{align*}
    where the last inequality follows from the fact that $k
    \geq 1$. By applying a union bound over the $t = O(\log n)$ candidate sets $S_i$ and choosing $\alpha$ appropriately, we deduce that with high probability for all levels $i \in [0, t-1]$ it holds that $|S_i|\leq 11\alpha \myhs k\log n$. Consequently, since $\alpha$ is a sufficiently large constant, with high probability the size of the bicriteria approximate solution $S$ is~$O(k \log^2 n)$.
\end{proof} 

\subsubsection{Upper Bound on the Cost of the Bicriteria Approximate Solution} \label{sec:upper_bound_variant}
\begin{lemma} \label{lem:v_to_one_ball}
    Each vertex $v \in V$ belongs to exactly one ball $B_i$, where $0 \leq i \leq t$.
\end{lemma}
\begin{proof}
    For a fixed vertex $v \in V$, let $i$ be the first level such that the vertex $v$ enters the $i$-th ball $B_i$. Due to Lines~\ref{line:construct_Bi}, \ref{line:new_exec_set}, and \ref{line:last_level} in~\cref{alg:static}, such a level $i$ exists. By construction in Line~\ref{line:new_exec_set}, the vertex $v$ is not part of any $U_j$ with $j > i$. Moreover according to Line~\ref{line:construct_Bi}, the vertex $v$ does not belong to any $B_j$ with $j > i$. Therefore, the vertex $v$ belongs only to the $i$-th ball $B_i$, where $0 \leq i \leq t$.
\end{proof} 

\begin{lemma} \label{lem:upper_bound_sol}
    It holds that $\cost^z(S) \leq \sum_{i=0}^{t-1} |B_i| (\nu_i)^z$.
\end{lemma}
\begin{proof}
The cost of the bicriteria approximate solution $S$ is evaluated as follows:
\begin{align*}
    \cost^z(S) &\;=\; \sum_{v \in V} \dist(v, S)^z \;=\; \sum_{i=0}^t \sum_{v \in B_i} \dist(v, S)^z &(\text{Since }\textstyle \bigsqcup_{i=0}^t B_i=V \text{ by~\cref{lem:v_to_one_ball}}) \\
    &\;=\; \sum_{i=0}^{t-1} \sum_{v \in B_i} \dist(v, S)^z+\sum_{v \in B_t} \dist(v, S)^z 
    \;=\; \sum_{i=0}^{t-1} \sum_{v \in B_i} \dist(v, S)^z & \text{(From \linescref{line:last_level}{line:final_S} in~\cref{alg:static})} \\
    &\;\leq\; \sum_{i=0}^{t-1} \sum_{v \in B_i} \dist(v, S_i)^z \;\leq\; \sum_{i=0}^{t-1} \sum_{v \in B_i} (\nu_i)^z &\text{(From \linecref{line:construct_Bi} in~\cref{alg:static})} \\
    &\;\leq\; \sum_{i=0}^{t-1} |B_i| (\nu_i)^z.
\end{align*}
\end{proof}

\subsubsection{Lower Bound on the Optimal Solution} \label{sec:lower_bound_variant}
We continue with the lower bound on the optimal $(k, z)$-clustering cost, adapting the
analysis in~\cite[Section 3]{mettu2004optimal}. We fix a positive real number 
$\gamma$ such that $\beta < \gamma < 1$, and we define $\nu^*_i$ and $\mu^*_i$ for a fixed level~$i \in [0, t]$, as described in Definition~\ref{def:nu_star} and Definition~\ref{def:mu_star} respectively.

\begin{definition} \label{def:nu_star}
For a fixed level $i: 0 \leq i \leq t$, consider the $i$-th execution set $U_i$ and the $i$-th candidate set $S_i$ constructed in Algorithm~\ref{alg:static}. For this specific level $i$, we define $\nu_i^*$ as follows:
\[
    \nu_i^* \coloneqq \min\{r \in \mathbb{R} \mid \lvert \ball[S_i, r]\cap U_i\rvert \geq \beta \myhs |U_i|\}.
\]
\end{definition}

Observe that for a fixed level $i: 0 \leq i < t$, we can think of $\tilde{\nu}_{i}$ computed in Line~\ref{line:dijkstra} of~\cref{alg:static} as being obtained by rounding $\nu_{i}^*$ to the next power of $(1+\epsilon)$. 

\begin{observation}\label{obs:relation_nu_tilde_nu_star}
    For a fixed level $i: 0 \leq i < t$, it holds that $\nu_i^* \leq \tilde{\nu}_i \leq (1+\epsilon) \myhs \nu_i^*$.
\end{observation}

\begin{definition}[\cite{mettu2004optimal}] \label{def:mu_star}
    For a fixed level $i: 0 \leq i \leq t$, consider the $i$-th execution set $U_i$
    constructed in Algorithm~\ref{alg:static}. For this specific level $i$, we define $\mu_i^*$ to be the minimum nonnegative real number such that there exists a subset of vertices $X_i \subseteq V$ with size at most $k$ (i.e., $|X_i| \leq k$) for which the following properties hold:
    \begin{enumerate}
        \item $\lvert \ball[X_i, \mu_i^*]\cap U_i\rvert \,\geq\, \gamma \myhs |U_i|$. 
        \item $\lvert U_i\setminus \ball(X_i, \mu_i^*)\rvert \,\geq\, (1-\gamma) \myhs |U_i|$. 
    \end{enumerate}
\end{definition}

We point out that in the above definition, the vertices of $U_i$ that are exactly at distance $\mu_{i}^*$ from $X_i$ belong to both $\ball[X_i, \mu_i^*]\cap U_i$ and $U_i\setminus \ball(X_i, \mu_i^*)$. The main probabilistic claim in~\cite[Section 3]{mettu2004optimal} is Lemma~3.3, restated below as~\cref{lem:nu<2mu}. We observe that in our context, analogous guarantees apply with a similar proof provided in \Cref{appendix:missingproofs} for completeness.

\newcommand{\lemmaref}{\cite[Lemma~3.3]{mettu2004optimal}}

\begin{restatable}[Lemma 3.3 in \cite{mettu2004optimal}]{lemma}{relationnumu}\label{lem:nu<2mu}
    With high probability, it holds that 
    $\nu_i^* \leq 2 \myhs \mu^*_i$ for all levels $i \in [0, t]$.
\end{restatable}

\begin{corollary}\label{cor:tilde_nu_bound}
    With high probability, it holds that $\tilde{\nu}_{i} \le 2 \myhs  (1+\epsilon) \myhs \mu_i^*$ for every level $i \in [0, t)$.
\end{corollary}
\begin{proof}
    The claim follows from \cref{obs:relation_nu_tilde_nu_star} and \cref{lem:nu<2mu}.
\end{proof}

Even though \cref{cor:tilde_nu_bound} provides an upper bound for $\tilde{\nu}_{i}$, recall that Algorithm~\ref{alg:static} uses the radius~$\nu_i$ computed in~\linecref{line:max_nu}, which is at least~$\tilde{\nu}_{i}$. In particular, we upper bound the cost of the bicriteria approximate solution $S$ with respect to $\nu_i$ in Lemma~\ref{lem:upper_bound_sol}.
Hence, in the next lemma we provide an upper bound on the value of $\nu_i$ with respect to the value of a $\mu^*_j$ which later is associated with the optimal $(k, z)$-clustering cost.

\begin{lemma} \label{lem:nu_bound}
For a fixed level $i: 0 \leq i < t$, the value of the $i$-th radius $\nu_i$ can be upper bounded as follows:
\[
    \nu_i \le 2 \myhs (1+\epsilon) \myhs \mu_j^*, \;\text{where }j = \max \{\ell \le i \mid \tilde{\nu}_\ell \ge \nu_{\ell-1}\}.
\]
\end{lemma}
\begin{proof}
By construction in~\linecref{line:max_nu}, we have 
$\nu_i = \max(\tilde{\nu}_{i}, \nu_{i-1})$.
Hence, if $\tilde{\nu}_{i} \ge \nu_{i-1}$ then $\nu_i = \tilde{\nu}_{i}$, and otherwise if~$\tilde{\nu}_{i} < \nu_{i-1}$ then $\nu_i =  \nu_{i-1}$. By applying such reasoning and by the definition of $j = \max \{\ell \le i \mid \tilde{\nu}_\ell \ge \nu_{\ell-1}\}$, we can infer that $\nu_i = \tilde{\nu}_j$.
Therefore by~\cref{cor:tilde_nu_bound}, it follows that $\nu_i = \tilde{\nu}_j \leq 2 (1+\epsilon) \mu_j^*$.
\end{proof}

In the analysis in~\cite[Section 3]{mettu2004optimal}, Mettu and Plaxton define a parameter~$r \coloneqq \ceil{\log_{(1-\beta)}((1-\gamma)/3)}$ and the sets $F_i$
and $F_i^r$, for every level $i \in [0, t]$.
Afterwards, the authors utilize these sets in the proofs of their Lemmas~$3.6, 3.7, 3.8, 3.9$, and~$3.10$, in order to conclude their Lemma~$3.11$.
We remark that for our analysis we can use the same definitions for the aforementioned sets, and in turn the corresponding lemmas hold here as well. Specifically, we use a version of their Lemma~$3.11$
whose proof ends on the penultimate line and the $(k, z)$-clustering cost is used, restated as Lemma~\ref{lem:lower_bound_opt}.\footnote{In the proof of Lemma~$3.11$ in \cite{mettu2004optimal}, there is a small typo in the definition of $\ell$. Specifically, the parameter $\mu_i$ is missing from the inner sum $\sum_{i \in G_{j, r}}$.}
For completeness, we provide a more detailed discussion of these in~\cref{appendix:missingproofs}.

\begin{restatable}{lemma}{lowerboundOPT} \label{lem:lower_bound_opt}
    For any subset of vertices $X \subseteq V$ with size at most $k$ (i.e., $|X| \leq k$), it holds that:
    \[
    \cost^z(X) \geq \frac{1-\gamma}{2r} \sum_{i=0}^t |U_i| \myhs (\mu^*_i)^z,
    \]
    where $r = \ceil{\log_{(1-\beta)}((1-\gamma)/3)}$.
\end{restatable}

\subsubsection{Finishing the Proof of~\cref{thm:variant_approx}} \label{sec:proof_variant}

\begin{observation} \label{obs:bound_sum_Uj_Ui}
    For a fixed level $i: 0 \leq i < t$, it holds that $\sum_{i \leq j < t} |U_j \setminus U_{j+1}| \leq |U_i|$.
\end{observation}
\begin{proof}
    For every level $j$, it holds that $U_{j+1} \subseteq U_j$ due to~\linecref{line:new_exec_set} in~\cref{alg:static}. Moreover, notice that the subsets $\{U_j \setminus U_{j+1} \}_{i \leq j < t}$ are pairwise disjoint and are all contained in the $i$-th execution set $U_i$. Therefore, the sum of their sizes is at most the size of $U_i$, as needed.
\end{proof}

By combining the previous observations and lemmas, we can now finish the proof of \cref{thm:variant_approx} which we restate for convenience.

\variantapprox*
\begin{proof}
    From~\cref{lem:size_of_S}, we have $|S|=O(k\log^2n)$ with high probability. Hence, it remains to upper bound the approximation ratio, which we establish in the rest of the proof. Based on \cref{lem:upper_bound_sol}, it holds that $\cost^z(S) \leq \sum_{i=0}^{t-1} |B_i| (\nu_i)^z$. By~\linescref{line:construct_Bi}{line:new_exec_set} in~\cref{alg:static}, we have $B_i \subseteq U_i$, $U_{i+1} \subseteq U_i$, and  
    $|B_i| = |U_i| - |U_{i+1}| = |U_i \setminus U_{i+1}|$, and thus it follows that $\cost^z(S) \leq \sum_{i=0}^{t-1} |U_i \setminus U_{i+1}| \, (\nu_i)^z$.
    Based on \cref{lem:nu_bound}, each radius $\nu_i$ is bounded by some $\mu^*_j$, where $i \in [0, t)$ and $j \in [0, i]$. For a fixed level~$j: 0 \leq j \leq t$, let $A_j$ be the set of all indices~$i: j \leq i < t$ that use $\mu^*_j$ in their bound in \cref{lem:nu_bound}.
    Then the value of $\cost^z(S)$ can be upper bounded as follows:
    \begin{align*}
        \cost^z(S) &\;\leq\; \sum_{i=0}^{t-1} |U_i \setminus U_{i+1}| (\nu_i)^z \;=\; \sum_{j=0}^t \sum_{i \in A_j} |U_i \setminus U_{i+1}| (\nu_i)^z \\
        &\;\leq\; \sum_{j=0}^t \sum_{i \in A_j} |U_i \setminus U_{i+1}| \myhs (2(1+\epsilon)\mu^*_j)^z \\
        &\;=\; (2(1+\epsilon))^z \sum_{j=0}^t (\mu^*_j)^z \sum_{i \in A_j} |U_i \setminus U_{i+1}|.
    \end{align*}
    Based on~\cref{obs:bound_sum_Uj_Ui}, it holds that 
    $\sum_{i \in A_j} |U_i \setminus U_{i+1}| \leq |U_j|$, and recall that if $i \in A_j$ then $j \leq i < t$.
    Thus, we can conclude that:
    \[
        \cost^z(S) \;\leq\; (2(1+\epsilon))^z \sum_{j=0}^t (\mu^*_j)^z \myhs |U_j|.
    \]
    Based on \cref{lem:lower_bound_opt}, we have:
    \[ \sum_{j=0}^t (\mu^*_j)^z |U_j| \;\leq\; \frac{2r}{1-\gamma} \myhs \OPT,\]
    which implies that:
    \[
        \cost^z(S) \; \leq \; \frac{2r (2(1+\epsilon))^z}{1-\gamma} \myhs \OPT.
    \]
    Therefore as~$z,\epsilon,r,\gamma=~O(1)$, it follows that the approximation ratio of \cref{alg:static} is upper bounded by~$O(1)$, as required.
\end{proof}
    \section{Incremental Bicriteria Approximation for $(k, z)$-Clustering on Graphs} \label{sec:incr_bicr}
In this section, we develop an incremental bicriteria approximation algorithm for the $(k, z)$-clustering problem on graphs undergoing adversarial edge insertions. Our incremental bicriteria approximation algorithm, with high probability, achieves an amortized update time of $n^{o(1)}$ and maintains a constant-factor bicriteria approximate solution of size at most $O(k \, \log^3 n \, \log_{1+\epsilon}nW)$. 
A pseudocode of our incremental bicriteria approximation algorithm is provided in~\cref{alg:incremental}, and the result is demonstrated in the following theorem.

\incralg*

\antonis{Our algorithm maybe works for the weighted (on vertices) $k$-median as well.}

Our incremental bicriteria approximation algorithm is an incremental version of the static Algorithm~\ref{alg:static}, which itself is a variant of the MP-bi algorithm. Throughout the description and analysis of the incremental bicriteria approximation algorithm in~\cref{sec:incr_bicr}, we fix a sufficiently large constant $\alpha \geq 1$ and a small constant $\beta \in (0, 1)$.  
In the next two paragraphs, we first describe the information maintained by our incremental bicriteria approximation algorithm, and then outline the properties satisfied by the radii~$\nu_i$. 

\paragraph{State of the incremental bicriteria approximation algorithm.}
Our incremental bicriteria approximation~\cref{alg:incremental} is structured into multiple \emph{levels} starting from $0$. The last level of the algorithm is denoted by~$t$, and the value of $t$ is upper bounded by $O(\log n)$ throughout the algorithm (see~\cref{cor:number_of_levels}). At each level $i \in [0, t]$, the  bicriteria approximation algorithm maintains an \emph{execution set} $U_i$, a \emph{candidate set} $S_i$, a \emph{radius} $\nu_i$, an incremental $(1+\epsilon)$-approximate SSSP algorithm $\mathcal{A}_i$ with a super-source attached to $S_i$, an \emph{approximate ball} $B_i$, and a \emph{leaking set} $Z_i$. 

The incremental algorithm then maintains
the bicriteria approximate solution $S \coloneq \bigcup_{i=0}^t S_i$, and an \emph{assignment} $\sigma: V \to S$ that maps each vertex to a \emph{candidate center} in~$S$ at a specified distance.
The algorithm also uses the temporary leaking set $Z$, an auxiliary set that essentially stores the \emph{pending leaked} vertices to be assigned to some leaking sets~$Z_i$ (see Lines~\ref{algline:leaking_Z},~\ref{algline:count_new_for_Zi},~\ref{algline:leaking_Zi}, and~\ref{algline:rem_Zi_from_Z} in~\cref{alg:incremental}).

\paragraph{Properties of the radii.}
For efficiency reasons, the values of the radii~$\nu_i$ (where $0 \leq i \leq t$) satisfy three properties over the course of the algorithm. 
(1) The first property is that the value of each radius $\nu_i$ is a power of $(1+\epsilon)$.
Hence, the number of possible values of a radius $\nu_i$ is reduced to $\log_{1+\epsilon}nW$,\footnote{Recall that $W$ is the maximum weight of an edge, and thus $nW$ is the maximum possible distance between two vertices in~$G$.}
at the cost of an extra factor of $(1+\epsilon)$ in the approximation ratio.
(2) The second property is that the value of each radius~$\nu_i$ is non-increasing over time, which means that after an adversarial edge insertion, the value of a specific~$\nu_i$ either decreases or remains the same. 
As a result, the number of distinct sequences of $\nu_i$ is reduced from $(\log_{1+\epsilon}nW)^t$ to $\log_{1+\epsilon}nW \cdot t$.
We note that the value of the last level $t$ may increase due to adversarial edge insertions, resulting in more radii $\nu_i$. However, the incremental algorithm ensures that~$t = O(\log n)$ throughout its execution.

Furthermore, in order to upper bound the cost of the bicriteria approximate solution, the values of the radii~$\nu_i$ (where $0 \leq i \leq t$) satisfy also a third property over the course of the algorithm. (3) The third property is that the order of the values of $\nu_i$ is non-decreasing (i.e., $\nu_0 \leq \nu_1 \leq \cdots \leq \nu_t)$. 
Let us provide a high-level justification for the third property: Consider a vertex $v \in U_i$ that is within a distance of $\nu_i$ from the candidate set $S_i$. 
Assume that after an adversarial edge insertion, the vertex $v$ moves to the next execution set $U_{i+1}$ and is not within a distance of $\nu_{i+1}$ from the next candidate set $S_{i+1}$. Since under edge insertions the distance between $v$ and $S_i$ remains at most $\nu_i$ and $\nu_i \leq \nu_{i+1}$ by the third property, the vertex $v$ can be charged a cost of $\nu_{i+1}$. Therefore using the third property, we can argue that a (leaked) vertex of any execution set $U_j$ can be charged a cost of $\nu_j$ (see also~\cref{obs:incr_nu_monoton,lem:cost_of_vert}), which in turn helps us to upper bound the approximation ratio of our incremental algorithm. To sum up, our incremental bicriteria approximation algorithm ensures that the following three properties are satisfied throughout its execution:
\begin{enumerate}[label=\textbf{Property \arabic*}, itemindent=3em, itemsep=-0.2em]
    \item\label{property_rad_1} \hspace{-0.4em}: The value of each radius $\nu_i$ is a power of $(1+\epsilon)$.
    \item\label{property_rad_2}\hspace{-0.4em}: The value of each radius $\nu_i$ is non-increasing over time. 
    \item\label{property_rad_3}\hspace{-0.4em}: The order of the values of the radii $\nu_i$ is non-decreasing (i.e., $\nu_0 \leq \nu_1 \leq \cdots \leq \nu_t$).
\end{enumerate}

\subsection{Incremental Bicriteria Approximation Algorithm on Graphs}

First we describe the preprocessing phase, and afterwards we describe how to handle adversarial edge insertions. A pseudocode of our incremental bicriteria approximation algorithm is provided in~\cref{alg:init-incremental,alg:incremental}.

\begin{definition}[Retain] \label{def:retain}
    Given a set $A$ and a nonnegative integer $\lambda$, the set $\Retain(A, \lambda)$ denotes an arbitrary subset of $A$ with size $\min(\lambda, |A|)$.
\end{definition}

\subsubsection{Preprocessing Phase}
The initialization of our incremental bicriteria approximation algorithm mimics the structure of our MP-bi variant~\cref{alg:static}. During the preprocessing phase, dynamic data structures are also employed, as described in \cref{alg:init-incremental}.

\begin{algorithm}[H]\footnotesize
\algnewcommand{\LineComment}[1]{\State \(\triangleright\) #1}
\caption{Initialization of Incremental Bicriteria Approximation Algorithm}\label{alg:init-incremental}
\begin{algorithmic}[1]

\Require{Graph $G$ and positive integer $k$ (problem input); 
positive constants $\alpha,\beta$ (MP-bi parameters);
positive constant $\epsilon$ (MP-bi variant parameter, used also as SSSP approximation parameter)}
\Ensure{Bicriteria approximate solution $S$ and bicriteria approximate assignment $\sigma$}

\LineComment{The algorithm has global access to all $U_i, S_i, \nu_i, \mathcal{A}_i, B_i, Z_i, Z, \sigma(\cdot)$}
\vspace{0.5em}

\Function{Initialize}{$G, k$}

\State $\nu_{-1} \gets 0$
\State $i \gets 0$
\State $U_0 \gets V$

\While{$|U_i| > \alpha \myhs k \log n$}
    \State Construct $S_i$ by sampling each $v\in U_i$ independently with probability $\min\Big(\frac{\alpha \myhs k \log n}{|U_i|}, 1\Big)$
    \State $\mathcal{A}_i.\textit{initialize}(G, S_i)$ \Comment{$\mathcal{A}_i$ is an incremental $(1+\epsilon)$-approximate SSSP algorithm, providing distance estimates $\delta_{S_i}(\cdot)$}
    \State $\tilde{\nu}_i \gets \min_{j \in \mathbb{Z}^{\geq 0}} \{(1+\epsilon)^j\mid \lvert\ball[S_i, (1+\epsilon)^j, \delta_{S_i}]\cap U_i\rvert \geq \beta \myhs |U_i|\}$ \label{algline:incr_dijkstra}
    \State $\nu_i \gets \max(\tilde{\nu}_i, \nu_{i-1})$ \label{algline:incr_max_nu}
    \State $B_i \gets  \Retain\bigl(\ball[S_i, \nu_i, \delta_{S_i}] \cap U_i,\, \ceil{\beta \myhs |U_i|}\bigr)$
    \label{algline:incr_construct_Bi}
    \State For every $v \in B_i:$ $\sigma(v) \gets \argmin_{u \in S_i} \delta_u(v)$\label{line:algassignment}
    \State $Z_i \gets \emptyset$
    \State $U_{i+1} \gets U_i \setminus B_i$ \label{algline:exec_assign_without_ball_prepro}
    \State $i \gets i + 1$
\EndWhile

\State $t \gets i$
\State $S_t \gets U_t$, $B_t \gets U_t$, $Z_t \gets \emptyset$, $Z \gets \emptyset$, $\nu_t \gets \nu_{t-1}$, $\forall v \in U_t: \sigma(v) \gets v$
\State $S \gets \bigcup_{i=0}^{t} S_i$
\EndFunction 
\end{algorithmic}
\end{algorithm}

\noindent Since we have to deal with adversarial edge insertions, an incremental $(1+\epsilon)$-approximate SSSP algorithm~$\mathcal{A}_i$ (from~\cref{lem:incr_appr_sssp}) with a super-source attached to $S_i$ is initialized for every level $i \in [0, t)$.\footnote{Namely, we introduce a fake root $s$ and add an edge $(s, v)$ of zero weight, for every candidate center $v \in S_i$. Then, the incremental $(1+\epsilon)$-approximate SSSP algorithm is executed with source $s$.} In Line~\ref{line:algassignment} of~\cref{alg:init-incremental}, we denote by~$\delta_u(v)$ the distance estimate from a candidate center $u \in S_i$ to a vertex~$v \in V$. We remark that even though we cannot access all $\delta_u(v)$ directly, we can efficiently compute the closest candidate center $u \in S_i$ to a vertex $v \in V$ with respect to $\delta_{S_i}(\cdot)$. This follows from the fact that the incremental $(1+\epsilon)$-approximate SSSP algorithm~$\mathcal{A}_i$ has access to the first edge of the corresponding approximate shortest path in~$\tilde{O}(1)$ time; such an incremental SSSP $\mathcal{A}_i$ can be either the deterministic algorithm from~\cref{lem:incr_appr_sssp} or a randomized algorithm against an oblivious adversary (implicitly given in ~\cite{HKN2014hopsets,AnconaHRWW19,LackiN22}). Thus, the assignment in Line~\ref{line:algassignment} can be computed properly and efficiently.

The main reason we introduce $\Retain(\cdot, \cdot)$ (see~\cref{def:retain}) in~\cref{alg:init-incremental} and~\cref{alg:incremental} is because we need
the following lemma and~\cref{lem:size_rel_exec_sets}. In turn, these are required for the correctness analysis of the lower bound in~\cref{lem:lower_bound_opt}, as discussed in~\cref{sec:incr_lower_bound}.

\begin{lemma} \label{lem:size_rel_exec_sets_prepro}
    Immediately after the preprocessing phase, it holds that $|U_{i+1}| = |U_i| - \ceil{\beta \myhs |U_i|}$ for every level~$i \in [0, t)$.
\end{lemma}
\begin{proof}   
    For a fixed level $i \in [0, t)$, the $i$-th radius $\nu_i$ is assigned the value of $\max(\tilde{\nu}_i, \nu_{i-1})$ in~\linecref{algline:incr_max_nu} of~\cref{alg:init-incremental}, which implies that $\nu_i \geq \tilde{\nu}_i$.
    By construction in~\linecref{algline:incr_dijkstra} of~\cref{alg:init-incremental},  
    there are at least $\beta \myhs |U_i|$ vertices in the $i$-th execution set $U_i$ whose distance estimates $\delta_{S_i}(\cdot)$ are at most~$\tilde{\nu}_i$.
    Hence, it holds that $\lvert \ball[S_i, \nu_i, \delta_{S_i}] \cap U_i\rvert \geq \beta \myhs |U_i|$, and because~$\lvert \ball[S_i, \nu_i, \delta_{S_i}] \cap U_i\rvert$ is an integer we deduce that~$\lvert \ball[S_i, \nu_i, \delta_{S_i}] \cap U_i\rvert \geq \ceil{\beta \myhs |U_i|}$. Thus by~\linecref{algline:incr_construct_Bi} of~\cref{alg:init-incremental} and~\cref{def:retain}, it follows that~$|B_i| = \lvert\Retain\bigl(\ball[S_i, \nu_i, \delta_{S_i}] \cap U_i,\, \ceil{\beta \myhs |U_i|}\bigr)\rvert = \ceil{\beta \myhs |U_i|}$. 
    
    Observe that the $(i+1)$-th execution set $U_{i+1}$ is set to $U_i \setminus B_i$ in~\linecref{algline:exec_assign_without_ball_prepro} of~\cref{alg:init-incremental} and that the $i$-th approximate ball $B_i$ is a subset of the $i$-th execution set $U_i$ (due to \linecref{algline:incr_construct_Bi} of~\cref{alg:init-incremental}).
    Therefore, we conclude that $|U_{i+1}| \,=\, |U_i| - |B_i| \,=\, |U_i| - \ceil{\beta \myhs |U_i|}$, as desired.
\end{proof}

\subsubsection{Adversarial Edge Insertions}\label{sec:edge insertions}
Before describing how our incremental bicriteria approximation algorithm for the $(k, z)$-clustering problem on graphs handles adversarial edge insertions, we define the \emph{$i$-th valid radius} in order to guarantee that~\ref{property_rad_3} is satisfied. Recall that \(\ball[{S_i}, r, \delta_{S_i}] \coloneqq \{v \in V \mid \delta_{S_i}(v) \leq r\}\) denotes the closed ball of radius \(r\) around \({S_i}\) using the distance estimates \(\delta_{S_i}(\cdot)\). 

\begin{definition}[$i$-th valid radius] \label{def:ith_valid_rad}
    For a fixed level $i \in [0, t]$, consider the $i$-th execution set $U_i$, the $i$-th candidate set $S_i$, the distance estimates $\delta_{S_i}(\cdot)$, and the $(i-1)$-th radius $\nu_{i-1}$.\footnote{For the definition of the $0$-th valid radius, we set $\nu_{-1} = 0$.}
    Let $\tilde{\nu}_i$ be the smallest power of~$(1+\epsilon)$ such that \(\lvert\ball[S_i, \tilde{\nu}_i, \delta_{S_i}] \cap U_i\rvert \geq \beta \myhs |U_i|\). Then
    $\hat{\nu}_i \coloneqq \max(\tilde{\nu}_i, \nu_{i-1})$ denotes the $i$-th valid radius.
\end{definition}

When an edge $(u^+, v^+)$ is inserted into the graph $G = (V, E, w)$ from the adversary, the incremental bicriteria approximation algorithm proceeds as follows. Initially, this adversarial edge insertion $(u^+, v^+)$ is forwarded to all the incremental $(1+\epsilon)$-approximate SSSP algorithms $\mathcal{A}_i$, where $0 \leq i < t$. The algorithm then finds the smallest level $\ell \in [0, t]$ for which the (updated) $\ell$-th valid radius $\hat{\nu}_\ell$ is less than the $\ell$-th radius~$\nu_\ell$ (i.e., $\hat{\nu}_\ell < \nu_\ell$).\footnote{If no such level $\ell$ exists, then the incremental bicriteria approximation algorithm does nothing.} This step is implemented through the function \textsc{First-Level-Decrease}$()$ (see \linescref{algline:func_first_lvl_decr}{algline:first_level_decr} in~\cref{alg:incremental}).

Next, the function \textsc{Update-DS-Rad-Decr}($\ell$) is called within the function \textsc{First-Level-Decrease}$()$, which updates the data structures at level $\ell$. In the function \textsc{Update-DS-Rad-Decr}($\ell$), the incremental algorithm updates the $\ell$-th radius $\nu_\ell$ to $\hat{\nu}_\ell$ and the $\ell$-th approximate ball $B_\ell$ to an arbitrary subset of~$\ball[S_\ell, \hat{\nu}_\ell, \delta_{S_\ell}] \cap U_\ell$ with size~$\ceil{\beta \myhs |U_\ell|}$
(in~\linescref{algline:set_nu}{algline:incr_construct_Bi_update} of~\cref{alg:incremental}).\footnote{In the analysis (see~\cref{lem:nu_decr_prop}) we show that $\lvert\ball[S_\ell, \hat{\nu}_\ell, \delta_{S_\ell}] \cap U_\ell\rvert \geq \ceil{\beta \myhs |U_\ell|}$, and thus there exists such a subset of size~$\ceil{\beta \myhs |U_\ell|}$.} For each vertex $v \in B_\ell$, the algorithm updates its assignment~$\sigma(v)$ to its nearest candidate center in $S_\ell$ with respect to $\delta_{S_\ell}(\cdot)$ (in~\linecref{algline:v_to_sigma} of~\cref{alg:incremental}). 
The incremental algorithm then assigns the temporary leaking set $Z$ as the union of the $\ell$-th leaking set $Z_\ell$ and the vertices that were in the old $\ell$-th approximate ball but are not in the updated $\ell$-th approximate ball $B_\ell$ (in~\linecref{algline:leaking_Z} of~\cref{alg:incremental}). 
Eventually, the algorithm empties the $\ell$-th leaking set~$Z_\ell$ (in~\linecref{algline:leaking_Zi_empty} of~\cref{alg:incremental}). 

Once the function \textsc{First-Level-Decrease}$()$ has finished, the algorithm updates the $(\ell+1)$-th execution set $U_{\ell+1}$ to $U_\ell \setminus B_\ell$. In turn, the incremental algorithm must verify whether the value of the $(\ell+1)$-th radius $\nu_{\ell+1}$ should also decrease. Note that as the $(\ell+1)$-th execution set $U_{\ell+1}$ may change, the current $(\ell+1)$-th candidate set $S_{\ell+1}$ which was sampled from the old $(\ell+1)$-th execution set, may no longer be a good representative in the new $(\ell+1)$-th execution set $U_{\ell+1}$. To that end, the algorithm proceeds with the \emph{Resampling Phase} in the updated $(\ell+1)$-th execution set $U_{\ell+1}$, as follows. Since $\ell$ denotes the smallest level at which the corresponding $\ell$-th radius $\nu_\ell$ is decreased, let $i \coloneqq \ell + 1$ represent the subsequent levels starting at level $\ell + 1$.

\paragraph{Resampling Phase.} 
First, a supporting candidate set $\tilde{S}_i$ is constructed by sampling each vertex of the new $i$-th execution set $U_i$ independently with probability $\min\Big(\frac{\alpha \myhs k \log n}{|U_i|}, 1\Big)$. Second, the incremental algorithm adds the $i$-th supporting candidate set $\tilde{S}_i$ to the $i$-th candidate set $S_i$, and restarts the $i$-th incremental $(1+\epsilon)$-approximate SSSP algorithm $\mathcal{A}_i$ with a super-source attached to the updated $i$-th candidate set~$S_i$.\footnote{When we say in~\linecref{algline:sssp_update} of~\cref{alg:incremental} ``$\mathcal{A}_i$ provides distance estimates $\delta_u(\cdot)$ for all $u \in S_i$'', we mean the nearest candidate center $u \in S_i$ to a vertex $v \in V$ with respect to $\delta_{S_i}(v)$.} 
Next, the algorithm computes the $i$-th valid radius $\hat{\nu}_i$, using the updated $i$-th candidate set $S_i$ and
the updated distance estimates $\delta_{S_i}(\cdot)$ (see~\linescref{algline:check_if_nu_decreases_2}{algline:set_hat_nu_2} in~\cref{alg:incremental}).
The incremental algorithm then proceeds based on whether the (updated) $i$-th valid radius $\hat{\nu}_i$ is less than the $i$-th radius $\nu_i$, as follows:
\begin{itemize}
    \item If \(\hat{\nu}_i < \nu_i\), then the algorithm 
    updates the data structures as described previously with further modifications (see the function \textsc{Update-DS-Rad-Decr}($i$) in~\linescref{algline:func_update_ds_rad_decr}{algline:update_ds_rad_decr_2} of~\cref{alg:incremental}). In particular, the incremental algorithm updates the $i$-th radius $\nu_i$ to $\hat{\nu}_i$ and the $i$-th approximate ball~$B_i$
    to an arbitrary subset of~$\ball[S_i, \hat{\nu}_i, \delta_{S_i}] \cap U_i$ with size~$\ceil{\beta \myhs |U_i|}$ (in~\linescref{algline:set_nu}{algline:incr_construct_Bi_update} of~\cref{alg:incremental}).\footnote{In the analysis (see~\cref{lem:nu_decr_prop}) we show that $\lvert\ball[S_i, \hat{\nu}_i, \delta_{S_i}] \cap U_i\rvert \geq \ceil{\beta \myhs |U_i|}$, and thus there exists such a subset of size~$\ceil{\beta \myhs |U_i|}$.}
    For each vertex $v \in B_i$, the algorithm updates its assignment $\sigma(v)$ to its nearest candidate center in $S_i$ with respect to $\delta_{S_i}(\cdot)$.
    Subsequently, the incremental algorithm adds to the temporary leaking set $Z$ both the $i$-th leaking set $Z_i$ and the vertices that were in the old $i$-th approximate ball but are not in the updated $i$-th approximate ball $B_i$ (see Figure~\ref{fig:example}). The algorithm then empties the $i$-th leaking set~$Z_i$ and updates the $(i+1)$-th execution set $U_{i+1}$ to $U_i \setminus B_i$.
    
    \item Otherwise if \(\hat{\nu}_i \geq \nu_i\), then the algorithm removes from the $i$-th approximate ball $B_i$ all vertices that no longer belong to the $i$-th execution set $U_i$.\footnote{A vertex that belonged to the old $i$-th execution set is removed from the current $i$-th execution set $U_i$ when it enters a previous $j$-th approximate ball $B_j$ at some level $j < i$.} Observe also that for a level $j: 0 \leq j < i$, some vertices that belonged to the old $j$-th approximate ball before the adversarial edge insertion, may no longer belong to the new $j$-th approximate ball $B_j$
    after the adversarial edge insertion. Note that some of these vertices have been added to the temporary leaking set $Z$. Subsequently, the incremental algorithm  
    adds to the $i$-th leaking set~$Z_i$ an arbitrary subset of the temporary leaking set~$Z$, ensuring that $Z_i \subseteq U_i$ and $|Z_i| = \ceil{\beta \myhs |U_i|} - |B_i|$ (see~\linesscref{algline:Bi_subset_Ui}{algline:count_new_for_Zi}{algline:leaking_Zi} in~\cref{alg:incremental}).\footnote{In the analysis (see~\cref{lem:size_rel_exec_sets,clm:new_vert_in_Z}) we show that such a subset of the temporary leaking set $Z$ exists.}
    The algorithm then removes from the temporary leaking set $Z$ the subset of $Z$ that entered the $i$-th leaking set $Z_i$ (see~\linecref{algline:rem_Zi_from_Z} in~\cref{alg:incremental}) and updates the $(i+1)$-th execution set $U_{i+1}$ to~$U_i \setminus (B_i \cup Z_i)$.
\end{itemize}

\noindent Afterwards, the incremental algorithm increases $i$ by one (i.e., moves to the next level), and continues with the \emph{Resampling Phase} in the updated $i$-th execution set $U_i$ (for the updated level $i$) as described earlier. 
The Resampling Phase terminates when the size of the current $i$-th execution set $U_i$ is at most $\alpha \myhs k \log n$. Then, the algorithm updates the value of the last level $t$ to the final value of $i$ (see~\linecref{algline:t_last_level} in~\cref{alg:incremental}). Finally, as a concluding step for the adversarial edge insertion $(u^+, v^+)$ into the graph $G$, the incremental bicriteria approximation algorithm trivially updates the data structures at level $t$ (in~\linecref{algline:last_level_assign} of~\cref{alg:incremental}) and assigns the maintained bicriteria approximate solution $S$ as $\bigcup_{i=0}^t S_i$ (in~\linecref{algline:assign_solution} of~\cref{alg:incremental}). 

\vspace{5em}
\begin{figure}[H]
    \centering
    \begin{subfigure}{.45\textwidth}
      \centering
      \includegraphics[width=0.7\linewidth]{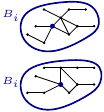}
    \end{subfigure}%
    \begin{subfigure}{.45\textwidth}
      \centering
      \includegraphics[width=0.7\linewidth]{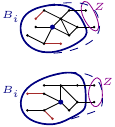}
    \end{subfigure}
    \vspace{1em}
    \caption{The blue thick regions depict the $i$-th approximate ball $B_i$ for a level $i \in [0, t]$, and the larger blue vertices represent the $i$-th candidate set $S_i$. The black lines indicate distances. In the right figure, 
     the brown lines and dots depict new shorter distances of vertices that enter $B_i$ due to the adversarial edge insertion.
     \\[0.5em]
     \underline{Left figure}: The $i$-th approximate ball before the adversarial edge insertion. \\[0.5em]
     \underline{Right figure}: After the adversarial edge insertion, some vertices enter the $i$-th approximate ball $B_i$ (brown vertices). Since the $i$-th radius $\nu_i$ is decreased (the dashed blue region indicates the old radius), some vertices enter the temporary leaking set $Z$.}
    \label{fig:example}
\end{figure}

\begin{algorithm}[H]\footnotesize
\algnewcommand{\LineComment}[1]{\State \(\triangleright\) #1}
\algdef{SE}[DOWHILE]{Do}{doWhile}{\algorithmicdo}[1]{\algorithmicwhile\ #1}%
\caption{Incremental Bicriteria Approximation Algorithm}\label{alg:incremental}

\begin{algorithmic}[1]

\Require{Graph $G$ and positive integer $k$ (problem input); 
positive constants $\alpha,\beta$ (MP-bi parameters);
positive constant $\epsilon$ (MP-bi variant parameter, used also as SSSP approximation parameter)}
\Ensure{Bicriteria approximate solution $S$ and bicriteria approximate assignment $\sigma$}

\LineComment{The algorithm has global access to all $U_i, S_i, \nu_i, \mathcal{A}_i, B_i, Z_i, Z, \sigma(\cdot)$}

\Statex
\Function{Update-DS-Rad-Decr}{$i$} \label{algline:func_update_ds_rad_decr}
    \State $\nu_i \gets \hat{\nu}_i$ \label{algline:set_nu}
    
    \State $B'_i \gets B_i$ \label{algline:prev_ball}
    
    \State $B_i \gets \Retain\bigl(\ball[S_i, \nu_i, \delta_{S_i}] \cap U_i,\, \ceil{\beta \myhs |U_i|}\bigr)$ \label{algline:incr_construct_Bi_update}
    
    \State For every $v \in B_i:$ $\sigma(v) \gets \argmin_{u \in S_i} \delta_u(v)$ \label{algline:v_to_sigma}
    
    \vspace{0.2em}
    \State $Z \gets Z \cup Z_i  \cup (B'_i \setminus B_i) $ \label{algline:leaking_Z}

    \State $Z_i \gets \emptyset$ \label{algline:leaking_Zi_empty}
\EndFunction

\Statex
\Function{First-Level-Decrease}{\hspace{0.5pt}} \label{algline:func_first_lvl_decr}
    \State $i \gets -1$
    \vspace{0.2em}
    \Do
        \State $i \gets i + 1$ \label{algline:incr_level}

        \State $\tilde{\nu}_i \gets \min_{j \in \mathbb{Z}^{\geq 0}} \{(1+\epsilon)^j \mid \lvert \ball[S_i, (1+\epsilon)^j, \delta_{S_i}] \cap U_i\rvert \geq \beta \myhs |U_{i}|\}$ \label{algline:check_if_nu_decreases_1}
    
        \State $\hat{\nu}_i \gets \max(\tilde{\nu}_i, \nu_{i-1})$ \label{algline:set_hat_nu_1}
    \doWhile{{$i < t \textbf{ and }\hat{\nu}_i \geq \nu_i$} \label{algline:while_find_small_nu}}

    \vspace{0.2em}
    \If{$i < t$} \label{algline:nu_decreases_first}
        \State $Z \gets \emptyset$
        \State \Call{Update-DS-Rad-Decr}{$i$} \label{algline:update_ds_rad_decr_1}
    \EndIf

    \vspace{0.2em}
    \State \Return $i$
\EndFunction

\Statex
\Function{Edge-Insertion}{$u^+, v^+$}
\vspace{0.1em}
\LineComment{If a $\nu_i$ is undefined, then assume that $\nu_i = nW + 1$} \label{algline:comment_nu_largest}
\LineComment{If a set is undefined, then assume that it is the empty set}

\vspace{0.2em}
\State Insert $(u^+, v^+)$ into $G$
\vspace{0.2em}

\For{$i: 0 \leq i < t$}
    \State $\mathcal{A}_i.\textit{insert}(u^+, v^+)$
\EndFor

\vspace{0.2em}
\State $i \gets \Call{First-Level-Decrease}$ \label{algline:first_level_decr}

\vspace{0.2em}
\State $U_{i+1} \gets U_i \setminus B_i$ \label{algline:next_exec_set_1}
    
\State $i \gets i + 1$ \label{algline:level_incr_1}
\vspace{0.2em}
\If{$i > t$}
    \State \textbf{exit} from \Call{Edge-Insertion}{\hspace{0.5pt}}
\EndIf
\vspace{0.2em}

\LineComment{Entering the Resampling Phase}\label{algline:enteringresamplingphase}
\vspace{0.2em}
\While{$|U_i| > \alpha k \log n$} \label{algline:beg_resample}
    \State Construct $\tilde{S}_i$ by sampling each $v \in U_i$ independently with probability $\min\Big(\frac{\alpha \myhs k \log n}{|U_i|}, 1\Big)$ \label{algline:construct_support_candidate_set}

    \State $S_i \gets S_i \cup \tilde{S}_i$ \label{algline:extend_candidate_set}
    \State $\mathcal{A}_i.\textit{initialize}(G, S_i)$ \label{algline:sssp_update} 
    \Comment{$\mathcal{A}_i$ provides distance estimates $\delta_{S_i}(\cdot)$ and $\delta_u(\cdot)$ for all $u \in S_i$}
    \vspace{0.2em}

    \State $\tilde{\nu}_i \gets \min_{j \in \mathbb{Z}^{\geq 0}} \{(1+\epsilon)^j\mid \lvert \ball[S_i, (1+\epsilon)^j, \delta_{S_i}]\cap U_i\rvert \geq \beta \myhs |U_i|\}$ \label{algline:check_if_nu_decreases_2}
    
    \State $\hat{\nu}_i \gets \max(\tilde{\nu}_i, \nu_{i-1})$ \label{algline:set_hat_nu_2}
    \vspace{0.2em}
        
    \If{$\hat{\nu}_i < \nu_i$} \label{algline:nu_decreases}
        \State \Call{Update-DS-Rad-Decr}{$i$} \label{algline:update_ds_rad_decr_2}

    \Else
        \State $B_i \gets B_i \cap U_i,\, Z_i \gets Z_i \cap U_i$
        \label{algline:Bi_subset_Ui}
        \State $Z' \gets \Retain\bigl(Z \cap U_i,\, \ceil{\beta \myhs |U_i|} - |B_i| - |Z_i|\bigr)$ \label{algline:count_new_for_Zi}
        \State $Z_i \gets (Z_i \cup Z') \cap U_i$ \label{algline:leaking_Zi}
        \State $Z \gets Z \setminus Z'$ \label{algline:rem_Zi_from_Z}
    \EndIf

    \vspace{0.2em}
    \State $U_{i+1} \gets U_i \setminus (B_i \cup Z_i)$ \label{algline:next_exec_set_2}
    \State $Z \gets Z \cap U_{i+1}$ \label{algline:Z_intersect_subseq_Ui}
    
    \State $i \gets i + 1$ \label{algline:level_incr}
\EndWhile

\vspace{0.2em}
\State $t \gets i$ \label{algline:t_last_level}
\State $S_t \gets U_t$, $B_t \gets U_t$, $Z_t \gets \emptyset$, $Z \gets \emptyset$, $\nu_t \gets \nu_{t-1}, \forall v \in U_t: \sigma(v) \gets v$ \label{algline:last_level_assign}
\State $S \gets \bigcup_{i=0}^{t} S_i$ \label{algline:assign_solution}

\EndFunction
\end{algorithmic}
\end{algorithm}

\paragraph{On the computation of $\tilde{\nu}_i$.}
We emphasize that the value $\tilde{\nu}_i$ in~\linescref{algline:check_if_nu_decreases_1}{algline:check_if_nu_decreases_2} of~\cref{alg:incremental} can be computed during the update procedure of the $i$-th incremental $(1+\epsilon)$-approximate SSSP algorithm $\mathcal{A}_i$. This is achieved by extending $\mathcal{A}_i$ to track the necessary quantities within its update procedure, without incurring additional asymptotic cost. The reason this is possible is that, based on~\cref{lem:incr_appr_sssp}, $\mathcal{A}_i$ detects and reports a change whenever the distance estimate $\delta_{S_i}(v)$ is updated for a vertex $v \in V$.
As a result, the cost of maintaining $\tilde{\nu}_i$ is included in the update time of $\mathcal{A}_i$.\footnote{Observe that in~\linecref{algline:check_if_nu_decreases_1} of~\cref{alg:incremental}, the value $\tilde{\nu}_t$ at the last level $t$ is trivially zero.
}

\subsection{Analysis of the Incremental Bicriteria Approximation Algorithm} \label{sec:incr_bicr_analysis}
When an edge is inserted into the graph $G = (V, E, w)$ from the adversary, multiple pairwise distances between vertices may change significantly, requiring the incremental algorithm to update the maintained data structures. To distinguish between the states of the input graph $G$ and our incremental algorithm before and after an adversarial edge insertion, we use $\old$ as a superscript for the state before the edge insertion, and we omit the $\old$ superscript for the state after the edge insertion.\footnote{Whenever we say ``after an (adversarial) edge insertion'', we mean after the algorithm has completed all updates to its data structures.}

Our goal in this section is to prove Theorem~\ref{thm:incr_alg}, by analyzing the incremental bicriteria approximation~\cref{alg:incremental}. The analysis is divided into four subsections, as outlined below.
First, in~\cref{sec:size_exec_sets} we analyze the amortized update time of our incremental bicriteria approximation algorithm. Next, in~\cref{sec:incr_upper_bound} we derive an upper bound on the cost of the maintained bicriteria approximate solution $S$.
Then, in~\cref{sec:incr_lower_bound} we establish results concerning a lower bound on the optimal $(k, z)$-clustering cost which is associated with the upper bound. Finally, in~\cref{sec:size_finish_proof} we provide an upper bound on the size of the bicriteria approximate solution and we
conclude the proof of Theorem~\ref{thm:incr_alg}.

\subsubsection{Analysis of the Update Time} \label{sec:size_exec_sets}
\begin{lemma} \label{lem:Bi_Zi_subsets_Ui}
    After an adversarial edge insertion, it holds that $B_i \subseteq U_i$ and $Z_i \subseteq U_i$ for every level~$i \in [0, t]$. In other words, both the $i$-th approximate ball and the $i$-th leaking set are subsets of the $i$-th execution set.
\end{lemma}
\begin{proof}
    Let $\ell$ be the smallest level for which the corresponding $\ell$-th radius $\nu_\ell$ is decreased, as returned in~\linecref{algline:first_level_decr}.\footnote{Note that for the corner case where $\ell$ equals $t$, nothing changes.\label{footnote:last_level_t}} For every level $i: 0 \leq i \leq \ell-1$, the corresponding
    approximate balls, leaking sets, and execution sets are not modified, and thus the statement holds using an induction argument on the number of edge insertions. For the last level $t$, the statement holds trivially due to~\linecref{algline:last_level_assign}. Regarding a fixed level~$i: \ell \leq i < t$, consider the two cases depending on whether the value of the $i$-th radius $\nu_i$ is decreased:
    \begin{itemize}
        \item Assume that after the adversarial edge insertion, the value of the $i$-th radius decreases (i.e., $\nu_i = \hat{\nu}_i < \nu_i^\old$). In this case, observe that for $i = \ell$ the function \textsc{Update-DS-Rad-Decr}($i$) is executed in~\linecref{algline:update_ds_rad_decr_1}, and for $i \neq \ell$
        it is executed in~\linecref{algline:update_ds_rad_decr_2}. In both situations, the $i$-th approximate ball $B_i$ is updated only in~\linecref{algline:incr_construct_Bi_update} to contain vertices solely from the $i$-th execution set $U_i$. Also, the $i$-th leaking set $Z_i$ becomes empty in~\linecref{algline:leaking_Zi_empty}.

        \item Assume that after the adversarial edge insertion, the value of the $i$-th radius
        does not decrease (i.e., $\hat{\nu}_i \geq \nu_i^\old = \nu_i$). In this case, the $i$-th approximate ball $B_i$ is updated only in~\linecref{algline:Bi_subset_Ui}, and
        the $i$-th leaking set $Z_i$ is updated only in~\linescref{algline:Bi_subset_Ui}{algline:leaking_Zi}, ensuring that both $B_i$ and $Z_i$ remain subsets of the $i$-th execution set $U_i$.
    \end{itemize}
    
    The $(i+1)$-th execution set $U_{i+1}$ is updated only in~\linescref{algline:next_exec_set_1}{algline:next_exec_set_2}, and since $i < t$, either \linescref{algline:incr_construct_Bi_update}{algline:leaking_Zi_empty} or \linescref{algline:Bi_subset_Ui}{algline:leaking_Zi} are executed to assign $B_{i+1}$ and $Z_{i+1}$ respectively (or~\linecref{algline:last_level_assign} if $i + 1$ equals $t$). Therefore by applying this argument for all levels $i$, we can conclude that both the $i$-th approximate ball and the
    $i$-th leaking set are always subsets of the $i$-th execution set.
\end{proof}

\begin{lemma} \label{lem:v_Bi_Zi}
    After an adversarial edge insertion, consider a fixed level $i: 0 \leq i < t$ and an arbitrary vertex~$v \in U_i \setminus U_{i+1}$. Then the vertex $v$ belongs either to the $i$-th approximate ball $B_i$ or to the $i$-th leaking set~$Z_i$.
    Moreover, if the vertex $v$ belongs to the $i$-th leaking set $Z_i$, then before the edge insertion it belonged to a set~$U_j^\old \setminus U_{j+1}^\old$ with~$j \leq i$.
\end{lemma}
\begin{proof}
    The $(i+1)$-th execution set $U_{i+1}$ is updated only in~\linescref{algline:next_exec_set_1}{algline:next_exec_set_2}, where it is set to $U_i \setminus (B_i \cup Z_i)$.\footnote{Notice that in~\linecref{algline:next_exec_set_1}, we have $Z_i = \emptyset$.\label{footnote:Zi_empty}} Hence, any vertex $v \in U_i \setminus U_{i+1}$ must belong to $B_i \cup Z_i$, as required. For the second part of the statement, consider an arbitrary vertex $v \in Z_i$. If the vertex $v$ was part of the old $i$-th leaking set $Z_i^\old$ before the adversarial edge insertion, then by~\cref{lem:Bi_Zi_subsets_Ui} it holds that $v \in U_i^\old$. Thus, based on the way $U_{i+1}^\old$ was updated in~\linecref{algline:next_exec_set_2} it follows that $v \in U_i^\old \setminus U_{i+1}^\old$, as needed.
    
    For this reason, assume that the vertex $v$ was not in the old $i$-th leaking set $Z^\old_i$ but enters the new $i$-th leaking set $Z_i$ after the adversarial edge insertion. The vertex $v$ can enter the $i$-th leaking set $Z_i$ only in~\linecref{algline:leaking_Zi}, which means that the vertex $v$ must have entered the temporary leaking set $Z$. Observe that the temporary leaking set $Z$ incorporates additional vertices only in~\linecref{algline:leaking_Z} within \textsc{Update-DS-Rad-Decr}($\cdot$). Since the levels are scanned in an increasing way (see~\linescref{algline:level_incr_1}{algline:level_incr}),  
    the vertex $v$ should belong before the adversarial edge insertion to a set $B^\old_j \cup Z^\old_j$ with $j \leq i$ (according to~\linecref{algline:leaking_Z}). 
    By~\cref{lem:Bi_Zi_subsets_Ui} it holds that~$v \in U^\old_j$, and by construction the $(j+1)$-th execution set $U^\old_{j+1}$ was updated only in~\linescref{algline:next_exec_set_1}{algline:next_exec_set_2} where it was set to $U^\old_j \setminus (B^\old_j \cup Z^\old_j)$. Therefore, it follows that $v \in U_j^\old \setminus U_{j+1}^\old$ concluding the claim.
\end{proof}

\begin{lemma} \label{lem:Bi_Zi_disjoint}
    After an adversarial edge insertion, it holds that $B_i \cap Z_i = \emptyset$ for every level $i \in [0, t]$. In other words, the $i$-th approximate ball and the $i$-th leaking set are disjoint. 
\end{lemma}
\begin{proof}
    Consider a fixed level $i \in [0, t)$, since at the last level $t$ the claim holds trivially due to~\linecref{algline:last_level_assign}. If the value of the $i$-th radius $\nu_i$ decreases (i.e., $\nu_i = \hat{\nu}_i < \nu_i^\old$) then the claim holds trivially, because 
    the $i$-th leaking set $Z_i$ becomes empty in~\linecref{algline:leaking_Zi_empty} within \textsc{Update-DS-Rad-Decr}($i$). Otherwise, assume that the value of the $i$-th radius does not decrease (i.e., $\hat{\nu}_i \geq \nu_i^\old = \nu_i$).
    In this case, the $i$-th approximate ball $B_i$ is updated to $B_i^\old \cap U_i$ in~\linecref{algline:Bi_subset_Ui}, and
    the updated $i$-th leaking set $Z_i$ receives new vertices only through the temporary leaking set $Z$ in~\linecref{algline:leaking_Zi} (see also~\linecref{algline:count_new_for_Zi}). 
    
    As described in the proof of~\cref{lem:v_Bi_Zi}, the temporary leaking set $Z$ incorporates additional vertices only in~\linecref{algline:leaking_Z} within \textsc{Update-DS-Rad-Decr}($\cdot$), and the levels are scanned in an increasing way (see~\linescref{algline:level_incr_1}{algline:level_incr}). In addition, since the function \textsc{Update-DS-Rad-Decr}($j$) is invoked only when the $j$-th radius $\nu_j$ decreases, and the $i$-th radius $\nu_i$ does not decrease, the
    temporary leaking set $Z$ incorporates additional vertices only at levels~$j < i$. Notice that $(B_j^\old \cup Z_j^\old) \cap B_i^\old = \emptyset$ for any levels $j < i$, by construction in~\linescref{algline:next_exec_set_1}{algline:next_exec_set_2} and~\cref{lem:Bi_Zi_subsets_Ui}.
    In turn, this implies that $B_i^\old \cap Z = \emptyset$, and so the $i$-th leaking set $Z_i$ receives new vertices that do not belong to $B_i^\old$. Using an induction argument on the number of adversarial edge insertions, we can infer that $B_i^\old \cap Z_i^\old = \emptyset$, and thus $B_i^\old \cap Z_i = \emptyset$. Therefore as $B_i \subseteq B_i^\old$, it follows that $B_i \cap Z_i = \emptyset$, as required.
\end{proof}

\begin{lemma} \label{lem:nu_decr_prop}
    After an adversarial edge insertion, consider a level $i: 0 \leq i < t$ for which the algorithm decreases the $i$-th radius (i.e., $\nu_i < \nu_i^\old$). Then it holds that~$|B_i| = \ceil{\beta \myhs |U_i|}$ and~$|U_{i+1}| = |U_i| - |B_i|$.
\end{lemma}
\begin{proof}
    Let $\ell$ be the smallest level for which the corresponding $\ell$-th radius $\nu_\ell$ is decreased, as returned in~\linecref{algline:first_level_decr} of~\cref{alg:incremental}.
    Note that for $i = \ell$ the function \textsc{Update-DS-Rad-Decr}($i$) is executed in~\linecref{algline:update_ds_rad_decr_1}, and for~$i \neq \ell$ it is executed in~\linecref{algline:update_ds_rad_decr_2}. In both cases, the $i$-th radius $\nu_i$ is updated to the $i$-th valid radius~$\hat{\nu}_i$ in~\linecref{algline:set_nu}, which in turn is determined in~\linecref{algline:set_hat_nu_1} or in~\linecref{algline:set_hat_nu_2}. Notice that the value of the $i$-th valid radius $\hat{\nu}_i$ is at least the value of $\tilde{\nu}_i$ determined in~\linecref{algline:check_if_nu_decreases_1} or in~\linecref{algline:check_if_nu_decreases_2} respectively. By construction in both~\linescref{algline:check_if_nu_decreases_1}{algline:check_if_nu_decreases_2},  
    there are at least $\beta \myhs |U_i|$ vertices in the $i$-th execution set $U_i$ whose distance estimates~$\delta_{S_i}(\cdot)$ are at most~$\tilde{\nu}_i$. As a result, since $\nu_i = \hat{\nu}_i \geq \tilde{\nu}_i$ it holds that $\lvert \ball[S_i, \hat{\nu}_i, \delta_{S_i}] \cap U_i\rvert \geq \beta \myhs |U_i|$, and because~$\lvert \ball[S_i, \hat{\nu}_i, \delta_{S_i}] \cap U_i\rvert$ is an integer we deduce that $\lvert \ball[S_i, \hat{\nu}_i, \delta_{S_i}] \cap U_i\rvert \geq \ceil{\beta \myhs |U_i|}$.
    Hence based on~\linecref{algline:incr_construct_Bi_update} and~\cref{def:retain}, it follows that~$|B_i| = \lvert\Retain\bigl(\ball[S_i, \nu_i, \delta_{S_i}] \cap U_i,\, \ceil{\beta \myhs |U_i|}\bigr)\rvert = \ceil{\beta \myhs |U_i|}$.
   
    Regarding the second part of the statement, observe that the execution set $U_{i+1}$ is updated only in~\linescref{algline:next_exec_set_1}{algline:next_exec_set_2}, where it is set to $U_i \setminus (B_i \cup Z_i)$.\footref{footnote:Zi_empty} 
    As the $i$-th radius is decreased,
    the $i$-th leaking set $Z_i$ becomes empty due to~\linecref{algline:leaking_Zi_empty}. Moreover by~\cref{lem:Bi_Zi_subsets_Ui}, the $i$-th approximate ball $B_i$ is a subset of the $i$-th execution set $U_i$, and thus
    we can conclude that $|U_{i+1}|$ is equal to $|U_i| - |B_i|$, as desired.
\end{proof}

\begin{lemma} \label{lem:size_rel_exec_sets}
    After an adversarial edge insertion, it holds that $|U_i| = |U_i^\old|$ and $|U_i| = |U_{i-1}| - \ceil{\beta \myhs |U_{i-1}|}$ for every level $i \in [0, t]$.\footnote{\label{ftnote:def_U_-1} For $i=0$, we choose $|U_{-1}|$ so that $|U_0| = |U_{-1}| - \ceil{\beta \myhs |U_{-1}|}$.}
\end{lemma}
\begin{proof}
    The proof consists of two inductions; the outer induction is on the number of adversarial edge insertions,
    and the nested induction is on the number of levels. The base case of the outer induction holds due to the preprocessing phase (and~\cref{ftnote:def_U_-1}), as established in~\cref{lem:size_rel_exec_sets_prepro}. The base case of the nested induction is for~$i = 0$, and the first part of the statement $|U_0| = |U_0^\old| = n$ holds by the construction of~\cref{alg:incremental} (with the second part of the statement  from~\cref{ftnote:def_U_-1}).

    Regarding the induction step of this double induction, we prove the
    statement for the $(i+1)$-th level after the adversarial edge insertion (where $i \in [0, t)$), assuming that:
    (1) the statement is true for all levels before the edge insertion (outer induction hypothesis),
    and (2) the statement is true for the $i$-th level after the edge insertion (nested induction
    hypothesis). To analyze the $(i+1)$-th level after the adversarial edge insertion, we consider two cases based on whether the value of the $i$-th radius $\nu_i$ is decreased:
    \begin{itemize}
        \item Assume that after the edge insertion, the value of the $i$-th radius
        decreases (i.e., $\nu_i < \nu_i^\old$). In this case, the algorithm ensures that $|B_i| = \ceil{\beta \myhs |U_i|}$ and $|U_{i+1}| = |U_i| - |B_i|$, as stated in~\cref{lem:nu_decr_prop}. This implies that $|U_{i+1}| = |U_i| - \ceil{\beta \myhs |U_i|}$, as needed for the second part of the statement. Using the nested induction hypothesis for the first part of the statement, we deduce that $|U_{i+1}| =|U^\old_i| - \ceil{\beta \myhs |U^\old_i|}$. Using the outer induction hypothesis for the second part of the statement, we conclude that $|U_{i+1}| \,=\,|U^\old_i| - \ceil{\beta \myhs |U^\old_i|} \,=\, |U^\old_{i+1}|$, as required for the first part of the statement.

        \item Assume that after the edge insertion, the value of the $i$-th radius $\nu_i$ does not decrease (i.e., $\nu_i \geq \nu_i^\old$). In this case, the $i$-th approximate ball $B_i$ is updated to $B_i^\old \cap U_i$ in~\linecref{algline:Bi_subset_Ui}, and the $i$-th leaking set~$Z_i$ is updated to $(Z_i^\old \cup Z') \cap U_i$ where
        $Z' \coloneqq \Retain\bigl(Z \cap U_i,\, \ceil{\beta \myhs |U_i|} - |B_i| - |Z_i^\old \cap U_i|\bigr)$
        according to~\linesscref{algline:Bi_subset_Ui}{algline:count_new_for_Zi}{algline:leaking_Zi}. Our aim is to show that $|Z'| \,=\, \ceil{\beta \myhs |U_i|} - |B_i| - |Z_i^\old \cap U_i|$ in the next paragraph, where we use the following auxiliary claim.\footnote{Here, $Z$ refers to the current temporary leaking set at level $i$, rather than the updated set obtained after~\linecref{algline:rem_Zi_from_Z} in~\cref{alg:incremental}. More precisely, it refers to $Z^\old$, but we write $Z$ for ease of reading.}
        
        \begin{claim} \label{clm:new_vert_in_Z}
            Every vertex $v \in U_i \setminus U^\old_i$ belongs to the temporary leaking set $Z$ at the moment the $i$-th execution set $U_i$ is updated in~\linescref{algline:next_exec_set_1}{algline:next_exec_set_2} of~\cref{alg:incremental}, and $v$ still belongs to $Z$ at the moment~$Z \cap U_i$ is used within~$\Retain\bigl(Z \cap U_i,\, \ceil{\beta \myhs |U_i|} - |B_i| - |Z_i^\old \cap U_i|\bigr)$ in~\linecref{algline:count_new_for_Zi} of~\cref{alg:incremental}.
        \end{claim}
        \begin{proof}
            Consider a vertex $v \in V$ such that $v \in U_i$ and $v \notin U^\old_i$. Since the vertex $v$ was not part of the old $i$-th execution set $U^\old_i$, there must exist a level $j$ such that $v \in U^\old_j \setminus U^\old_{j+1}$ and $j < i$, as dictated by the construction of the execution sets. Based on~\cref{lem:v_Bi_Zi}, the vertex $v$ belonged to the set~$B^\old_j \cup Z^\old_j$.
            Since the vertex $v$ enters the updated $i$-th execution set $U_i$ after the adversarial edge insertion,
            based on~\linescref{algline:next_exec_set_1}{algline:next_exec_set_2} it must have left the union of the $j$-th approximate ball and the $j$-th leaking set. In other words it holds that $v \notin B_j \cup Z_j$ and $v \in U_i$, which occurs (based on the construction of the execution sets) only when the $j$-th radius decreases
            and the function \textsc{Update-DS-Rad-Decr}($j$) is invoked. In turn due to~\linecref{algline:leaking_Z}, the vertex $v$ must have entered the temporary leaking set $Z$ at level $j < i$.
            Since the vertex $v$ belongs to the updated $i$-th execution set $U_i$, the construction of~\cref{alg:incremental} in~\linesscref{algline:rem_Zi_from_Z}{algline:next_exec_set_2}{algline:Z_intersect_subseq_Ui} (i.e., each subsequent execution set is a subset of the previous one and previous leaking sets are disjoint from subsequent execution sets)
            ensures that~$v$ remains in $Z$ when $Z \cap U_i$ is utilized in~\linecref{algline:count_new_for_Zi}.
        \end{proof}
    
        \subparagraph{Equality for $|Z'|$.}
        Since~$|U_i| = |U_i^\old|$ by the nested induction hypothesis for the first part of the statement, the number of vertices that leave the union of the $i$-th approximate ball and the $i$-th leaking set must be at most the number of 
        new vertices that enter the $i$-th execution set $U_i$ (using~\cref{lem:Bi_Zi_subsets_Ui}). Based on~\cref{clm:new_vert_in_Z}, any new vertex entering the updated $i$-th execution set $U_i$ belongs to the temporary leaking set $Z$ at the moment~$U_i$ is updated and $Z \cap U_i$ is used.
        Therefore, by~\cref{clm:new_vert_in_Z,lem:Bi_Zi_subsets_Ui} it holds that
        $\lvert(B_i^\old \setminus B_i) \cup (Z_i^\old \setminus Z_i)\rvert \,\leq\, |Z \cap U_i|$. Specifically, using~\cref{lem:Bi_Zi_disjoint} it follows that~$\lvert B_i^\old \setminus B_i\rvert + \lvert Z_i^\old \setminus Z_i\rvert \,\leq\, |Z|$, where $Z \subseteq U_i$ due to~\linecref{algline:Z_intersect_subseq_Ui}.

        By the outer induction hypothesis for the second part of the statement, we have $|U_{i+1}^\old| = |U_i^\old| - \ceil{\beta \myhs |U_i^\old|}$. Hence, by~\cref{lem:v_Bi_Zi,lem:Bi_Zi_subsets_Ui} and as $U_{i+1}^\old \subseteq U_i^\old$ by construction in~\linescref{algline:next_exec_set_1}{algline:next_exec_set_2}, we can infer that $|B_i^\old \cup Z_i^\old| = \ceil{\beta \myhs |U_i^\old|}$. Specifically, by~\cref{lem:Bi_Zi_disjoint} we deduce that $|B_i^\old| + |Z_i^\old| = \ceil{\beta \myhs |U_i^\old|}$. Notice that
        $|B_i^\old| = |B_i^\old \setminus B_i| + |B_i|$ and $|Z_i^\old| = |Z_i^\old \setminus Z_i| + |Z_i^\old \cap U_i|$ by the construction of~\cref{alg:incremental} in~\linescref{algline:Bi_subset_Ui}{algline:leaking_Zi}, and combining this with the two expressions:
        \begin{align*}
            \lvert B_i^\old \setminus B_i\rvert \,+\, \lvert Z_i^\old \setminus Z_i\rvert \;&\leq\; |Z|\;\; \text{and} \\
            |B_i^\old| \,+\, |Z_i^\old| \;&=\; \ceil{\beta \myhs |U_i^\old|},
        \end{align*}
        implies that:
        \begin{align*}
            |Z| &\;\geq\; \lvert B_i^\old \setminus B_i\rvert \,+\, \lvert Z_i^\old \setminus Z_i\rvert  \\
            &\;=\; \lvert B_i^\old \rvert \,-\, |B_i| \,+\, |Z_i^\old| \,-\, |Z_i^\old \cap U_i| \\
            &\;=\; \ceil{\beta \myhs |U_i^\old|} \,-\, |B_i| \,-\, |Z_i^\old \cap U_i|.
        \end{align*}  
        Using the nested induction hypothesis for the first part of the statement, we obtain $|Z| \,\geq\, \ceil{\beta \myhs |U_i|} \,-\, |B_i| \,-\, |Z_i^\old \cap U_i|$. 
        Thus, since $Z' = \Retain\bigl(Z \cap U_i,\, \ceil{\beta \myhs |U_i|} - |B_i| - |Z_i^\old \cap U_i|\bigr)$ and $Z \subseteq U_i$, according to~\cref{def:retain} we get that $|Z'| = \ceil{\beta \myhs |U_i|} - |B_i| - |Z_i^\old \cap U_i|$, as needed for the equality for $|Z'|$.
        \medskip

        With the desired equality for $|Z'|$ established, observe that the $i$-th leaking set~$Z_i$ is updated to $(Z_i^\old \cup Z') \cap U_i$ in~\linecref{algline:leaking_Zi}. Since the $i$-th radius $\nu_i$ does not decrease, the temporary leaking set $Z$ has incorporated additional vertices only at levels $j < i$ (see~\linecref{algline:leaking_Z}). Notice that $(B_j^\old \cup Z_j^\old) \cap Z_i^\old = \emptyset$ for any levels $j < i$, by construction in~\linescref{algline:next_exec_set_1}{algline:next_exec_set_2} and~\cref{lem:Bi_Zi_subsets_Ui}. Hence as $Z' \subseteq Z$ by~\linecref{algline:count_new_for_Zi}, the set $Z'$ is disjoint from the set $Z_i^\old \cap U_i$, and thus it holds that $|Z_i| = |Z_i^\old \cap U_i| + |Z'|$. In turn, it follows that $|Z_i| = \ceil{\beta \myhs |U_i|} -  |B_i|$. 
        
        The $(i+1)$-th execution set~$U_{i+1}$ is set to $U_i \setminus (B_i \cup Z_i)$ in~\linecref{algline:next_exec_set_2}, and since~$B_i \cup Z_i \subseteq U_i$ by~\cref{lem:Bi_Zi_subsets_Ui} and $B_i \cap Z_i = \emptyset$ by~\cref{lem:Bi_Zi_disjoint} it holds that~$|U_{i+1}| = |U_i| - |B_i \cup Z_i| = |U_i| - (|B_i| + |Z_i|)$. Therefore, using the induction hypotheses it follows that $|U_{i+1}| = |U_i| - \ceil{\beta \myhs |U_i|} = |U^\old_i| - \ceil{\beta \myhs |U^\old_i|} = |U^\old_{i+1}|$, as required for both parts of the statement.
    \end{itemize}
\end{proof}

The next corollary establishes that the incremental bicriteria approximation~\cref{alg:incremental} maintains~$O(\log n)$ levels throughout its execution.

\begin{corollary} \label{cor:number_of_levels}
    After an adversarial edge insertion, the value of the last level $t$ is at most $O(\log n)$.
\end{corollary}
\begin{proof}
    Based on~\cref{lem:size_rel_exec_sets}, after an adversarial edge insertion we have $|U_i| \,=\, |U_{i-1}| - \ceil{\beta \myhs |U_{i-1}|} \,\leq\, |U_{i-1}| - \beta \myhs |U_{i-1}| \,=\, (1 - \beta)|U_{i-1}|$ for every level $i \in [1, t]$. In turn, we always have $|U_{t-1}| \leq (1-\beta)^{t-1} |U_0|$.
    Recall that $|U_0| = n$, $|U_{t-1}| \geq 1$, and $0 < (1-\beta) < 1$. Hence, it follows that $(1-\beta)^{t-1} \cdot n \geq 1 \; \iff \; t~\leq~\log_{(1-\beta)} \frac{1}{n} + 1 \; \iff \; t \leq \frac{-\log n}{\log (1-\beta)} + 1$. 
    Therefore as $\beta$ is a fixed constant,  the value of $t$ is always at most $O(\frac{\log n}{\beta}) \,=\, O(\log n)$.\footnote{Note that we can assume $n > 1$, as otherwise the $(k, z)$-clustering instance can be trivially solved.}
\end{proof}

\begin{observation} \label{obs:nu_only_decr_power_(1+eps)}
    After an adversarial edge insertion, the value of the $i$-th radius $\nu_i$ either decreases or remains the same. In other words, it holds that $\nu_i \leq \nu^\old_i$ for every level $i \in [0, t]$ (i.e., \ref{property_rad_2} is satisfied). Moreover, the value of each $i$-th radius $\nu_i$ is a power of $(1+\epsilon)$ (i.e., \ref{property_rad_1} is satisfied).
\end{observation}
\begin{proof} 
    At the last level $t$, we have $\nu_t = \nu_{t-1}$ due to~\linecref{algline:last_level_assign}.
    For a fixed level $i: 0 \leq i < t$, observe that the $i$-th radius $\nu_i$ is updated only in~\linecref{algline:set_nu} inside the
    the function \textsc{Update-DS-Rad-Decr}($i$), where it is assigned the value of the $i$-th valid radius $\hat{\nu}_i$. This update occurs only if 
    $\hat{\nu}_i < \nu^\old_i$, and thus \ref{property_rad_2} is satisfied over the course of the algorithm.
    
    Moreover, the value of the $i$-th valid radius $\hat{\nu}_i$ is determined in~\linecref{algline:set_hat_nu_1} or in~\linecref{algline:set_hat_nu_2}.
    In turn, the value of $\hat{\nu}_i$ depends on the value of $\nu_{i-1}$ and
    on the value of $\tilde{\nu}_i$ determined in~\linecref{algline:check_if_nu_decreases_1} or in~\linecref{algline:check_if_nu_decreases_2} respectively.
    By construction, the value of $\tilde{\nu}_i$ is a power of $(1+\epsilon)$. 
    Therefore, using an induction argument for $\nu_{i-1}$, we can deduce that \ref{property_rad_1} is always satisfied as well.
\end{proof}

\begin{lemma} \label{lem:amort_update_time} 
    The amortized update time of the incremental Algorithm~\ref{alg:incremental} is~$n^{o(1)}$.
\end{lemma}
\begin{proof}
    Based on~\cref{lem:incr_appr_sssp}, the total update time of a single incremental $(1+\epsilon)$-approximate SSSP algorithm is $m^{1+o(1)}$.\footref{ftnote:W_poly_n} To compute the total update time due to all incremental $(1+\epsilon)$-approximate SSSP algorithms, we provide an upper bound on the number of times our incremental Algorithm~\ref{alg:incremental} initializes a new incremental $(1+\epsilon)$-approximate SSSP algorithm.
    
    After an adversarial edge insertion, Algorithm~\ref{alg:incremental} initializes a new incremental $(1+\epsilon)$-approximate SSSP algorithm
    only in~\linecref{algline:sssp_update} during the Resampling Phase. There are two crucial observations to be made: First, this event occurs only if the function
    \textsc{First-Level-Decrease}$()$ returns a level $i$ smaller than $t$ (i.e., $i < t$) in~\linecref{algline:first_level_decr}. Second, the value of the
    $i$-th radius $\nu_i$ is decreased inside the function \textsc{Update-DS-Rad-Decr}($i$) in~\linecref{algline:update_ds_rad_decr_1}. 
    Thus, whenever Algorithm~\ref{alg:incremental} initializes the first new incremental $(1+\epsilon)$-approximate SSSP algorithm during the Resampling Phase, there is a level $i$ such that the $i$-th radius $\nu_i$ is decreased. 
    Note that per Resampling Phase, there can be at most $t$ initializations of new incremental $(1+\epsilon)$-approximate SSSP algorithms. 
    
    Based on~\cref{obs:nu_only_decr_power_(1+eps)} (i.e., \ref{property_rad_1} and \ref{property_rad_2}), the total number of times a radius of a fixed level can be decreased is upper bounded by $O(\log_{(1+\epsilon)} nW)$. Therefore, \cref{obs:nu_only_decr_power_(1+eps),cor:number_of_levels} imply that the total number of times any radius can be decreased is upper bounded by~$O(\log_{(1+\epsilon)} nW \cdot \log n)$. In turn, the incremental Algorithm~\ref{alg:incremental} enters the Resampling Phase at most $O(\log_{(1+\epsilon)} nW \cdot \log n)$ times, and each
    time at most $t = O(\log n)$ new incremental $(1+\epsilon)$-approximate SSSP algorithms are initialized.
    Consequently, the total update time for all incremental $(1+\epsilon)$-approximate SSSP algorithms is: \[O\big(m^{1+o(1)} \cdot \log_{(1+\epsilon)} nW \cdot \log^2 n\big) \;=\; m^{1+o(1)}.\] 
    
    \noindent Furthermore, the total time needed for the rest of the operations can be upper bounded by $O(m^{1+o(1)} \cdot \log_{(1+\epsilon)} nW \cdot \log^2 n) = m^{1+o(1)}$. The reason is that Algorithm~\ref{alg:incremental} enters the Resampling Phase at most $O(\log_{(1+\epsilon)} nW \cdot \log n)$ times, and within the Resampling Phase each operation can be performed in $m^{1+o(1)}$ time.\footnote{
    Notice that each isolated vertex must be a center, as otherwise the cost of the solution would be infinite. Thus, whenever a vertex is sampled to enter a supporting candidate set, or whenever a vertex leaves or enters an approximate ball or a leaking set, we charge the vertex to one of its incident edges.\label{ftn:isolated_vert}} As discussed at the end of \cref{sec:edge insertions}, the computational costs of \linescref{algline:check_if_nu_decreases_1}{algline:check_if_nu_decreases_2} (the computation of~$\tilde{\nu}_i$) are included in the update time of the $i$-th incremental $(1+\epsilon)$-approximate SSSP algorithm $\mathcal{A}_i$, and so the total computational time outside the Resampling Phase is dominated by $m^{1+o(1)}$ as well.
    Finally, since in the incremental setting we are allowed to amortize over all the edge insertions, the corresponding amortized update time is equal to $\frac{2 \cdot m^{1+o(1)}}{m} = n^{o(1)}$, as $m \leq n^2$.
\end{proof}

\subsubsection{Analysis of the Upper Bound} \label{sec:incr_upper_bound}

\begin{lemma}\label{lem:v_to_one_set}
    After an adversarial edge insertion, each vertex $v \in V \setminus U_t$ belongs to exactly one set $U_i \setminus U_{i+1}$, where $0 \leq i < t$.
\end{lemma}
\begin{proof}
    Consider a fixed vertex $v \in V \setminus U_t$, and note that $v \in U_0$. 
    Observe that by construction in~\linescref{algline:next_exec_set_1}{algline:next_exec_set_2}, each subsequent execution set is a subset of the previous one.
    Hence, once the vertex $v$ is no longer part of an execution set, then $v$ is not part of any later execution set as well. Let $\ell$ be the largest level such that $v \in U_\ell$, and as $v \notin U_t$ we have $0 \leq \ell < t$. Thus, it follows that the vertex $v \in V \setminus U_t$ belongs only to the set $U_\ell \setminus U_{\ell+1}$, where $0 \leq \ell < t$.
\end{proof}

\begin{lemma} \label{lem:Zi_nonempty_nu_same}
    After an adversarial edge insertion, consider a level $i \in [0, t]$ for which the $i$-th leaking set $Z_i$ is non-empty. Then the algorithm does not decrease the $i$-th radius $\nu_i$. In other words, if $Z_i \neq \emptyset$ then $\nu_i^\old \leq \nu_i$.
\end{lemma}
\begin{proof}
    The last level $t$ is not relevant, since $Z_t = \emptyset$ by~\linecref{algline:last_level_assign}.
    For some level $i \in [0, t)$ suppose to the contrary that after the adversarial edge insertion, the new $i$-th leaking set 
    $Z_i$ is non-empty but the incremental~\cref{alg:incremental} decreases the $i$-th radius $\nu_i$.
    Observe that the $i$-th radius $\nu_i$ is updated only in~\linecref{algline:set_nu} within \textsc{Update-DS-Rad-Decr}($i$).
    But in this case, the $i$-th leaking set $Z_i$ becomes empty due to~\linecref{algline:leaking_Zi_empty} in \textsc{Update-DS-Rad-Decr}($i$), which leads to a contradiction. Therefore, if $Z_i \neq \emptyset$ then the $i$-th radius is not decreased after the adversarial edge insertion.
\end{proof}

\begin{observation} \label{obs:S_extends_only}
    After an adversarial edge insertion, no vertices are removed from the bicriteria approximate solution~$S$. Hence, it holds that $S^\old \subseteq S$.
\end{observation}
\begin{proof}
    The bicriteria approximate solution $S$ is assigned as the union of all candidate sets in~\linecref{algline:assign_solution}. Notice that only new vertices can be added to the candidate sets due to~\linecref{algline:extend_candidate_set}, and no vertices are removed from them, as required.
\end{proof}

\begin{observation} \label{obs:incr_nu_monoton}
    After an adversarial edge insertion, the order of the values of the radii $\nu_i$ is non-decreasing. In other words, it holds that~$\nu_0 \;\leq\; \nu_1 \;\leq\; \cdots \;\leq \nu_t$ (i.e., \ref{property_rad_3} is satisfied).
\end{observation}
\begin{proof}
    Observe that at the last level $t$, we trivially have $\nu_t = \nu_{t-1}$  by~\linecref{algline:last_level_assign}.
    At the beginning, the preprocessing phase ensures that \ref{property_rad_3} is satisfied.
    After an adversarial edge insertion, for a fixed level $i: 0 \leq i < t$ the $i$-th radius $\nu_i$ is updated only in~\linecref{algline:set_nu} within \textsc{Update-DS-Rad-Decr}($i$), where it is assigned the value of the $i$-th valid radius $\hat{\nu}_i$. In turn, the value of the $i$-th valid radius $\hat{\nu}_i$ is updated to $\max(\tilde{\nu}_i, \nu_{i-1})$ in~\linescref{algline:set_hat_nu_1}{algline:set_hat_nu_2}. Therefore, together with \ref{property_rad_2} (see~\cref{obs:nu_only_decr_power_(1+eps)}), it follows by an induction argument on the radii that \ref{property_rad_3} is satisfied over the course of the algorithm.
\end{proof}

\begin{lemma}\label{lem:cost_of_vert}
    After an adversarial edge insertion, consider a fixed level $i: 0 \leq i < t$ and
    an arbitrary vertex~$v \in U_i \setminus U_{i+1}$. Then it holds that $\dist(v, S) \leq \dist(v, \sigma(v)) \leq \nu_i$. 
\end{lemma}
\begin{proof}
    The proof is by induction on the number of adversarial edge insertions, where the base case holds due to the preprocessing phase. Regarding the induction step note that based on~\cref{lem:v_Bi_Zi}, the vertex $v$ belongs either to the $i$-th approximate ball $B_i$ or to the $i$-th leaking set $Z_i$. We consider both cases based on whether $v \in B_i$ or $v \in Z_i$. In the case where $v \in B_i$, we further distinguish based on whether the value of the $i$-th radius $\nu_i$ is decreased. 
    \begin{itemize}
        \item Assume that $v \in B_i$ and that the algorithm decreases the $i$-th radius $\nu_i$ (see~\linescref{algline:nu_decreases_first}{algline:nu_decreases}).
        Let $\ell$ be the smallest level for which the corresponding $\ell$-th radius $\nu_\ell$ is decreased, as returned in~\linecref{algline:first_level_decr}.
        Observe that for $i = \ell$ the function \textsc{Update-DS-Rad-Decr}($i$) is executed in~\linecref{algline:update_ds_rad_decr_1}, and for $i \neq \ell$ it is executed in~\linecref{algline:update_ds_rad_decr_2}. In both cases, by the construction in~\linecref{algline:incr_construct_Bi_update} where
        $B_i = \Retain\bigl(\ball[S_i, \nu_i, \delta_{S_i}] \cap U_i,\, \ceil{\beta \myhs |U_i|}\bigr)$
        and by~\linecref{algline:v_to_sigma}, together with the fact that the incremental $(1+\epsilon)$-approximate SSSP algorithm of~\cref{lem:incr_appr_sssp} does not underestimate distances, it holds that~$\dist(v, S_i) \leq \dist(v, \sigma(v)) \leq \delta_{S_i}(v) \leq \nu_i$. 
        As the $i$-th candidate set $S_i$ is a subset of the bicriteria approximate solution $S$, it follows that
        $\dist(v, S) \leq \dist(v, S_i) \leq \nu_i$, as needed.
        
        Otherwise assume that $v \in B_i$ and that the algorithm does not decrease the $i$-th radius $\nu_i$ (i.e., $\nu_i^\old \leq \nu_i$). By construction in~\linecref{algline:Bi_subset_Ui} it holds that
        $B_i \subseteq B_i^\old$. Hence we have $v \in B_i^\old$, and specifically that~$v \in U_i^\old \setminus U_{i+1}^\old$ by~\cref{lem:Bi_Zi_subsets_Ui,algline:next_exec_set_1,algline:next_exec_set_2}.\footnote{This $i$-th level existed before the adversarial edge insertion, because the radius of a new level must decrease the first time it appears (see also the comment at~\linecref{algline:comment_nu_largest} in~\cref{alg:incremental}).} Since $S^\old \subseteq S$ by~\cref{obs:S_extends_only}, we get that $\dist(v, S) \leq 
        \dist(v, S^\old)$. Moreover, the assignment $\sigma(v)$ is not modified in this case, implying that $\dist(v, \sigma(v)) = \dist(v, \sigma^\old(v))$.
        Under edge insertions the distances are non-increasing, and thus it holds that $\dist(v, S^\old) \leq \dist^\old(v, S^\old)$ and that $\dist(v, \sigma^\old(v)) \leq \dist^\old(v, \sigma^\old(v))$. Together with the induction hypothesis, it follows that:
        \begin{align*}
            &\dist(v, S) \leq \dist(v, S^\old) \leq \dist^\old(v, S^\old) \leq \nu_i^\old \leq \nu_i \;\text{ and } \\
            &\dist(v, \sigma(v)) = \dist(v, \sigma^\old(v)) \leq \dist^\old(v, \sigma^\old(v)) \leq \nu_i^\old \leq \nu_i.
        \end{align*}
        
        \item Assume that $v \in Z_i$ after the adversarial edge insertion. Based on~\cref{lem:v_Bi_Zi}, the vertex $v$ belonged to a set $U_j^\old \setminus U_{j+1}^\old$ with $j \leq i$, before the adversarial edge insertion. Notice that the assignment~$\sigma(v)$ is not modified in this case as well. Since under edge insertions the distances are non-increasing and~$S^\old \subseteq S$ by~\cref{obs:S_extends_only}, using the induction hypothesis we can deduce that:
        \begin{align*}
            &\dist(v, S) \leq \dist^\old(v, S) \leq \dist^\old(v, S^\old) \leq \nu_j^\old \; \text{ and } \\
            &\dist(v, \sigma(v)) \leq \dist^\old(v, \sigma(v)) = \dist^\old(v, \sigma^\old(v)) \leq  \nu_j^\old.
        \end{align*} 
        Since we have $j \leq i$, using~\ref{property_rad_3} (see~\cref{obs:incr_nu_monoton}) we can infer that $\dist(v, S) \leq \nu_i^\old$ and that $\dist(v, \sigma(v)) \leq \nu_i^\old$. Note that as the vertex $v$ belongs to the $i$-th leaking set $Z_i$, the set~$Z_i$ is non-empty. Therefore by~\cref{lem:Zi_nonempty_nu_same} it holds that $\nu_i^\old \leq \nu_i$, and in turn that $\dist(v, S) \leq \nu_i$ and that~$\dist(v, \sigma(v)) \leq \nu_i$.
    \end{itemize}

    Finally, recall that the image of the assignment $\sigma$ is the bicriteria approximate solution $S$. As a consequence, it follows that $\dist(v, S) \leq \dist(v, \sigma(v))$ for every vertex $v \in V$, concluding the claim.
\end{proof}

The next lemma provides an upper bound on the $(k, z)$-clustering cost of both the bicriteria approximate solution $S$ and the bicriteria approximate assignment $\sigma$.

\begin{lemma} \label{lem:incr_upper_bound_sol}
    After an adversarial edge insertion, both values $\cost^z(S)$ and $\cost^z(\sigma)$ are at most: \[\sum_{i=0}^{t-1} |U_i \setminus U_{i+1}| \cdot (\nu_i)^z.\]
\end{lemma}
\begin{proof}
     The $(k, z)$-clustering cost of both the bicriteria approximate solution $S$ and the bicriteria approximate assignment $\sigma$ is upper bounded as follows:
    \begin{align*}
        \cost^z(S) &= \sum_{v \in V} \dist(v, S)^z \;=\; \sum_{i=0}^{t-1} \sum_{v \in U_i \setminus U_{i+1}} \dist(v, S)^z \\
        &\le \sum_{i=0}^{t-1} \sum_{v \in U_i \setminus U_{i+1}} (\nu_i)^z \;\leq\; \sum_{i=0}^{t-1} |U_i \setminus U_{i+1}| \myhs (\nu_i)^z. \\
        \cost^z(\sigma) &= \sum_{v \in V} \dist(v, \sigma(v))^z \;=\; \sum_{i=0}^{t-1} \sum_{v \in U_i \setminus U_{i+1}} \dist(v, \sigma(v))^z \\
        &\le \sum_{i=0}^{t-1} \sum_{v \in U_i \setminus U_{i+1}} (\nu_i)^z \;\leq\; \sum_{i=0}^{t-1} |U_i \setminus U_{i+1}| \myhs (\nu_i)^z.
    \end{align*}
    The second equality in both expressions derives from~\cref{lem:v_to_one_set}
    and the fact that if a vertex $v$ belongs to~$U_t$ then the algorithm adds $v$ to $S$ and sets $\sigma(v)$ to $v$ (see~\linescref{algline:last_level_assign}{algline:assign_solution} in~\cref{alg:incremental}). 
    Moreover, the first inequality in both expressions follows from Lemma~\ref{lem:cost_of_vert}, and this completes the proof.
\end{proof}

\subsubsection{Results Regarding the Lower Bound} \label{sec:incr_lower_bound}
We continue with the lower bound on the optimal $(k, z)$-clustering cost  under adversarial edge insertions, adjusting the analysis in~\cref{sec:lower_bound_variant} (which in turn adapts the analysis in \cite[Section 3]{mettu2004optimal}). Similarly to the analysis in~\cref{sec:analysis_static}, we fix a positive real number $\gamma$ such that $\beta < \gamma < 1$. We also define $\nu^*_i$ and $\mu^*_i$ for a fixed level $i \in [0, t]$, as described in Definition~\ref{def:nu_star_incr} and Definition~\ref{def:mu_star_incr} respectively.\footnote{Note that these definitions are similar to  Definition~\ref{def:nu_star} and Definition~\ref{def:mu_star}.}

\begin{definition} \label{def:nu_star_incr}
After an adversarial edge insertion, for a fixed level $i \in [0, t]$ consider the $i$-th execution set~$U_i$ and the $i$-th candidate set $S_i$ maintained in the incremental Algorithm~\ref{alg:incremental}.
For this specific level $i$, we define $\nu_i^*$ as follows:
\[
    \nu_i^*\coloneqq \min\{r \in \mathbb{R} \mid \lvert\ball[S_i, r]\cap U_i\rvert \geq \beta \myhs |U_i|\}.
\]
\end{definition}

\begin{lemma}\label{lem:relation_nu_tilde_nu_star_incr}
    After an adversarial edge insertion, for a fixed level $i: 0 \leq i < t$ it holds that $\nu_i^* \leq \tilde{\nu}_i \leq (1+\epsilon)^2 \, \nu_i^*$.
\end{lemma}
\begin{proof}
    The value of $\tilde{\nu}_i$ is computed either in~\linecref{algline:check_if_nu_decreases_1} or in~\linecref{algline:check_if_nu_decreases_2}, using the (updated) $i$-th execution set~$U_i$ and the (updated) $i$-th candidate set $S_i$, in the following way:
    \[\tilde{\nu}_i \,=\, \min_{j \in \mathbb{Z}^{\geq 0}} \{(1+\epsilon)^j \mid \lvert \ball[S_i, (1+\epsilon)^j, \delta_{S_i}]\cap U_i\rvert \geq \beta \myhs |U_i|\}.\]
    Furthermore, in both lines the (updated) distance estimates~$\delta_{S_i}(\cdot)$ are used for the computation of the corresponding closed balls.   
    Based on~\cref{lem:incr_appr_sssp}, we have $\delta_{S_i}(v) \leq (1+\epsilon) \dist(v, S_i)$ for every vertex $v \in V$, and the $i$-th incremental $(1+\epsilon)$-approximate SSSP algorithm $\mathcal{A}_i$ of~\cref{lem:incr_appr_sssp} does not underestimate distances. Therefore as the value of $\tilde{\nu}_i$ is also a power of $(1+\epsilon)$, it follows that $\nu_i^* \leq \tilde{\nu}_i \leq (1+\epsilon)^2 \, \nu_i^*$, as required.
\end{proof}

\begin{definition} \label{def:mu_star_incr}
    After an adversarial edge insertion, for a fixed level $i \in [0, t]$ consider the $i$-th execution set~$U_i$
    maintained in the incremental Algorithm~\ref{alg:incremental}. For this specific level $i$, we define $\mu_i^*$ to be the minimum nonnegative real number such that there exists a subset of vertices $X_i \subseteq V$ with size at most $k$ (i.e., $|X_i| \leq k$) for which the following properties hold:
    \begin{enumerate}
        \item $\lvert \ball[X_i, \mu_i^*]\cap U_i\rvert \geq \gamma \myhs |U_i|$. 
        \item $\lvert U_i\setminus \ball(X_i, \mu_i^*)\rvert \geq (1-\gamma) \myhs |U_i|$. 
    \end{enumerate}
\end{definition}

Recall that in the definition above, the vertices of $U_i$ that are exactly at distance $\mu_{i}^*$ from~$X_i$ belong to both $\ball[X_i, \mu_i^*]\cap U_i$ and $U_i\setminus \ball(X_i, \mu_i^*)$. At this point in the analysis, we would like to use~\cref{lem:nu<2mu} to establish a relationship between $\nu_i^*$ and $\mu_i^*$. However,~\cref{lem:nu<2mu} is applicable only to the static setting, whereas in the incremental setting the execution sets and candidate sets are subject to changes. To that end, we first provide an extension of~\cref{lem:nu<2mu}
demonstrated as~\cref{lem:nu<2mu_1}, and then we establish an `incremental' relationship between $\nu_i^*$ and $\mu_i^*$ in~\cref{lem:nu<2mu_incr} utilizing the fact that the adversary is oblivious. The proof of~\cref{lem:nu<2mu_1} is very similar to that of \Cref{lem:nu<2mu}, and both proofs are deferred to \Cref{appendix:missingproofs}.

\begin{restatable}{lemma}{nulesstwomuone}\label{lem:nu<2mu_1}
    After an adversarial edge insertion, for a fixed level $i \in [0, t]$ consider the $i$-th execution set~$U_i$ 
    and a subset of vertices $Y \subseteq U_i$ obtained by sampling each vertex of $U_i$ independently with probability~$\min\Big(\frac{\alpha \myhs k \log n}{|U_i|}, 1\Big)$. Let $\nu_i^*(Y)\coloneqq \min\{r \in \mathbb{R} \mid \lvert\ball[Y, r]\cap U_i\rvert \geq \beta \myhs |U_i|\}$ be an extension of~\cref{def:nu_star_incr}.
    Then it holds with high probability that $\nu_i^*(Y) \leq 2 \myhs \mu_i^*$.
\end{restatable}

\begin{lemma}\label{lem:nu<2mu_incr}
    After an adversarial edge insertion, with high probability it holds that $\nu_i^* \leq 2 \myhs \mu_i^*$ for all levels~$i \in [0, t]$.
\end{lemma}
\begin{proof}
    The key observation is that by the construction of the incremental bicriteria approximation~\cref{alg:incremental}, whenever the $i$-th execution set $U_i$ is modified, a supporting candidate set $\tilde{S}_i$ is constructed in~\linecref{algline:construct_support_candidate_set} and is added to the $i$-th candidate set $S_i$ in~\linecref{algline:extend_candidate_set} (within the Resampling Phase). Hence, the $i$-th most recent supporting candidate set $\tilde{S}_i$ is obtained at the last modification of the $i$-th execution set~$U_i$ in a manner similar to the set $Y$ in~\cref{lem:nu<2mu_1}.
    Furthermore, since the edge insertions are sent by an oblivious adversary, the choice of the random set $\tilde{S}_i$ is independent of the at most $O(n^2)$ versions of the graph in the incremental setting. 
    
    Thus by applying a union bound and using the obliviousness of the adversary,~\cref{lem:nu<2mu_1} holds
    for the $i$-th most recent supporting candidate set $\tilde{S}_i$ and the current $i$-th execution set $U_i$ still with high probability. Therefore it holds with high probability that $\nu_i^*(\tilde{S}_i) \leq 2 \myhs \mu_i^*$, where $i \in [0, t]$ is a fixed level and $\nu_i^*(\cdot)$ is the extension of~\cref{def:nu_star_incr} defined in~\cref{lem:nu<2mu_1}. Since we have $\tilde{S}_i \subseteq S_i$ and $\nu_i^* = \nu_i^*(S_i)$, we infer that $\nu_i^* = \nu_i^*(S_i) \leq \nu_i^*(\tilde{S}_i) \leq 2 \myhs \mu_i^*$ with high probability.
    By applying a union bound over the $t = O(\log n)$ levels (see~\cref{cor:number_of_levels}), the high probability claim follows for all levels $i \in [0, t]$.
\end{proof}

\begin{corollary}\label{cor:tilde_nu_bound_incr}
    After an adversarial edge insertion, with high probability it holds that $\tilde{\nu}_{i} \le 2(1+\epsilon)^2 \myhs \mu_i^*$ for every level $i \in [0, t)$.
\end{corollary}
\begin{proof}
    The claim follows from \cref{lem:relation_nu_tilde_nu_star_incr} and \cref{lem:nu<2mu_incr}.
\end{proof}

Even though \cref{cor:tilde_nu_bound_incr} provides an upper bound for $\tilde{\nu}_{i}$, recall that the incremental Algorithm~\ref{alg:incremental} uses the radius~$\nu_i$ computed in~\linecref{algline:set_nu}, which is at least~$\tilde{\nu}_{i}$. In particular, we upper bound the cost of the bicriteria approximate solution $S$ with respect to $\nu_i$ in Lemma~\ref{lem:incr_upper_bound_sol}.
Hence, in the next lemma we provide an upper bound on the value of the $i$-th radius $\nu_i$ in terms of the value of some $\mu^*_j$ with $j\leq i$, which later is associated with the optimal $(k, z)$-clustering cost.

\begin{lemma} \label{lem:nu_bound_incr}
After an adversarial edge insertion, for a fixed level $i: 0 \leq i < t$ the value of the $i$-th radius can be upper bounded as follows:
\[
    \nu_i \le 2 (1+\epsilon)^2 \myhs \mu_j^*, \;\text{where }j = \max \{\ell \in [0, i] \mid \tilde{\nu}_\ell \ge \nu_{\ell-1}\}.\footnote{Such a level $j$ exists since $\tilde{\nu}_0 \geq \nu_{-1} = 0$.}
\]
\end{lemma}
\begin{proof}
    The $i$-th radius $\nu_i$ is updated only in~\linecref{algline:set_nu} within \textsc{Update-DS-Rad-Decr}($i$), where it is assigned the value of the $i$-th valid radius~$\hat{\nu}_i$. This event occurs only if $\hat{\nu}_i < \nu_i^\old$ as indicated by~\linescref{algline:while_find_small_nu}{algline:nu_decreases}, and thus we have~$\nu_i \leq \hat{\nu}_i$.
    In turn, the $i$-th valid radius $\hat{\nu}_i$ is updated to $\max(\tilde{\nu}_i, \nu_{i-1})$ in~\linescref{algline:set_hat_nu_1}{algline:set_hat_nu_2}. Based on these observations, the proof of the statement proceeds by induction on the number of levels after the adversarial edge insertion. For the base case where $i = 0$, since $\nu_{-1} = 0$ we have  
    $\nu_0 \leq \hat{\nu}_0 = \tilde{\nu}_0$, which is at most $2(1+\epsilon)^2\mu_0^*$ based on~\cref{cor:tilde_nu_bound_incr}. Regarding the induction step, we analyze separately the two possible cases:
    \begin{itemize}
        \item Assume that $\tilde{\nu}_i \ge \nu_{i-1}$, which implies that $\hat{\nu}_i = \tilde{\nu}_{i}$. Using that $\nu_i \leq \hat{\nu}_i$ and~\cref{cor:tilde_nu_bound_incr}, it follows that $\nu_i \leq \tilde{\nu}_i \leq 2 (1+\epsilon)^2 \myhs \mu_i^*$, as needed.

        \item Assume that $\tilde{\nu}_i < \nu_{i-1}$, which implies that
        $\hat{\nu}_i = \nu_{i-1}$. Using that $\nu_i \leq \hat{\nu}_i$ and \ref{property_rad_3} (see~\cref{obs:incr_nu_monoton}), it follows that $\nu_i = \nu_{i-1}$. By the induction hypothesis, we have $\nu_{i-1} \le 2 (1+\epsilon)^2 \mu_{j'}^*$ where $j'= \max \{\ell \in [0, i - 1] \mid \tilde{\nu}_\ell \ge \nu_{\ell-1}\}$. By our assumption that $\tilde{\nu}_i < \nu_{i-1}$, we can infer that $j'$ is also equal to $\max \{\ell \in [0, i] \mid \tilde{\nu}_\ell \ge \nu_{\ell-1}\}$. Therefore,
        we conclude that $\nu_i = \nu_{i-1} \leq 2 (1+\epsilon)^2 \myhs \mu_j^*$ where~$j = \max \{\ell \in [0, i] \mid \tilde{\nu}_\ell \ge \nu_{\ell-1}\}$, as required.
    \end{itemize}
\end{proof}

Notably, we can lower bound the optimal $(k, z)$-clustering cost with respect to the $\mu_i^*$ values by reusing Lemma~\ref{lem:lower_bound_opt},
which we restate for convenience. 
The reason Lemma~\ref{lem:lower_bound_opt} continues to hold after every adversarial edge insertion in the incremental setting is as follows. After every adversarial edge insertion, the proof of Lemma~\ref{lem:lower_bound_opt} (given in~\cref{appendix:missingproofs}) can be repeated with two modifications: in~\cref{clm:rel_Y_Fi} we replace~\cref{def:mu_star} by its incremental equivalent~\cref{def:mu_star_incr}, and in~\cref{clm:Fi_shrink} we replace~\cref{lem:exec_decr_beta} by its incremental counterpart~\cref{lem:size_rel_exec_sets}. 
The reader may also find it helpful to refer to the discussion in~\cref{sec:lower_bound_variant} just before Lemma~\ref{lem:lower_bound_opt}. Recall that $\gamma \in (\beta, 1)$ is the parameter used in~\cref{def:mu_star_incr}.

\lowerboundOPT*

\subsubsection{Size of the Bicriteria Approximate Solution and Proof of Theorem~\ref{thm:incr_alg}} \label{sec:size_finish_proof}

\begin{lemma} \label{lem:size_of_sol}
    After an adversarial edge insertion, the size of the bicriteria approximate solution $S$
    is with high probability at most $O(k \log^3 n \, \log_{1+\epsilon}nW)$.
\end{lemma}
\begin{proof}
    After an adversarial edge insertion, the bicriteria approximate solution $S$ equals $\bigcup_{i=0}^{t} S_i$ according to~\linecref{algline:assign_solution}. Since we always have $t = O(\log n)$ by~\cref{cor:number_of_levels}, in order to conclude the proof, we upper bound the size of each $i$-th candidate set $S_i$ by $O(k \log^2 n \, \log_{1+\epsilon}nW)$.
    Notice that the size of the last $t$-th candidate set~$S_t$ is at most $O(k \log n)$ by construction in~\linescref{algline:beg_resample}{algline:last_level_assign}. For a fixed level $i \in [0, t-1]$, the $i$-th candidate set~$S_i$ is updated in~\linecref{algline:extend_candidate_set}, where it is merged with the corresponding $i$-th supporting candidate set $\tilde{S}_i$. The $i$-th supporting candidate set $\tilde{S}_i$ is constructed in~\linecref{algline:construct_support_candidate_set}, by including each vertex of the $i$-th execution set $U_i$ independently with probability $\min\big(\frac{\alpha \myhs k\log n}{|U_i|}, 1\big)$, where $\alpha$ is a sufficiently large constant. Hence by a
    straightforward application of Chernoff bound (see~\cref{lem:chernoff}), we can infer that whenever a supporting candidate set 
    $\tilde{S}_i$ is obtained, it holds that $|\tilde{S}_i| = O(\alpha \myhs k \log n) = O(k \log n)$ with probability at least $1 - \frac{1}{n^{8\alpha}}$ (see also~\cref{lem:size_of_S}).
    
    Regarding the size of the $i$-th candidate set $S_i$, there are two crucial observations to be made: First, the $i$-th candidate set $S_i$ receives new vertices only if the function \textsc{First-Level-Decrease}$()$ returns a level $j$ smaller than $t$ (i.e., $j < t$) in~\linecref{algline:first_level_decr}. Second, the value of the
    $j$-th radius $\nu_j$ is decreased inside the function \textsc{Update-DS-Rad-Decr}($j$) in~\linecref{algline:update_ds_rad_decr_1}. 
    Thus, whenever the $i$-th candidate set $S_i$
    is merged with a new $i$-th supporting candidate set $\tilde{S}_i$ during the Resampling Phase, there is a level $j$ such that the $j$-th radius $\nu_j$ is decreased. 
    
    Based on~\cref{obs:nu_only_decr_power_(1+eps)} (\ref{property_rad_1} and~\ref{property_rad_2}), the total number of times a radius of a fixed level can be decreased is upper bounded by $O(\log_{1+\epsilon} nW)$. Therefore, \cref{obs:nu_only_decr_power_(1+eps),cor:number_of_levels} imply that the total number of times any radius can be decreased
    is upper bounded by~$O(\log_{1+\epsilon} nW \cdot \log n)$. In turn, the number of times an $i$-th supporting candidate set is constructed and the $i$-th candidate set receives new vertices is upper bounded by
    $O(\log_{1+\epsilon} nW \cdot \log n)$. By applying a union bound over the possible supporting candidate sets, we
    can deduce that with high probability~the $i$-th candidate set receives at most~$O(k \log n)$ new vertices up to
    $O(\log_{1+\epsilon} nW \cdot \log n)$ times. 
    As a result, the size of the $i$-th candidate set~$S_i$ is always upper bounded by $O(k \log^2 n \cdot \log_{1+\epsilon} nW)$ with high probability. 
    
    By applying a union bound over the $t = O(\log n)$ candidate sets $S_i$, we conclude that with high probability for all levels $i \in [0, t]$ it holds that $|S_i| = O(k \log^2 n \cdot \log_{1+\epsilon} nW)$ after an adversarial edge insertion. As a consequence, with high probability it holds that $|S| = O(k \log^3 n \cdot \log_{1+\epsilon} nW)$ after an adversarial edge insertion, as required.
\end{proof}

\begin{observation} \label{obs:bound_for_size_U}
    After an adversarial edge insertion, for a fixed level $i: 0 \leq i < t$ 
    it holds that $\sum_{i \leq j < t} |U_j \setminus U_{j+1}| \leq |U_i|$.
\end{observation}
\begin{proof}
    For every level $j$, we have  $U_{j+1} \subseteq U_j$ due to~\linescref{algline:next_exec_set_1}{algline:next_exec_set_2}. Moreover, notice that the subsets $\{U_j \setminus U_{j+1} \}_{i \leq j < t}$ are pairwise disjoint and subsets of the $i$-th execution set $U_i$. Therefore, the sum of their sizes is at most the size of $U_i$, as needed.
\end{proof}

By combining the previous observations and lemmas, we can now finish the proof of \cref{thm:incr_alg} which we restate for convenience.

\incralg*

\begin{proof}
    Consider the incremental bicriteria approximation~\cref{alg:incremental} and its analysis. The claim regarding the amortized update time follows from Lemma~\ref{lem:amort_update_time}, and the claim regarding the size of the bicriteria approximate solution follows from Lemma~\ref{lem:size_of_sol}. Since the bicriteria approximate assignment $\sigma$ maps each vertex to a candidate center in $S$, the size of the bicriteria approximate assignment follows from~\cref{lem:size_of_sol} as well.
    Hence, it remains to upper bound the approximation ratio, which we establish in the rest of the proof.
    
    Based on \cref{lem:nu_bound_incr}, each radius $\nu_i$ is bounded by some $\mu^*_j$, where $i \in [0, t)$ and $j \in [0, i]$. For a fixed level $j \in [0, t]$, let $A_j$ be the set of all indices $i \in [j, t)$ that use $\mu^*_j$ in their bound in \cref{lem:nu_bound_incr}. Notice that each level $i \in [0, t)$ belongs to exactly one such set $A_j$ with $j \in [0, i]$.
    Then by~\cref{lem:incr_upper_bound_sol}, the value of~$\cost^z(S)$ and $\cost^z(\sigma)$ can be upper bounded as follows:\footnote{Note that~\cref{lem:nu_bound_incr} inherits the high probability bound from \cref{cor:tilde_nu_bound_incr}.}
    \begin{align*}
        \cost^z(S), \cost^z(\sigma) &\;\leq\; \sum_{i=0}^{t-1} |U_i \setminus U_{i+1}| \myhs (\nu_i)^z \;=\; \sum_{j=0}^t \sum_{i \in A_j} |U_i \setminus U_{i+1}| \myhs (\nu_i)^z \\
        &\;\leq\; \sum_{j=0}^t \sum_{i \in A_j} |U_i \setminus U_{i+1}| \, (2(1+\epsilon)^2\mu^*_j)^z \\
        &\;=\; (2(1+\epsilon)^2)^z \sum_{j=0}^t (\mu^*_j)^z \sum_{i \in A_j} |U_i \setminus U_{i+1}|.
    \end{align*}
    Based on~\cref{obs:bound_for_size_U}, it holds that $\sum_{i \in A_j} |U_i \setminus U_{i+1}| \leq |U_j|$, and recall that if $i \in A_j$ then $j \leq i < t$. Thus, we can conclude that:
    \[
        \cost^z(S), \cost^z(\sigma) \;\leq\; (2(1+\epsilon)^2)^z \cdot  \sum_{j=0}^t (\mu^*_j)^z \cdot |U_j|.
    \]
    Based on~\cref{lem:lower_bound_opt}, we have:
    \[
        \sum_{j=0}^t (\mu^*_j)^z \cdot |U_j| \;\leq\; \frac{2r}{1-\gamma} \cdot \OPT,
    \]
    where $r = \ceil{\log_{(1-\beta)}((1-\gamma)/3)}$. In turn, this implies that:
    \[
        \cost^z(S), \cost^z(\sigma) \; \leq \; \frac{2r \myhs (2(1+\epsilon)^2)^z}{1-\gamma} \cdot \OPT.
    \]
    Therefore, the approximation ratio of the incremental~\cref{alg:incremental} is upper bounded by $\frac{2r \myhs (2(1+\epsilon)^2)^z}{1-\gamma}~=~O(1)$, and the claim follows.
\end{proof}

    \section{From Bicriteria Approximation to~$(k, z)$-Clustering on Incremental Graphs} \label{sec:bicr_to_kmed_incr}

Thus far, our incremental bicriteria approximation~\cref{alg:incremental} of~\cref{thm:incr_alg}  efficiently maintains an $(O(1), O(\log^3 n \, \log_{1+\epsilon}nW))$-bicriteria approximate solution to the $(k,z)$-clustering problem on a graph undergoing edge insertions.
Recall that by~\cref{def:bicriteria}, this means that \cref{alg:incremental} maintains a subset of vertices $S \subseteq V$ such that $\cost^z(S) \leq O(\OPT)$ and $|S| \leq O(k \log^3 n \, \log_{1+\epsilon}nW)$. Our goal though is to compute a constant-factor approximate solution $C$ for the $(k, z)$-clustering problem (see~\cref{def:k_median}), and thus the output $C$ must consist of at most $k$ vertices. In other words, we want to maintain a subset of vertices $C \subseteq V$ such that $\cost^z(C)\leq O(\OPT)$ and $|C| \leq k$.

In this section, we describe an efficient \emph{incremental reduction algorithm} that reduces the size of an $(O(1), O(\log^3 n \, \log_{1+\epsilon}nW))$-bicriteria approximate solution $S$ by converting it into a constant-factor approximate $(k, z)$-clustering solution $C$ in a graph that undergoes edge insertions. To achieve this, we maintain a $(k, z)$-clustering solution $C$ over the bicriteria approximate solution $S$, using weights on~$S$ determined by the corresponding bicriteria approximate assignment $\sigma$.

We remark that our~\cref{thm:bicr_to_kzclustering} (the incremental reduction algorithm) is an adaptation of a theorem in~\cite{mettu2004optimal,guha2000clustering} (its variant appears as~\cref{lem:convert_bicr_kzclustering}) in the incremental graph setting.
The following definitions are used to analyze the total update time in~\cref{thm:bicr_to_kzclustering}.
Notice that in~\cref{def:batches_assign}, the value of $\sigma_\text{inc}(V)$ is equal to the total number of batches of new vertices entering the image of $\sigma$. Hence, whenever~$|\sigma_\text{max}(V)|$ increases in a single update due to multiple new vertices, $\sigma_\text{inc}(V)$ increases by one.

\begin{definition} \label{def:sigma_max}
    Given a dynamic $(\alpha, \beta)$-bicriteria approximate assignment $\sigma$, the dynamic set \(\sigma_\text{max}(V) \coloneqq \{s \in V \mid s \in \sigma(V) \text{ at some moment}\}\) contains all vertices that have ever belonged to the image of $\sigma$.
\end{definition}

\begin{definition} \label{def:batches_assign}
    Given a dynamic $(\alpha, \beta)$-bicriteria approximate assignment $\sigma$, the value  $\sigma_\text{inc}(V)$ denotes the total number of times that $|\sigma_\text{max}(V)|$ increases. Formally, $\sigma_\text{inc}(V) \coloneq |\{\tau \in \mathbb{N} \mid |\sigma_\text{max}(V)| \text{ increases at the } \tau\text{-th update}\}|$.
\end{definition}

\begin{restatable}{theorem}{bicrtokzclustering}\label{thm:bicr_to_kzclustering}
There is a (randomized) dynamic algorithm for the $(k,z)$-clustering problem that given:
\begin{itemize}
    \item a weighted undirected graph $G = (V, E, w)$ subject to edge insertions,
    \item an integer $k \geq 1$, constants $z \geq 1, \epsilon \in (0, \frac{1}{2})$, and
    \item a dynamic $(\alpha, \beta)$-bicriteria approximate assignment $\sigma$, for any $\alpha \geq 1, \beta \geq 1$,
\end{itemize}
maintains an $O(\alpha)$-approximate $(k,z)$-clustering solution with a total update time of: 
\[
     |\sigma_\text{max}(V)| \cdot m^{1+o(1)} \;+\; \tilde{O}\big(|\sigma_\text{max}(V)|^2 \cdot \sigma_\text{inc}(V) \,+\, m \cdot |\sigma_\text{max}(V)|^{1+\epsilon}\big).
\]
\end{restatable}

\vspace{1em} 
A pseudocode of our incremental reduction algorithm is provided in~\cref{alg:bicr_to_kmedian}.
Our incremental reduction algorithm for the $(k,z)$-clustering problem (in~\cref{thm:bicr_to_kzclustering}) uses the
fully dynamic spanner algorithm of Baswana, Khurana, and Sarkar~\cite{BaswanaKS12}.
Their initial result is for unweighted graphs, but as the authors mention in their Remark 1.2 in~\cite{BaswanaKS12}, it can be extended to weighted graphs. 

\begin{lemma}[Theorem 4.12/Theorem 5.7 in~\cite{BaswanaKS12}]\label{lem:fully_dyn_spanner}
    There is a randomized fully dynamic algorithm that, given a weighted undirected graph $G = (V, E, w)$ subject to edge updates and a constant $\lambda \geq 1$, maintains an $O(2\lambda-1)$-spanner of expected size $\tilde{O}(n^{1+\frac{1}{\lambda}})$ and has an expected amortized update time of $\tilde{O}(1)$. The preprocessing time of the algorithm is~$\tilde{O}(m)$.\footnote{For unweighted graphs, the stretch of the spanner is $2\lambda-1$; the $O(2\lambda-1)$ arises for weighted graphs.} 
\end{lemma}

Note that an incremental spanner (i.e., non-increasing distances) would be sufficient; however, the fully dynamic spanner algorithm of~\cref{lem:fully_dyn_spanner} is already suitable in our context. 
The \emph{expected} guarantees can be turned into \emph{high probability} guarantees by running $\Theta(\log n)$ copies of the algorithm in parallel; see~\cite{HenzingerKN14} for more details.
The incremental reduction algorithm also utilizes the static $(k, z)$-clustering algorithm of Dupre la Tour and Saulpic~\cite{latour2025fastersimplergreedyalgorithm}. 

\begin{lemma}[\cite{latour2025fastersimplergreedyalgorithm}] \label{lem:static_kz_clust}
    There is randomized algorithm for the $(k, z)$-clustering problem that, given a weighted undirected graph $G = (V, E, w)$, an integer $k \geq 1$, and constants $z \geq 1, \lambda \geq 5$, computes with high probability an $O(\lambda^6)$-approximate $(k, z)$-clustering solution in $\tilde{O}(m^{1+\frac{1}{\lambda}})$ time.
\end{lemma}

At a high level, we maintain the $(k, z)$-clustering solution $C$ over the bicriteria approximate solution~$S$, using weights on $S$ derived from the bicriteria approximate assignment $\sigma$. Hence, the input graph instance contains both edge and vertex weights, and the static $(k, z)$-clustering algorithm should be able to work with vertex weights in our context. The algorithm by Dupre la Tour and Saulpic~\cite{latour2025fastersimplergreedyalgorithm} is a simplification of the algorithm by Mettu and Plaxton~\cite{MettuP03online}.\footnote{The two algorithms by Mettu and Plaxton in~\cite{MettuP03online} and~\cite{mettu2004optimal} (which includes the MP-bi algorithm) are different.} Even though the authors in~\cite{latour2025fastersimplergreedyalgorithm} present a version without vertex weights, the authors in~\cite{MettuP03online} present a version with weights on points. Thus, it is not surprising that the static $(k, z)$-clustering algorithm of~\cref{lem:static_kz_clust} can be generalized to handle vertex weights. For completeness, we provide a more
detailed discussion of this in~\cref{appendix:vertex_weights}. 

\weightedkzclustering*

Based on the discussion in~\cref{appendix:vertex_weights}, we can safely assume for the rest of the paper that the static $(k, z)$-clustering algorithm of~\cref{lem:static_kz_clust} works on weighted $(k, z)$-clustering instances, as defined in~\cref{def:wgt_k_median}.
The randomization in~\cref{thm:bicr_to_kzclustering} arises  from the randomization in~\cref{lem:fully_dyn_spanner,lem:static_kz_clust}; replacing the randomized subroutines with deterministic counterparts would yield a deterministic version of~\cref{thm:bicr_to_kzclustering}.


\paragraph{Adversarial updates.}
Our aim is to combine the incremental bicriteria approximation~\cref{alg:incremental} with the incremental reduction~\cref{alg:bicr_to_kmedian}.
Hence, the output $(\alpha, \beta)$-bicriteria approximate assignment~$\sigma$ of~\cref{alg:incremental}
becomes the input $(\alpha, \beta)$-bicriteria approximate assignment $\sigma$ of~\cref{alg:bicr_to_kmedian}. 
The assignment~$\sigma$ in~\cref{alg:incremental} can be modified due to an adversarial edge insertion into the graph $G$, making it a \emph{dynamic assignment}
in~\cref{thm:bicr_to_kzclustering}.

Since our incremental reduction algorithm is given as input a dynamic bicriteria approximate assignment, one type of adversarial update consists of modifications to the assignment $\sigma$. The other type of adversarial update involves edge insertions into the input graph $G$, where each adversarial edge insertion can trigger a batch of distance decreases between vertices in $G$. 

Note that an efficient combination of \cref{alg:incremental} and \cref{alg:bicr_to_kmedian}, requires that both (1) the total number of vertices that have ever belonged to the image of the dynamic assignment $\sigma$ (i.e., $|\sigma_\text{max}(V)|$) 
and (2) the total number of batches of new vertices entering $\sigma_\text{max}(V)$ (i.e., $\sigma_\text{inc}(V)$) remain bounded.
For efficiency reasons, we employ the approximate distance estimates from incremental $(1+\epsilon)$-approximate SSSP algorithms and perform lazy updates by maintaining distances as powers of $ (1+\epsilon)$.

\paragraph{State of the reduction algorithm.}
Our incremental reduction~\cref{alg:bicr_to_kmedian} maintains a dynamic $(\alpha, \beta)$-bicriteria approximate solution $S \coloneqq \sigma(V)$, which is the image of the dynamic $(\alpha, \beta)$-bicriteria approximate assignment $\sigma$ from the input.\footnote{Since the algorithm never removes vertices from $S$, we actually have $S = \sigma_\text{max}(V)$.} For every vertex $s \in S$, the reduction algorithm maintains an incremental $(1+\epsilon)$-approximate SSSP algorithm $\mathcal{A}_s$ providing distance estimates $\delta_s(\cdot)$. In turn, the reduction algorithm maintains a complete subgraph $H \coloneqq (S, S \times S, w_H)$ with edge and vertex weights, specified as follows:  
\begin{itemize}
    \item Each vertex $s \in S$ belongs to the vertex set of $H$, and the vertex weight
    $\wt(s)$ is equal to $|\sigma^{-1}(s)|$. In other words, the vertex weight of $s$ in $H$ is the number of vertices mapped to $s$ under the assignment~$\sigma(\cdot)$.
    
    \item For every pair of vertices $x, y \in S$, the edge $(x, y)$ belongs to the edge set of $H$ with an edge weight~$w_H(x, y)$ equal to $\min(\delta_x(y), \delta_y(x))$, rounded up to the nearest power of $(1 + \epsilon)$.
\end{itemize}
\noindent
The incremental reduction algorithm then applies the dynamic spanner algorithm $\mathcal{SP}$ of~\cref{lem:fully_dyn_spanner} on the subgraph $H$ to obtain an edge-sparsified graph $\tilde{H}$. Finally, the reduction algorithm runs the static $(k, z)$-clustering algorithm~$\mathcal{CL}$ of~\cref{lem:static_kz_clust} on the weighted $(k, z)$-clustering instance $(\tilde{H}, k, z, \wt)$, producing a $(k, z)$-clustering solution~$C$.

\subsection{Adversarial Edge Insertions and Assignment Modifications}
In the next two paragraphs, we first describe how our incremental reduction algorithm processes an adversarial edge insertion into the graph $G$, and then
how it processes a batch of modifications to the $(\alpha, \beta)$-bicriteria approximate assignment $\sigma$. We remark that in fact, each adversarial edge insertion into the graph $G$
can trigger a batch of distance decreases between vertices in $G$ alongside a batch of modifications to~$\sigma$.

\paragraph{Edge insertion into $G$.}
When an edge $(u, v)$ is inserted into the graph $G = (V, E, w)$ from the adversary, the incremental reduction algorithm proceeds as follows. Initially, this edge insertion $(u, v)$ is forwarded to all the incremental $(1+\epsilon)$-approximate SSSP algorithms $\mathcal{A}_s$, where $s \in S$. Based on~\cref{lem:incr_appr_sssp}, the incremental $(1+\epsilon)$-approximate SSSP algorithm $\mathcal{A}_s$ reports a change whenever the distance estimate $\delta_s(x)$ is updated for a vertex $x \in V$. After the adversarial insertion of the edge $(u, v)$ into the graph $G$, let $\Delta_S$ be the set containing all the pairs of vertices $(x, y) \in S \times S$ such that the value of the distance estimate $\delta_x(y)$ becomes less than $\frac{w_H(x, y)}{1+\epsilon}$.

Next, for every pair of vertices $(x, y) \in \Delta_S$, the incremental reduction algorithm updates the edge weights~$w_H(x, y)$ and $w_H(y, x)$ to $\min(\delta_{x}(y),\delta_{y}(x))$, rounding them up to the nearest power of $(1+\epsilon)$. Let $B \coloneqq \big\{\big(x, y, w_H(x, y)\big) \mid (x, y) \in \Delta_S \text{ or } (y, x) \in \Delta_S\big\}$ be the batch of distance decreases. 
The reduction algorithm then forwards the batch $B$ of distance decreases to the dynamic spanner algorithm~$\mathcal{SP}$ from~\cref{lem:fully_dyn_spanner},\footnote{The dynamic spanner algorithm~$\mathcal{SP}$ processes each distance decrease in the batch $B$ separately.} which runs on $H$ and maintains the edge-sparsified graph $\tilde{H}$.

\paragraph{Modifications to the assignment $\sigma$.}
Let $R_S$ be the set containing all vertices $s \in V$ for which the preimage $\sigma^{-1}(s)$ is modified. For every vertex $s \in R_S$ that already belongs to the dynamic $(\alpha, \beta)$-bicriteria approximate solution $S$, the reduction algorithm updates only its vertex weight $\wt(s)$ to $|\sigma^{-1}(s)|$. 

For every vertex $s \in R_S$ that does not belong to the dynamic $(\alpha, \beta)$-bicriteria approximate solution $S$, the incremental reduction algorithm proceeds as follows:
\begin{enumerate}
    \item First, the reduction algorithm adds $s$ to $S$ and to the vertex set of $H$, and sets its vertex weight~$\wt(s)$ to $|\sigma^{-1}(s)|$.
    
    \item Second, the reduction algorithm initializes an incremental $(1+\epsilon)$-approximate SSSP algorithm~$\mathcal{A}_s$ providing distance estimates $\delta_s(\cdot)$. Then for every vertex $x \in S$, the edge weights $w_H(s, x)$ and~$w_H(x, s)$ are set to $\min(\delta_s(x), \delta_x(s))$, rounded up to the nearest power of $(1+\epsilon)$.
\end{enumerate}

\noindent
Next, if at least one new vertex is added to $S$ (i.e., $|\sigma_\text{max}(V)|$ of~\cref{def:sigma_max} increases) then the reduction algorithm restarts the dynamic spanner algorithm~$\mathcal{SP}$ on the subgraph~$H$, producing a new edge-sparsified graph~$\tilde{H}$. Otherwise all vertices in $R_S$ already belong to $S$, and the dynamic spanner algorithm $\mathcal{SP}$ maintains the same edge-sparsified graph~$\tilde{H}$.

\vspace{1em}
At the end of each adversarial update, the incremental reduction algorithm runs
the static $(k, z)$-clustering algorithm $\mathcal{CL}$ of~\cref{lem:static_kz_clust} on the updated weighted $(k, z)$-clustering instance $(\tilde{H}, k, z, \wt)$, producing the maintained $(k, z)$-clustering solution $C$.

\begin{algorithm}\caption{Incremental Reduction Algorithm}
  \label{alg:bicr_to_kmedian}

\begin{algorithmic}[1]
\Require{Graph $G$ and positive integer $k$ (problem input); 
Bicriteria approximate solution $S$ and bicriteria approximate assignment $\sigma$}
\Ensure{$(k,z)$-Clustering solution $C$}
\vspace{0.5em}

\Function{Edge-Insertion}{$u, v$}
\State Insert $(u, v)$ into $G$
\vspace{0.2em}

\For{$s \in S$}
    \State $\mathcal{A}_s.\textit{insert}(u, v)$ 
\EndFor

\State $\Delta_S \gets \{(x, y) \in S \times S \mid \delta_x(y) < \frac{w_H(x, y)}{1+\epsilon} \}$ \label{algline:Delta_S}\Comment{Set of affected distances}

\vspace{0.2em}
\For{$(x, y) \in \Delta_S$}
    \State $w_H(x, y) \gets \min(\delta_x(y), \delta_y(x))$

    \State $w_H(y, x) \gets \min(\delta_x(y), \delta_y(x))$ 

    \State Round $w_H(y, x)$ and $w_H(x, y)$ up to the nearest power of $(1+\epsilon)$
\EndFor

\vspace{0.2em}

\State $B \gets \big\{\big(x, y, w_H(x, y)\big) \mid (x, y) \in \Delta_S \text{ or } (y, x) \in \Delta_S\big\}$ \label{algline:def_of_batch}

\State $\mathcal{SP}.\textit{update}(B)$
\vspace{0.2em}
\State $\tilde{H} \gets \mathcal{SP}.\textit{output}(H)$ \Comment{$\mathcal{SP}$ is applied on $H$}
\State $C \gets \mathcal{CL}.\textit{output}(\tilde{H})$ \Comment{$\mathcal{CL}$ is applied on $\tilde{H}$}
\EndFunction

\Statex
\Function{Assignment-Modifications}{$\hat{\sigma}$}

\State $R_S \gets \{s \in V \mid \sigma^{-1}(s) \neq \hat{\sigma}^{-1}(s)\}$ \label{algline:addto_Rs}
\State $\sigma \gets \hat{\sigma}$
\vspace{0.2em}

\For{$s \in R_S$}
    \State $\wt(s) \gets |\sigma^{-1}(s)|$ \label{algline:update_wt}
    \vspace{0.2em}
    
    \If{$s \notin S$}
        \State Add $s$ to $S$ and to the vertex set of $H$
        \State $\mathcal{A}_s.\textit{initialize}(G, s)$ \Comment{$\mathcal{A}_s$ provides distance estimates $\delta_s(\cdot)$}

        \vspace{0.2em}
        \For{$x \in S$}
            \State $w_H(s, x) \gets \min(\delta_s(x), \delta_x(s))$
        
            \State $w_H(x, s) \gets \min(\delta_s(x), \delta_x(s))$
    
            \State Round $w_H(s, x)$ and $w_H(x, s)$ up to the nearest power of $(1+\epsilon)$
        \EndFor
    \EndIf
\EndFor

\vspace{0.2em}
\If{$|\sigma_\text{max}(V)|$ increases} \Comment{Some new vertices have been added to $S$ (see~\cref{def:sigma_max})} \label{algline:if_sigma_max_inc}
    \State $\mathcal{SP}.\textit{initialize}(H)$ \label{algline:rest_spanner}
\EndIf
\vspace{0.2em}
\State $\tilde{H} \gets \mathcal{SP}.\textit{output}(H)$ \Comment{$\mathcal{SP}$ is applied on $H$}
\State $C \gets \mathcal{CL}.\textit{output}(\tilde{H})$ \Comment{$\mathcal{CL}$ is applied on $\tilde{H}$}
\EndFunction
    
\end{algorithmic}
\end{algorithm}
\paragraph{Implementation of~\linecref{algline:Delta_S}.}
We emphasize that the set $\Delta_S$ in~\linecref{algline:Delta_S} of~\cref{alg:bicr_to_kmedian} can be computed during the update procedure of the incremental $(1+\epsilon)$-approximate SSSP algorithm $\mathcal{A}_x$. This is achieved by extending~$\mathcal{A}_x$ to track the necessary quantities within its update procedure, without incurring additional asymptotic cost.

\subsection{Analysis of the Incremental Reduction Algorithm}
Our goal in this section is to prove Theorem~\ref{thm:bicr_to_kzclustering}, by analyzing the incremental reduction~\cref{alg:bicr_to_kmedian}.
We remind the reader that the reduction algorithm maintains weighted $(k, z)$-clustering instances 
$\mathcal{K} \coloneqq (H, k, z, \wt)$ and
$\mathcal{\tilde{K}} \coloneqq (\tilde{H}, k, z, \wt)$ (edge-sparsified) induced by the dynamic bicriteria approximate solution~$S = \sigma(V)$ in the incremental graph $G$. To distinguish the cost of a $(k, z)$-clustering solution $C$ in $\mathcal{K}$ from its cost in $G$, we use the notation $\cost^z_{\wt}(C, \mathcal{K})$ and $\cost^z(C, G)$ respectively. 

\subsubsection{Approximation Ratio}
Let $\rho_1 = O(1)$ denote the constant in the approximation ratio of the dynamic spanner algorithm $\mathcal{SP}$ from~\cref{lem:fully_dyn_spanner}, and $\rho_2 = O(1)$ denote the constant in the approximation ratio of the static $(k, z)$-clustering algorithm~$\mathcal{CL}$ from~\cref{lem:static_kz_clust}.

\begin{lemma} \label{lem:comp_dist_Htild_G}
    For every pair of vertices $x, y \in S$, it always holds that: \[
    \dist_G(x, y) \,\leq\, \dist_{\tilde{H}}(x, y) \,\leq\, (1+\epsilon)^2 \myhs \rho_1 \cdot \dist_G(x, y).\]
\end{lemma}
\begin{proof} 
    For any two vertices $x, y \in S$, the edge weight $w_H(x, y)$ in the subgraph $H$ is set and maintained as~$\min(\delta_x(y), \delta_y(x))$, rounded up to the nearest power of $(1+\epsilon)$. In turn, this implies that $w_H(x, y) \leq (1+\epsilon) \min(\delta_x(y), \delta_y(x))$. Based on~\cref{lem:incr_appr_sssp}, the edge weight $w_H(x, y)$ is upper bounded by $(1+\epsilon)^2 \dist_G(x, y)$.
    Since $\dist_H(x, y) \leq w_H(x, y)$, by~\cref{lem:fully_dyn_spanner} it follows that $\dist_{\tilde{H}}(x, y) \leq \rho_1 \cdot \dist_H(x, y) \leq \rho_1 \myhs (1+\epsilon)^2 \cdot \dist_G(x, y)$, as needed. 
    Furthermore, the incremental $(1+\epsilon)$-approximate SSSP algorithm of~\cref{lem:incr_appr_sssp} and the dynamic spanner algorithm of~\cref{lem:fully_dyn_spanner} do not underestimate distances.
    Since the edge weights in~$H$ are also rounded up, it follows that $\dist_G(x, y) \leq \dist_{\tilde{H}}(x, y) \leq (1+\epsilon)^2 \myhs \rho_1 \cdot \dist_G(x, y)$, as required.
\end{proof}

\begin{lemma} \label{lem:wt_of_H} 
    For every vertex $s \in S$, its vertex weight $\wt(s)$ in $H$ and $\tilde{H}$ is always equal to~$|\sigma^{-1}(s)|$.
\end{lemma}
\begin{proof}
    Whenever the preimage $\sigma^{-1}(s)$ of a vertex $s \in V$ is modified,
    the incremental reduction algorithm adds the vertex $s$ to the set $R_S$ in~\linecref{algline:addto_Rs} of~\cref{alg:bicr_to_kmedian}. In turn, its vertex weight $\wt(s)$ is updated to~$|\sigma^{-1}(s)|$ in~\linecref{algline:update_wt} of~\cref{alg:bicr_to_kmedian}.
\end{proof}

The following lemma is a variant of a theorem from~\cite{mettu2004optimal,guha2000clustering} usually stated for point sets in metric spaces; the same statement holds naturally for graphs, and we extend it to apply to our edge-sparsified graph $\tilde{H}$.
Essentially, we use it to argue that it suffices to solve the weighted $(k, z)$-clustering instance induced by the dynamic bicriteria approximate assignment $\sigma$. The proof of~\cref{lem:convert_bicr_kzclustering} follows from~\cref{lem:cost_G_with_OPTS,lem:upper_bound_OPT_tildeH}.

\begin{lemma}[see also~\cite{mettu2004optimal,guha2000clustering,bhattacharya2023fully}] \label{lem:convert_bicr_kzclustering} 
    Any $\rho$-approximate solution to the weighted $(k, z)$-clustering instance $\tilde{\mathcal{K}} = (\tilde{H}, k, z, \wt)$ is also an $O(\alpha \myhs \rho)$-approximate solution to the $(k, z)$-clustering instance~$(G, k, z)$. 
\end{lemma}
\begin{proof}
By~\cref{lem:cost_G_with_OPTS} and~\cref{lem:upper_bound_OPT_tildeH}, we have:
\begin{align*}
    \cost^z(C, G) &\;\leq\; 2^{z-1} \cdot \big(\alpha \myhs\myhs \OPT + \rho \myhs\myhs \OPT_{\tilde{\mathcal{K}}}\big),\;\text{and} \\
    \OPT_{\tilde{\mathcal{K}}} &\;\leq\; \big(4 \myhs (1+\epsilon)^2 \myhs \rho_1\big)^z \cdot \big(\alpha \myhs\myhs \OPT + \OPT\big).
\end{align*}
By combining these two, we deduce that:
\begin{align*}
    \cost^z(C, G) &\;\leq\; 2^{z-1} \cdot \Big(\alpha \myhs\myhs \OPT + \rho \, \big(4 \myhs (1+\epsilon)^2 \myhs \rho_1\big)^z \cdot \big(\alpha \myhs\myhs \OPT + \OPT\big)\Big) \\
    &\;\leq\; 2^z \cdot \Big(\alpha \cdot \big(\rho \, \big(4 \myhs (1+\epsilon)^2 \myhs \rho_1\big)^z + 1\big)\Big) \cdot \OPT \; \leq \; \alpha \myhs \rho \, \big(16 \myhs (1+\epsilon)^2 \myhs \rho_1\big)^z \cdot \OPT.
\end{align*}
Therefore as $\rho_1, z = O(1)$ and $\epsilon \in (0, \frac{1}{2})$, the $(k, z)$-clustering cost of $C$ in $(G, k, z)$ is upper bounded by~$O(\alpha \rho \cdot \OPT)$, as required.
\end{proof}

\begin{corollary} \label{cor:reduct_approx}
    After an adversarial update, the cost of the $(k, z)$-clustering solution $C$ is at most $O(\alpha \cdot \OPT)$.
\end{corollary}
\begin{proof} 
    The incremental reduction algorithm applies the dynamic spanner algorithm $\mathcal{SP}$ of~\cref{lem:fully_dyn_spanner} on the subgraph $H$ to obtain the subgraph $\tilde{H}$. The maintained $(k, z)$-clustering solution $C$ is the output of the static $(k, z)$-clustering algorithm~$\mathcal{CL}$ of~\cref{lem:static_kz_clust} on the weighted $(k, z)$-clustering instance $\tilde{\mathcal{K}} = (\tilde{H}, k, z, \wt)$.\footnote{Recall that based on the discussion in~\cref{appendix:vertex_weights}, we can safely assume that the static $(k, z)$-clustering algorithm of~\cref{lem:static_kz_clust} works on weighted $(k, z)$-clustering instances.} Based on~\cref{lem:static_kz_clust}, we have $\cost^z_{\wt}(C, \tilde{\mathcal{K}}) \leq \rho_2 \cdot \OPT_{\tilde{\mathcal{K}}}$, where $\rho_2 = O(1)$. Thus by~\cref{lem:convert_bicr_kzclustering}, it follows that $\cost^z(C, G) \leq O(\alpha \myhs \rho_2 \cdot \OPT) = O(\alpha \cdot \OPT)$, as needed.
\end{proof}

\begin{lemma} \label{lem:cost_G_with_OPTS}
    Consider a $\rho$-approximate solution $C$ to the weighted $(k, z)$-clustering instance $\tilde{\mathcal{K}} = (\tilde{H}, k, z, \wt)$. Then it holds that:
    \[
        \cost^z(C, G) \;\leq\; 2^{z-1} \cdot \big(\alpha \myhs\myhs \OPT + \rho \myhs\myhs \OPT_{\tilde{\mathcal{K}}}\big).
    \]
\end{lemma}
\begin{proof}
Let $C$ be a $\rho$-approximate solution to the weighted $(k, z)$-clustering instance $\tilde{\mathcal{K}} = (\tilde{H}, k, z, \wt)$.  Since the vertex set of $\tilde{H}$ is the bicriteria approximate solution~$S = \sigma(V)$, by~\cref{def:wgt_k_median} we have:
\[\cost_{\wt}^z(C, \tilde{\mathcal{K}}) \;=\; \sum_{s \in S} \wt(s) \cdot \dist_{\tilde{H}}(s, C)^z \;\leq\; \rho \cdot \OPT_{\tilde{\mathcal{K}}}. \]
The $(k, z)$-clustering cost of $C$ in the $(k, z)$-clustering instance $(G, k, z)$ (i.e., in the graph $G$) is evaluated as follows: 
\begin{align*}
    \cost^z(C, G) &\;=\; \sum_{v \in V} \dist_G(v, C)^z \;\leq\; \sum_{v\in V} \big(\dist_G(v, \sigma(v)) +\dist_G(\sigma(v),C)\big)^z \;\;\;\; \text{(triangle inequality)} \\
    &\;\leq\; 2^{z-1} \cdot \sum_{v\in V} \big(\dist_G(v, \sigma(v))^z +\dist_G(\sigma(v),C)^z\big) \\
    &\;=\;  2^{z-1} \cdot \Big(\cost^z(\sigma) + \sum_{s \in \sigma(V)} |\sigma^{-1}(s)| \cdot \dist_G(s,C)^z\Big) \;\;\;\; (\text{see~\cref{def:assignment}}) \\
    &\;=\;  2^{z-1} \cdot \Big(\cost^z(\sigma) + \sum_{s \in \sigma(V)} \wt(s) \cdot \dist_G(s,C)^z\Big) \;\;\;\; (\text{based on~\cref{lem:wt_of_H}}) \\
    &\;\leq\; 2^{z-1} \cdot \Big(\cost^z(\sigma) + \sum_{s \in \sigma(V)} \wt(s) \cdot \dist_{\tilde{H}}(s,C)^z\Big) \;\;\;\; (\text{based on~\cref{lem:comp_dist_Htild_G}})\\
    &\;=\; 2^{z-1} \cdot \big(\cost^z(\sigma) + \cost_{\wt}^z(C, \tilde{\mathcal{K}})\big) \;\leq\; 2^{z-1} \cdot \big(\alpha \myhs\myhs \OPT + \rho \myhs\myhs \OPT_{\tilde{\mathcal{K}}}\big) \;\;\;(\text{$\sigma$ is $(\alpha, \beta)$-bicriteria approximate}).
\end{align*}
\end{proof}

\begin{lemma} \label{lem:upper_bound_OPT_tildeH}
    For the optimal $(k, z)$-clustering cost of $\tilde{\mathcal{K}} = (\tilde{H}, k, z, \wt)$, it holds that:
    \[
        \OPT_{\tilde{\mathcal{K}}} \;\leq\; \big(4 \myhs (1+\epsilon)^2 \myhs \rho_1\big)^z \cdot \big(\alpha \myhs\myhs \OPT + \OPT\big).
    \]
\end{lemma}
\begin{proof}
Let $C^*$ be an optimal $(k, z)$-clustering solution to $(G, k, z)$,
let $\sigma^*: C^* \to \sigma(V)$ be the assignment that maps each center in $C^*$ to its closest vertex in $\sigma(V) = S$ with respect to $\dist_G(\cdot, \cdot)$, and let
$\pi: \sigma(V) \to C^*$ be the assignment that maps each vertex in $\sigma(V)$ to its closest center in $C^*$ with respect to $\dist_G(\cdot, \cdot)$. Consider the subset of vertices $Y \coloneqq \sigma^*(C^*) \subseteq \sigma(V)$. Then using~\cref{lem:comp_dist_Htild_G}, we obtain that:
\begin{align*}
    \cost_{\wt}^z(Y, \tilde{\mathcal{K}}) &\;=\; \sum_{s \in \sigma(V)} \wt(s) \cdot \dist_{\tilde{H}}(s, Y)^z \;\leq\; \big((1+\epsilon)^2 \myhs \rho_1\big)^z \cdot \sum_{s \in \sigma(V)} \wt(s) \cdot \dist_G(s, Y)^z \\
    &\;\leq\; \big((1+\epsilon)^2 \myhs \rho_1\big)^z \cdot \sum_{s \in \sigma(V)} \wt(s) \cdot \big(\dist_G(s, \pi(s)) + \dist_G(\pi(s), Y)\big)^z \\
    &\;\leq\; 2^{z-1} \myhs \big((1+\epsilon)^2 \myhs \rho_1\big)^z \cdot \sum_{s \in \sigma(V)} \wt(s) \cdot \big(\dist_G(s, \pi(s))^z + \dist_G(\pi(s), \sigma^*(\pi(s)))^z\big).
\end{align*}
By the construction of $\sigma^*$, we have $\dist_G(\pi(s), \sigma^*(\pi(s))) \leq \dist_G(\pi(s), s)$, and thus it follows that: 
\[
    \cost_{\wt}^z(Y, \tilde{\mathcal{K}}) \leq \big(2 \myhs (1+\epsilon)^2 \myhs \rho_1\big)^z \cdot \sum_{s \in \sigma(V)} \wt(s) \cdot \dist_G(s, \pi(s))^z.
\]
Since $\OPT_{\tilde{\mathcal{K}}} \leq \cost_{\wt}^z(Y, \tilde{\mathcal{K}})$,
it holds that $\OPT_{\tilde{\mathcal{K}}} \leq \big(2 \myhs (1+\epsilon)^2 \myhs \rho_1\big)^z \cdot \sum_{s \in \sigma(V)} \wt(s) \cdot \dist_G(s, \pi(s))^z$.
 An upper bound on $\sum_{s \in \sigma(V)} \wt(s) \cdot \dist_G(s, \pi(s))^z$ is given by:
\begin{align*}
    \sum_{s \in \sigma(V)} \wt(s) \cdot \dist_G(s, \pi(s))^z &\;=\; \sum_{s \in \sigma(V)} \wt(s) \cdot \dist_G(s, C^*)^z  \\
    &\;\leq\; \sum_{s \in \sigma(V)} \sum_{v \in \sigma^{-1}(s)} \big(\dist_G(\sigma(v), v) + \dist_G(v, C^*)\big)^z \;\;\; \text{(based on~\cref{lem:wt_of_H})} \\
    &\;\leq\; 2^{z-1} \cdot \sum_{v \in V} \big(\dist_G(\sigma(v), v)^z + \dist_G(v, C^*)^z\big) \\
    &\;\leq\; 2^{z-1} \cdot \big(\alpha \myhs\myhs \OPT + \OPT\big).
\end{align*}
Therefore, we conclude that $\OPT_{\tilde{\mathcal{K}}} \leq \big(2 \myhs (1+\epsilon)^2 \myhs \rho_1\big)^z \cdot 2^{z-1} \cdot \big(\alpha \myhs\myhs \OPT + \OPT\big) \leq \big(4 \myhs (1+\epsilon)^2 \myhs \rho_1\big)^z \cdot \big(\alpha \myhs\myhs \OPT + \OPT\big)$.
\end{proof}

\subsubsection{Total Update Time and Proof Completion of~\cref{thm:bicr_to_kzclustering}}
\begin{lemma} \label{lem:reduct_time}
    The total update time of the incremental reduction~\cref{alg:bicr_to_kmedian} is: 
    \[
        |\sigma_\text{max}(V)| \cdot m^{1+o(1)} \;+\; \tilde{O}\big(|\sigma_\text{max}(V)|^2 \cdot \sigma_\text{inc}(V) \,+\, m \cdot |\sigma_\text{max}(V)|^{1+\epsilon}\big).
    \]
\end{lemma}
\begin{proof}
    Regarding the update time of~\cref{alg:bicr_to_kmedian}, we need to provide an upper bound on the total time due to all the incremental $(1+\epsilon)$-approximate SSSP algorithms. In addition, we need to provide an upper bound on the total time due to the 
    dynamic spanner algorithm $\mathcal{SP}$ and the static $(k, z)$-clustering algorithm~$\mathcal{CL}$. 
 
    Based on~\cref{lem:incr_appr_sssp}, the total update time of a single incremental $(1+\epsilon)$-approximate SSSP algorithm is $m^{1+o(1)}$.\footref{ftnote:W_poly_n} Notice that the reduction Algorithm~\ref{alg:bicr_to_kmedian} initializes a new incremental $(1+\epsilon)$-approximate SSSP algorithm only when a new vertex is added to the $(\alpha, \beta)$-bicriteria approximate solution $S$.  Hence, the total update time due to all incremental $(1+\epsilon)$-approximate SSSP algorithms is upper bounded by~$|\sigma_\text{max}(V)| \cdot m^{1+o(1)}$. Regarding the total update time due to the dynamic spanner algorithm $\mathcal{SP}$, we consider
    separately the two types of adversarial updates:
    \begin{itemize}
        \item Consider an adversarial edge insertion into the graph $G$. For every pair of vertices $x, y \in S$ such that the distance estimate $\delta_x(y)$ decreases by a factor of at least $(1+\epsilon)$, the reduction Algorithm~\ref{alg:bicr_to_kmedian} adds the pair $(x, y)$ to the set~$\Delta_S$. In turn, a batch $B$ of distance decreases (see~\linecref{algline:def_of_batch} in~\cref{alg:bicr_to_kmedian}) is passed to the dynamic spanner algorithm $\mathcal{SP}$. Hence, the total number of distance decreases passed to $\mathcal{SP}$ is $\tilde{O}(|\sigma_\text{max}(V)|^2)$. Thus based on~\cref{lem:fully_dyn_spanner}, the total update time of $\mathcal{SP}$ under adversarial edge insertions is $\tilde{O}(|\sigma_\text{max}(V)|^2) \cdot \tilde{O}(1) = \tilde{O}(|\sigma_\text{max}(V)|^2)$.

        \item Consider a batch of modifications to the $(\alpha, \beta)$-bicriteria approximate assignment $\sigma$. The dynamic spanner algorithm $\mathcal{SP}$ is restarted whenever a batch of new vertices enters the bicriteria approximate solution~$S$ (see~\linescref{algline:if_sigma_max_inc}{algline:rest_spanner} of~\cref{alg:bicr_to_kmedian}), which occurs $\sigma_\text{inc}(V)$ times over all assignment modifications (see~\cref{def:batches_assign}). Since the subgraph $H$ 
        consists of $|S|^2$ edges, by~\cref{lem:fully_dyn_spanner} the total time incurred by $\mathcal{SP}$ under assignment modifications is at most $\tilde{O}(|S|^2 \cdot \sigma_\text{inc}(V))$.
    \end{itemize}

    Therefore, the total update time due to the dynamic spanner algorithm is $\tilde{O}(|S|^2 \cdot \sigma_\text{inc}(V))$; we have $S = \sigma_{\text{max}}(V)$, since no vertices are ever removed from the bicriteria approximate solution~$S$.
    After every adversarial update, the static $(k, z)$-clustering algorithm $\mathcal{CL}$ runs from scratch on the edge-sparsified graph~$\tilde{H}$. Since the subgraph $H$ consists of $|S|$ vertices, the edge-sparsified subgraph $\tilde{H}$ contains at most~$\tilde{O}(|S|^{1+\frac{1}{\lambda}})$ edges according to~\cref{lem:fully_dyn_spanner}, where~$\lambda = \frac{3}{\epsilon}$.
    Consequently, using~\cref{lem:static_kz_clust} with $\lambda = \frac{3}{\epsilon}$, the running time of $\mathcal{CL}$ per adversarial update is~$\tilde{O}(|S|^{1+\epsilon})$. The $\tilde{O}(|\sigma_\text{max}(V)|^{1+\epsilon})$ time of $\mathcal{CL}$ per adversarial update must be multiplied by the number of adversarial updates, which we can safely assume  is the number of adversarial edge insertions $m$, since any modification to $\sigma$ is triggered by an adversarial edge insertion.
    By combining all these arguments, the claim on the total update time of the incremental reduction~\cref{alg:bicr_to_kmedian} follows.
\end{proof}

\begin{proof}[Proof of \cref{thm:bicr_to_kzclustering}]
Consider the incremental reduction~\cref{alg:bicr_to_kmedian} and its analysis. The result then follows by~\cref{cor:reduct_approx,lem:reduct_time}.
\end{proof}

\subsection{Putting Everything Together}
In this section, we combine the tools we developed in order to obtain our result for the $(k,z)$-clustering problem on incremental graphs. The next lemma shows that the dynamic spanner algorithm $\mathcal{SP}$ is restarted only~$\tilde{O}(1)$ times in total.

\begin{lemma} \label{lem.}
    Consider the $(\alpha, \beta)$-bicriteria approximate assignment $\sigma$ which is the output of~\cref{alg:incremental} and the input of~\cref{alg:bicr_to_kmedian}. Then it holds that $\sigma_\text{inc}(V) = O(\log n \, \log_{1+\epsilon} nW)$.
\end{lemma}
\begin{proof}
    The dynamic assignment $\sigma$ in the incremental bicriteria approximation~\cref{alg:incremental} is modified in~\linescref{algline:v_to_sigma}{algline:last_level_assign}, and this occurs only when \textsc{First-Level-Decrease()} in~\linecref{algline:first_level_decr} returns a level~$i < t$. In such cases, there exists a level $i$ for which the corresponding $i$-th radius $\nu_i$ is decreased. As argued in the proofs of~\cref{lem:size_of_sol,lem:amort_update_time} (see also~\cref{obs:nu_only_decr_power_(1+eps),cor:number_of_levels}), this can happen at most~$O(\log n \cdot \log_{1+\epsilon} nW)$ times overall.

    The $(\alpha, \beta)$-bicriteria approximate assignment $\sigma$, which is the output of~\cref{alg:incremental}, is the input $(\alpha, \beta)$-bicriteria approximate assignment $\sigma$ in~\cref{alg:bicr_to_kmedian}. 
    Since $\sigma$ is modified at most $O(\log n \cdot \log_{1+\epsilon} nW)$ times, it follows by~\cref{def:batches_assign} that $\sigma_\text{inc}(V) = O(\log n \, \log_{1+\epsilon} nW)$, as needed.
\end{proof}

\mainthmkzclust*
\begin{proof}
    We choose $\epsilon$ to be $\frac{0.5}{\lambda}$.
    Based on~\cref{thm:incr_alg}, the total update time required to maintain a $(O(1), O(\log^3 n \log_{1+\epsilon} n W))$-bicriteria approximate assignment $\sigma$ under adversarial edge insertions is~$m^{1+o(1)}$. Based on~\cref{obs:S_extends_only} and the fact that the incremental reduction~\cref{alg:bicr_to_kmedian} does not remove vertices from the bicriteria approximate solution $S$, it follows that $|\sigma_\text{max}(V)| = |S| = O(k \log^3 n \, \log_{1+\epsilon}nW)$. 
    Also by~\cref{lem.}, we have~$\sigma_\text{inc}(V) = O(\log n \, \log_{1+\epsilon} nW)$.
    
    Therefore based on~\cref{thm:bicr_to_kzclustering}, the total update time to maintain the $O(1)$-approximate $(k,z)$-clustering solution is $\tilde{O}(k \myhs m^{1+o(1)} + k^2 + k^{1+\epsilon} \myhs m) = \tilde{O}(k \myhs m^{1+o(1)} + k^2 + k^{1+\frac{1}{2\lambda}} \myhs m)$. Notice that we can safely assume that $m \geq k$, since otherwise we can run a static $(k, z)$-clustering algorithm~\cite{JiangJLL25,latour2025fastersimplergreedyalgorithm} from scratch after each adversarial edge insertion (see~\cref{lem:static_kz_clust}). As a result, the stated bounds on the total and amortized update times follow.
\end{proof}

\paragraph{Disconnected graphs.}
Although our arguments are simpler when the input graph $G$ is connected, the same guarantees hold for disconnected graphs as well. In particular, the input graph $G$ may be initially empty (i.e., without any edges). In this paragraph, we explain in more detail how our incremental algorithm handles disconnected graphs. 

First, observe that in our incremental bicriteria approximation~\cref{alg:incremental} the $i$-th valid radius~$\hat{\nu}_i$ depends on the value of $\tilde{\nu}_i$ determined at~\linecref{algline:check_if_nu_decreases_1} or~\linecref{algline:check_if_nu_decreases_2} in~\cref{alg:incremental} (see also~\cref{def:ith_valid_rad}), where~$i \in [0, t)$ is a fixed level. When the graph $G$ is disconnected, it is possible that $\tilde{\nu}_i$ is undefined, since there may be no such ball of size at least $\beta |U_i|$ (even if $\tilde{\nu} = \infty$). In this case, based on~\cref{cor:tilde_nu_bound_incr} and~\cref{def:mu_star_incr}, the corresponding $\mu_i^*$ is also $\infty$.\footnote{Using the convention that the minimum of an empty set is $\infty$.} Since~\cref{lem:lower_bound_opt} allows us to lower bound the optimal $(k, z)$-clustering cost using the $\mu_i^*$ values after each adversarial edge insertion in the incremental setting, we conclude that the optimal $(k, z)$-clustering cost is infinite. Therefore, whenever the value of some $\tilde{\nu}_i$ is $\infty$, the incremental algorithm can safely report that $\OPT = \infty$.

Another subtle situation concerns the $i$-th supporting candidate set $\tilde{S}_i$, the $i$-th candidate set $S_i$, and the restart of the $i$-th incremental $(1+\epsilon)$-approximate SSSP algorithm~$\mathcal{A}_i$ at~\linecref{algline:sssp_update} in~\cref{alg:incremental}, where~$i \in [0, t)$ is a fixed level. As explained in~\cref{ftn:isolated_vert}, each isolated vertex must be a center, as otherwise the cost of the solution would be infinite. Consequently, each vertex can be given the opportunity to be sampled once it has an incident edge. To avoid restarting~$\mathcal{A}_i$ every time a new vertex is sampled in this ``lazy'' way, we simply extend the super-source attached to $S_i$ by adding the corresponding new zero-weight edges. In other words, we simply add a zero-weight edge from the super-source to each newly sampled candidate center.\footnote{During the first stage in which we use our incremental bicriteria approximation~\cref{alg:incremental}, the maintained solution is of size $\tilde{O}(k)$. Thus whenever $m < k$, an alternative approach is to extend the solution by adding the endpoints of each inserted edge; once~$m \geq k$, the algorithm begins.}

Finally, the static $(k, z)$-clustering algorithm in~\cite{latour2025fastersimplergreedyalgorithm} can be adapted to work with disconnected graphs. Another way to transform any static $(k, z)$-clustering algorithm for connected graphs to handle disconnected graphs is as follows. There must be at most $k$ connected components, since otherwise the cost of the solution is infinite. 
Hence, we can connect the $O(k)$ components with sufficiently large-weight edges (e.g., weight equal to $n \cdot (nW) + 1$).\footnote{During the second stage in which we use our incremental reduction~\cref{alg:bicr_to_kmedian}, the amortized update time is~$\tilde{O}(k^{1+\epsilon})$ for some $\epsilon \in (0, 1)$.} 
Specifically, we order the connected components and add an edge of sufficiently large weight between the $i$-th and the $(i+1)$-th connected component.
Then the new instance is a connected graph whose set of ($O(1)$-approximate) solutions remains unchanged.
As a result, our incremental algorithm can handle disconnected graphs as well.

    \subsection*{Acknowledgments}
We would like to thank Aditi Dudeja for taking the time to meet and engage in thoughtful conversations in Salzburg.

    \newpage
    \printbibliography[heading=bibintoc] 

    \newpage
    \appendix
    \section*{Appendix}
    \section{Missing Proofs}
\label{appendix:missingproofs}
In this section, we prove the lemmas whose proofs were omitted from the main body of the paper.

\relationnumu*
\begin{proof}[Proof sketch]
For a fixed level $i \in [0, t]$, let $X_i$ and $\mu_i^*$ be as in~\cref{def:mu_star}, and let $B^* \coloneqq \ball[X_i,\mu_i^*]\cap U_i$. We partition the set $B^*$ by assigning each vertex in $B^*$ to its closest point in $X_i$. Let $B_x$ denote the set of vertices in $B^*$ that are assigned to $x\in X_i$. We say that the $i$-th candidate set $S_i$ covers $B_x$ iff $B_x\cap S_i\neq\emptyset$, where $x \in X_i$. Similarly, we say that $y\in B^*$ is covered by $S_i$ iff $y\in B_x$ and $B_x \cap S_i \neq \emptyset$, where $x \in X_i$. A fixed set $B_x$ is covered by $S_i$ iff a vertex in $B_x$ is sampled in~\linecref{line:sample_Si} of~\cref{alg:static}. Let $B_\text{cov}\subseteq B^*$ denote the set of vertices covered by $S_i$. By construction, the following statements hold:
\begin{enumerate}
    \item By the definition of $\mu^*_{i}$, we have $\card{B^*} \geq \gamma \myhs |U_i|$.

    \item By the triangle inequality, any point $y \in B_\text{cov}$ is within distance $2 \myhs \mu_i^*$ from the $i$-th candidate set $S_i$.
    
    \item The factor $\beta$ appearing in the definition of $\nu_i^*$ (\cref{def:nu_star}) is less than the factor $\gamma$ appearing in the definition of $\mu_i^*$ (\cref{def:mu_star}).
\end{enumerate}
Since $\mathbb{E}[|S_i|] = \alpha \myhs k \log n$, by applying a Chernoff bound (\Cref{lem:chernoff}), we can infer that $|S_i| \geq \frac{\alpha}{2} k \log n$ with high probability. Based on the proof of Lemma 3.3 in~\cite{mettu2004optimal}, it follows that $|B_\text{cov}|\geq (1-\epsilon') \myhs \card{B^*}$ with high probability, for a suitably large choice of the constant $\alpha$ (see~\linecref{line:sample_Si} of~\cref{alg:static}) and any positive real $\epsilon'$. Therefore, it holds that $\lvert\ball[S_i, 2 \myhs \mu_i^*] \cap U_i\rvert$ is at least $|B_\text{cov}| \geq (1-\epsilon')\gamma \myhs |U_i| \geq \beta \myhs |U_i|$, which in turn implies that $\nu_i^* \leq 2 \myhs \mu_i^*$, and the claim follows.
\end{proof}

\nulesstwomuone*
\begin{proof}
    By an application of Chernoff bound (\cref{lem:chernoff}), it holds that $\card{Y} \geq \frac{\alpha}{2} k\log n$ with high probability.
    The proof is similar to that of \Cref{lem:nu<2mu}, with $S_i$ replaced by $Y$,~\cref{def:mu_star} replaced by~\cref{def:mu_star_incr}, and~$\nu_i^*$ replaced by~$\nu_i^*(Y)$.
\end{proof}

\subsection{Proof of \Cref{lem:lower_bound_opt}}

We first prove a series of claims, and then use them to establish~\Cref{lem:lower_bound_opt}. The analysis follows a similar structure to that presented in~\cite[Section 3]{mettu2004optimal}.

\lowerboundOPT*

Consider an arbitrary fixed subset of vertices $X \subseteq V$ of size at most $k$. For each level $i \in [0, t]$, we define $F_i \coloneqq \{v \in U_i \mid \dist(v, X) \geq \mu^*_i\}$, $F_i^\lambda \coloneqq F_i \setminus \bigcup_{j>0} F_{i+j\lambda}$, and $M_{i,\lambda} \coloneqq \{j\in [0, t]\mid j \equiv i\mod \lambda\}$. Recall that $t$ denotes the last level of the algorithm.

\begin{claim}
    Let $l, \lambda, i, j$ be integers such that
    $l \in [0, t]$, $\lambda > 0$, $i \in M_{l, \lambda}, j \in M_{l, \lambda}$, and $i \neq j$.
    Then it holds that~$F_i^\lambda \cap F_j^\lambda = \emptyset$.
\end{claim}
\begin{proof}
    Assume without loss of generality that $i<j$. By the definition of $F_i^\lambda$, we have $F_i^\lambda \subseteq F_i \setminus F_j$ as~$j=i + s \myhs \lambda$ for some $s>0$. 
    Since $F_{j}^\lambda \subseteq F_j$, it follows that $F_{i}^\lambda\cap F_{j}^\lambda=\emptyset$, as needed.
\end{proof}

We extend the definition of $\cost^z(\cdot)$ to pairs of vertex sets by defining $\cost^z(X, Y) \coloneqq \sum_{v \in Y} \dist(v, X)^z$, so that it gives the cost of $X$ evaluated on $Y$.

\begin{claim} \label{clm:rel_Y_Fi}
    For a fixed level $i \in [0, t]$, let $Y \subseteq F_i$ be a subset of $F_i$. Then it holds that $|F_i| \geq (1-\gamma) \myhs |U_i|$ and~$\cost^z(X,Y) \geq |Y| \myhs (\mu^*_{i})^z$.
\end{claim}
\begin{proof}
    The first part of the claim follows from the definition of $\mu^*_i$ (see~\cref{def:mu_star}).\footnote{We emphasize that in~\cref{def:mu_star}, the corresponding set $X_i$ must be a subset of $V$ and not just of $U_i$.} Since $Y \subseteq F_i$, the second part follows from the fact that:
    \begin{align*}
        \cost^z(X,Y) \,=\, \sum_{v \in Y} \dist(v, X)^z \,\geq\, \sum_{v \in Y} (\mu^*_{i})^z \,=\, |Y| \myhs (\mu^*_{i})^z.
    \end{align*}
\end{proof}

\begin{corollary}
    For all integers $l \in [0, t]$ and $\lambda > 0$, it holds that:
    \begin{align*}
        \cost^z\Big(X, \bigcup_{i\in M_{l,\lambda}}F^\lambda_i\Big) = \sum_{i\in M_{l,\lambda}} \cost^{z}(X,F^\lambda_{i})\geq \sum_{i\in M_{l,\lambda}}|F_{i}^\lambda|(\mu^*_{i})^z.
    \end{align*}
\end{corollary}

Recall that $r\coloneqq \ceil{\log_{(1-\beta)}((1-\gamma)/3)}$ and $0 < \beta < \gamma < 1$ are small constants. We show that for any  level $i$, the set $F_{i+r}$ has size at most one third of $F_i$, and thus $F_i^r$ has size at least half of $F_i$.

\begin{claim} \label{clm:Fi_shrink}
    For every level $i\in [0, t]$, it holds that $|F_{i+r}|\leq \frac{1}{3}|F_{i}|$.
\end{claim}
\begin{proof}
    By definition we have $F_{i+r} \subseteq U_{i+r}$. Based on~\cref{lem:exec_decr_beta,clm:rel_Y_Fi}, it holds that $|F_{i+r}|\leq |U_{i+r}|\leq (1-\beta)^r|U_i|\leq \frac{(1-\beta)^r}{1-\gamma}|F_{i}|$. The claim then follows by substituting the value of $r$.
\end{proof}
\begin{corollary}
    For every level $i\in [0, t]$, it holds that $|F_i^r|\geq \frac{|F_i|}{2}$.
\end{corollary}
\begin{proof}
    The size of $F_i^r$ is evaluated as follows:
    \begin{align*}
        |F_i^r| \;&=\; \Big|F_i \setminus \bigcup_{j>0} F_{i+jr}\Big|
        \;\geq \; |F_i| - \Big|\bigcup_{j>0} F_{i+jr}\Big| \\ \;&\geq\; |F_i| - \sum_{j > 0} \frac{|F_i|}{3^j} 
        \;\geq\; \frac{|F_i|}{2}.
    \end{align*}
\end{proof}

\begin{proof}[Proof of \Cref{lem:lower_bound_opt}]
Let $l = \arg\max_{0\leq j< r} \sum_{i\in M_{j,r}}|F_{i}^r| \myhs (\mu_i^*)^z$. The $(k, z)$-clustering cost of $X$ is then evaluated as follows:
\begin{align*}
    \cost^z(X) &\;\geq\; \cost^z\Big(X, \bigcup_{i \in M_{l, r}} F_i^r \Big)\\
    &\;\geq\; \sum_{i\in M_{l,r}}|F^{r}_i| \myhs (\mu_i^*)^z\\
    &\;\geq\; \frac{1}{r}\sum_{i\in [0, t]}|F_i^r| \myhs (\mu_i^*)^z\\
    &\;\geq\; \frac{1}{2r}\sum_{i\in [0, t]}|F_i| \myhs (\mu_i^*)^z\\
    &\;\geq\; \frac{1-\gamma}{2r}\sum_{i\in [0, t]} |U_i| \myhs (\mu_i^*)^z. 
    \qedhere
\end{align*}
\end{proof}

\section{Extending~\cref{lem:static_kz_clust} to vertex weights}
\label{appendix:vertex_weights}

The static $(k, z)$-clustering algorithm of Dupre la Tour and Saulpic~\cite{latour2025fastersimplergreedyalgorithm} defines a function $\text{Value}(\ball[x, r]) \coloneqq \lvert\ball[x, r]\rvert \cdot r^z$, for a vertex $x \in V$, a positive real number~$r$, and a constant $z \geq 1$. In fact, the authors show that an approximation of $\lvert\ball[x, r]\rvert \cdot r^z$ suffices. In particular, $\text{Value}(\ball[x, r])$ should satisfy the following inequalities (where $\lambda \geq 5$ is a constant):
\[
    \lvert\ball[x, r]\rvert \cdot \frac{r^z}{3} \;\leq\; \text{Value}(\ball[x, r]) \;\leq\; \lvert\ball[x, \lambda r]\rvert \cdot 3 \myhs r^z.
\]
To generalize the arguments to work with vertex weights $\wt(\cdot)$, the function $\text{Value}(\cdot)$ should be generalized to $\text{Value}_{\wt}(\ball[x, r]) \coloneqq \wt(\ball[x, r]) \cdot r^z$, where $\wt(\ball[x, r]) \coloneqq \sum_{v \in \ball[x, r]} \wt(v)$. Similarly, it suffices that $\text{Value}_{\wt}(\ball[x, r])$ satisfy the following inequalities (where $\lambda \geq 5$ is a constant):
\[
    \wt(\ball[x, r]) \cdot \frac{r^z}{3} \;\leq\; \text{Value}_{\wt}(\ball[x, r]) \;\leq\; \wt(\ball[x, \lambda \myhs r]) \cdot 3 \myhs r^z.
\]
Equipped with this definition of $\text{Value}_{\wt}(\ball[x, r])$, the same arguments presented in Appendix A of~\cite{latour2025fastersimplergreedyalgorithm} can be repeated for the weighted $(k, z)$-clustering problem.

Furthermore, the static $(k, z)$-clustering algorithm of Dupre la Tour and Saulpic~\cite{latour2025fastersimplergreedyalgorithm} employs the near-linear time algorithm of Cohen~\cite{Cohen97} in order to efficiently compute the values $\text{Value}_{\wt}(\ball[x, r])$. Since the algorithm in~\cite{Cohen97} works with vertex weights,~\cref{lem:static_kz_clust} can be extended to handle vertex weights as well.

\end{document}